\documentclass[a4paper,10pt,fleqn]{article}
\pdfoutput=1

\usepackage{jheppub}
\usepackage[T1]{fontenc}

\usepackage{setspace}
\usepackage{enumitem}
\usepackage{amsmath}
\usepackage{amsthm}
\usepackage{bm}
\usepackage{mathtools}
\usepackage{caption}
\usepackage{subcaption}
\usepackage{tikz}
\usepgflibrary{arrows}
\usetikzlibrary{calc}

\theoremstyle{plain}
\newtheorem{theorem}{Theorem}[section]
\newtheorem*{theorem*}{Theorem}
\newtheorem{corollary}[theorem]{Corollary}
\newtheorem{lemma}[theorem]{Lemma}
\newtheorem{proposition}{Proposition}[section]
\theoremstyle{definition}

\renewcommand{\L}{\mathbf{L}}
\renewcommand{\vec}[1]{\mathbf{#1}}
\newcommand{\ul}[1]{\underline{#1}}
\renewcommand{\Re}[0]{\operatorname{Re}}
\renewcommand{\Im}[0]{\operatorname{Im}}
\newcommand{\amu}[0]{a_\mu^{\text{HVP,LO}}}

\newcommand{\MartinStable}[0]{Luscher:1985dn} %
\newcommand{\TMR}[0]{Bernecker:2011gh} %
\newcommand{\ShortPaper}[0]{Hansen:2019rbh} %
\newcommand{\Cstar}[0]{Lucini:2015hfa} %
\newcommand{\GraphBook}[0]{nakanishi1971graph} %
\newcommand{\citeFermi}[0]{Carey:2009zzb,Grange:2015fou,Flay:2016vuw,Hong:2018kqx} %
\newcommand{\citeJPARC}[0]{Shimomura:2015aza,Sato:2017sdn} %
\newcommand{\whitepaper}[0]{whitepaper:2020} %
\newcommand{\discrepexp}[0]{Bennett:2004pv,Bennett:2006fi} %
\newcommand{\discrepth}[0]{Jegerlehner:2009ry,Jegerlehner:2017lbd,Davier:2017zfy,Keshavarzi:2018mgv} %
\newcommand{\latticeHVP}[0]{Blum:2002ii,Burger:2013jya,Chakraborty:2014mwa,Chakraborty:2015ugp,
Blum:2015you,Chakraborty:2016mwy,Blum:2016xpd,DellaMorte:2017dyu,
Giusti:2017jof,Borsanyi:2017zdw,Giusti:2018mdh,Giusti:2018vrc,
Blum:2018mom,Davies:2019efs,Gerardin:2019rua,Borsanyi:2020mff,Gulpers:2020pnz} %
\newcommand{\discrepvalref}[0]{Keshavarzi:2018mgv,Keshavarzi:2019abf} %
\newcommand{\hLbLandNLO}[0]{Blum:2016lnc,Giusti:2018vrc,Asmussen:2018ovy} %
\newcommand{\sizeofFVC}[0]{Blum:2018mom,\whitepaper,Borsanyi:2020mff} %
\newcommand{\TMRaint}[0]{\TMR} %
\newcommand{\powerlawME}[0]{Huang:1957im,Luscher:1986pf,Lellouch:2000pv,Meyer:2011um} %
\newcommand{\specFVcorr}[0]{Lellouch:2000pv,Meyer:2011um} %
\newcommand{\corrOBandBP}[0]{Burger:2013jya,Chakraborty:2014mwa,Chakraborty:2015ugp,
Blum:2015you,Chakraborty:2016mwy,Blum:2016xpd,DellaMorte:2017dyu,
Giusti:2017jof,Borsanyi:2017zdw,Giusti:2018mdh,Giusti:2018vrc,
Blum:2018mom,Davies:2019efs,Gerardin:2019rua,Borsanyi:2020mff} %
\newcommand{\monopole}[0]{Brommel:2006ww}
\newcommand{\convergence}[0]{DellaMorte:2017dyu}
\newcommand{\HVPdisp}[0]{Jegerlehner:2017lbd,Davier:2017zfy,Keshavarzi:2018mgv,Davier:2019can} %
\newcommand{\ward}[0]{Ward:1950xp}
\newcommand{\BD}[0]{Bjorken:1965zz}
\newcommand{\BMW}[0]{Borsanyi:2020mff}

\newcommand{\aunc}[0]{4} %
\newcommand{\hvpuncper}[0]{60} %
\newcommand{\hvpunc}[0]{2.5} %
\newcommand{\discrepval}[0]{27(7)} %

\title{Finite-volume and thermal effects in the leading-HVP contribution to muonic $(g-2)$}

\author[a]{M.~T.~Hansen}
\author[b]{and A.~Patella}

\affiliation[a]{Theoretical Physics Department, CERN, 1211 Geneva 23, Switzerland}
\affiliation[b]{Institut f\"ur Physik und IRIS Adlershof, Humboldt-Universit\"at zu Berlin, \\ \hspace{100pt} Zum Gro{\ss}en Windkanal 6, D-12489 Berlin, Germany}

\emailAdd{maxwell.hansen@cern.ch}
\emailAdd{agostino.patella@physik.hu-berlin.de}

\preprint{CERN-TH-2020-053, HU-EP-20/08}

\abstract{The leading finite-volume and thermal effects, arising in numerical lattice QCD calculations of $\amu \equiv (g- 2)^{\text{HVP,LO}}_\mu/2$, are determined to all orders with respect to the interactions of a generic, relativistic effective field theory of pions. In contrast to earlier work \cite{\ShortPaper} based in the finite-volume Hamiltonian, the results presented here are derived by formally summing all Feynman diagrams contributing to the Euclidean electromagnetic-current two-point function, with any number of internal pion loops and interaction vertices. As was already found in Ref.~\cite{\ShortPaper}, the leading finite-volume corrections to $\amu$ scale as $\exp[- m L]$ where $m$ is the pion mass and $L$ is the length of the three periodic spatial directions. In this work we additionally control the two sub-leading exponentials, scaling as $\exp[- \sqrt{2} m L]$ and $\exp[- \sqrt{3} m L]$.
As with the leading term, the coefficient of these is given by the forward Compton amplitude of the pion, meaning that all details of the effective theory drop out of the final result. 
 Thermal effects are additionally considered, and found to be sub-percent-level for typical lattice calculations. All finite-volume corrections are presented both for $\amu$ and for each time slice of the two-point function, with the latter expected to be particularly useful in correcting small to intermediate current separations, for which the series of exponentials exhibits good convergence.
}

\begin{document} 
\maketitle
\flushbottom

\abovedisplayskip 11pt
\belowdisplayskip 11pt

\section{Introduction}
\label{sec:intro}

The anomalous magnetic moment of the muon, $(g-2)_\mu$, has become a central focus in the broader particle physics community, due to significant tension between the best experimental \cite{\discrepexp} and theoretical determinations \cite{\discrepth}. The $\sim 3$ to $4$ sigma discrepancy represents a real opportunity to discover new physics beyond the Standard Model (BSM), especially given  the incredibly clean determinations on both the experimental and theoretical sides, together with the quadratic sensitivity to BSM effects, $(g-2)_\mu^{\text{BSM}} \sim (m_\mu / \Lambda_{\text{BSM}})^2$, where $m_\mu$ is the muon mass and $\Lambda_{\text{BSM}}$ the scale of the putative new physics. 

Experiments are underway at both Fermilab \cite{\citeFermi} and JPARC \cite{\citeJPARC}, with an update expected some time this year, and at least a factor of 2 uncertainty reduction targeted in the coming years.
The theory community is committed to maximizing the impact of this update, by providing a Standard Model determination of comparable overall uncertainty. The effort on this side is also very advanced and is summarized in detail in a forthcoming theory white paper \cite{\whitepaper}.

On the theoretical side, the leading uncertainties in $a_\mu \equiv (g-2)_\mu/2$, arise from hadronic contributions, generated via the coupling of QCD fields to the muon-photon vertex. These break into three categories: the leading hadronic-vacuum-polarization (HVP) contribution, $\amu$, the leading hadronic-light-by-light and the sub-leading corrections to the HVP. Outstanding progress has been made in the determination of the latter two contributions such that these are well in line to reach the overall target uncertainty \cite{\hLbLandNLO}. Since both the hadronic light-by-light and the sub-leading HVP are suppressed relative to the leading HVP, the targeted relative uncertainties here are $\sim 10\%$ 
 and, despite the complicated nature of these quantities, well within reach; see again ref.~\cite{\whitepaper}. In this work we restrict attention to the $\amu$ contribution, for which sub-percent uncertainty is required to reach the overall $ (g-2)_\mu$ precision target. 

Numerical lattice QCD (LQCD) provides an ideal tool in the determination of $\amu$, and many leading collaborations have already presented very advanced calculations \cite{\latticeHVP}. The observable can be directly extracted from a Euclidean electromagnetic-current two-point function and is thus well-suited to high-precision lattice determinations.
To make progress in practice, a deep theoretical understanding of all uncertainties is crucial, with the dominant sources being discretization effects, scale-setting uncertainty, statistical uncertainty (especially for large separations of the vector currents as well as those arising from quark-disconnected diagrams) and, finally, uncertainties arising from the effects of working in a finite-volume spacetime, the focus of this article.

The approximate value of $\amu$ is $700 \times 10^{-10}$. It is instructive to compare this to two other values. First, the overall theoretical uncertainty of $(g - 2)_\mu$ is approximately $\aunc \times 10^{-10}$ of which $\hvpuncper \%$ ($\hvpunc \times 10^{-10}$) is due to the uncertainty of $\amu$, with the present highest-precision determination arising from a data-driven dispersive approach \cite{\HVPdisp}. Second, the current discrepancy between the best theoretical and experimental determinations has been reported, by some groups, to be as high as $\sim \discrepval \times 10^{-10}$ \cite{\discrepvalref}. As we will discuss in more detail below, and as is also described in refs.~\cite{\sizeofFVC}, the finite-volume effect on $\amu$ for a standard spatial extent of $L \sim 4/m$, with $m$ the physical mass of the pion, is itself $\sim 20 \times 10^{-10}$. Thus, this effect alone is commensurate with the entire discrepancy, and it is of crucial importance that finite-volume corrections be treated reliably. 

A possible estimator for the leading-order HVP contribution to $(g-2)_\mu$ is
given by the integral \cite{\TMRaint}
\begin{gather}
   \label{eq:stat:amu-def}
   \amu(T,L) =  \int_0^{\frac{T}{2}} dx_0 \ \mathcal{K}(x_0) \, G(x_0|T,L)
   \ ,
\end{gather}
written in terms of the finite-volume, zero-momentum two-point function
of the Euclidean electromagnetic current
\begin{gather}
   \label{eq:stat:G-def}
   G(x_0|T,L)
   =
   - \frac{1}{3} \sum_{k=1}^3
   \int_0^L d^3 x \ 
   \langle j_k(x) j_k(0) \rangle_{T,L}
   \ ,
\end{gather}
and an analytically known function, $\mathcal{K}(x_0)$, usually
referred to as the \textit{kernel}. 
The kernel depends on the muon mass and the
electromagnetic fine-structure constant and its explicit expression is given in eq.~\eqref{eq:stat:K-def} below. 
The two-point function, $G(x_0|T,L)$,
can be determined numerically using lattice QCD calculations. In this work, we are not
concerned with lattice-spacing effects, and we assume that the continuum
limit has already been performed. We have used notation to emphasize that the
calculation is performed in a finite box with size $T \times L^3$. 

The
leading-order HVP contribution  to $(g-2)_\mu$ is then simply given by the
infinite-volume limit ($T,L \to \infty$) of the estimator $\amu(T,L)$. In
order to reliably quantify  the systematic error generated by the
extrapolation of lattice data to the infinite volume, it is necessary
to develop theoretical understanding of the finite-volume correction to $\amu$,
which is simply defined as the difference
\begin{gather}
   \Delta a(T,L) = \amu(T,L) - \lim_{L,T \to \infty} \amu(T,L) \ .
   \label{eq:stat:DeltaA}
\end{gather}

Naturally, finite-volume corrections depend on the choice of boundary conditions.
Here we consider a toroidal geometry with periodic boundary conditions for
gluons and periodic or antiperiodic boundary conditions for quarks. Our analysis
can be easily generalized to the case of phase-periodic boundary conditions for
quarks (as described also in ref.~\cite{\ShortPaper}), and perhaps less easily to the case of non-periodic geometry in the
temporal direction (e.g.~open boundary conditions).

\begin{figure}
\begin{center}
\includegraphics[width=\textwidth]{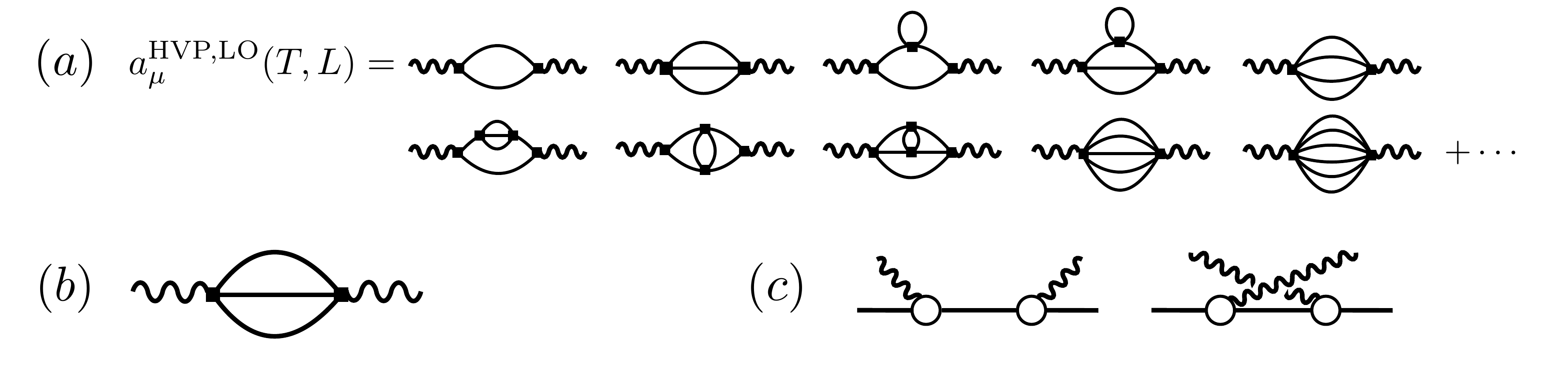}
\caption{(a) Examples of diagrams appearing in the all-orders expansion. (b) Sunset diagram that generates a finite-volume effect beyond the order we control.  (c) Single-pion exchange diagrams in the pion Compton scattering amplitude. Here the white circles denote the pion form factor, $F(k^2)$ sampled for space-like $k^2$. \label{fig:diags}}
\end{center}
\end{figure}

The results described in this paper are derived in an effective theory of pions
with generic local interactions in the isospin-symmetric limit, in the same
spirit as ref.~\cite{\MartinStable}, from which our analysis borrows the main
ingredients and ideas. In figure~\ref{fig:diags}(a) we schematically illustrate the infinite-series of Feynman diagrams contributing to $\amu(T,L)$. The interaction Lagrangian can be arbitrarily
complicated, and in order to make sense of the Feynman integrals, an ultraviolet
cutoff that preserves all relevant symmetries is assumed. However, the resulting
formulae and their proof are insensitive to these details. One could also
consider an effective theory that contains all stable hadrons. This would
introduce technical difficulties essentially related to the fact that one needs
to consider the constraints imposed by flavour conservation, as done for
instance in appendix A of ref.~\cite{\Cstar}. Including additional flavours would not, however, change the results
presented here, provided the quark masses are such that there is no other degree of freedom 
generating a finite-volume effect larger than our first neglected exponential. As explained in eq.~\eqref{eq:spec:Das} below, the leading neglected term falls (slightly) slower than $e^{- 2 m L}$ with increasing $L$, so that other states should be considered only if their mass is below $2m$. 
From a technical point of view, many similarities exist between this work and
the calculation of finite-volume effects to masses of stable hadrons, presented
in ref.~\cite{\MartinStable}. Here we focus on the differences. The observable
$\amu(T,L)$ is purely Euclidean, while in ref.~\cite{\MartinStable} one is interested
in the on-shell (therefore Minkowskian) hadron self-energy. This distinction makes
the present calculation technically simpler. On the other hand, this article includes an
analysis of finite-$T$ corrections and subleading finite-$L$
corrections, which have no counterpart in ref.~\cite{\MartinStable}.

The main results of this paper are summarized in the following points.
\begin{enumerate}
   \item Finite-volume corrections to $\amu(T,L)$, and to $G(x_0 \vert T, L)$ (for any fixed $x_0$), are exponentially suppressed and the
   exponential is controlled by the pion mass, $m$. The muon mass, entering $\amu(T,L)$ through the kernel, appears only in a polynomial function of $L$ and $T$ multiplying the exponentials. 
   \item Finite-volume corrections to $\amu(T,L)$ are given by the sum of finite-$L$ and
   finite-$T$ contributions, up to subleading terms. More precisely, the
   following asymptotic formula holds
   \begin{gather}
      \Delta a(T,L)
       =
      \Delta a(\infty,L)
      + \Delta a(T,\infty)
      + O\left(e^{-m\sqrt{L^2+T^2}}\right)
      \ , \label{eq:stat:separation} 
      \end{gather}
      with
      \begin{gather}
         \Delta a(\infty,L)   = O(e^{-mL}) \ , \ \ \ \ \ \ \ \ \ 
          \Delta a(T,\infty)   = O(e^{-mT}) \ ,
          \label{eq:stat:LTscaling}
   \end{gather}
   together with a similar decomposition for $G(x_0 \vert T, L)$, see eq.~\eqref{eq:stat:Gexp} below.
   \item The finite-$L$ corrections to both $\amu(\infty,L)$ and $G(x_0 \vert \infty, L)$ 
   can be written as an infinite series of
   exponentials of the type $e^{- m_\text{eff} L}$, with $m_\text{eff} \ge m$,
   multiplied by functions that grow at most polynomially in $L$. Intuitively,
   different exponentials are generated by hadron loops wrapping around the
   periodic spatial directions. A large fraction of this paper is devoted to
   making this intuitive picture precise. The finite-$L$ correction to the leading-order HVP can be
   decomposed as
   \begin{gather}
   \label{eq:spec:Das}
      \Delta a(\infty,L)
      =
      \Delta a_\text{s}(L) + O\left(e^{-\sqrt{2+\sqrt{3}}mL}\right)
      \ , 
   \end{gather}
   where $\Delta a_\text{s}(L)$ contains all diagrams with a single pion loop
   wrapping around the torus, and a generic number of topologically trivial
   loops. If the pion loop wraps $n_k$ times around the $k$-th spatial
   direction, the diagram contributes at order $e^{-|\vec{n}| m L}$. 
   
   \item Diagrams
   with at least two pion loops wrapping around the torus contribute to the
   neglected exponential in the above formula, where we observe that
   \begin{gather}
      \sqrt{2+\sqrt{3}} \simeq 1.93 \ .
   \end{gather}
   In fact, finite-$L$ contributions of order $e^{-\sqrt{2+\sqrt{3}}mL}$ can be
   found in the sunset diagram, figure~\ref{fig:diags}(b), by explicit calculation. At physical quark
   masses, heavier hadrons contribute with subleading exponentials. We obtain an
   explicit formula for $\Delta a_\text{s}(L)$, valid at the non-perturbative
   level, in terms of the Compton scattering amplitude of an off-shell spacelike
   photon against a pion in the forward limit. This formula is presented and
   discussed in section~\ref{sec:statement:analytical}, in particular in eqs.~\eqref{eq:stat:DeltaGs} and \eqref{eq:stat:DeltaIs}. 
   The result captures entirely contributions
   of order $e^{-mL}$, $e^{-\sqrt{2} mL}$ and $e^{-\sqrt{3} mL}$,
   corresponding to wrapping numbers with $\vec{n}^2=1,2,3$.
   \item The infinite series of $e^{-m_{\text{eff}} L}$-type exponentials (item 3), 
   in fact arises generically 
   whenever an observable has exponentially suppressed finite-volume effects. In such cases one must consider not only the form of the leading exponential, but also the convergence of the series. A striking example is given by $G(x_0 \vert \infty, L)$, evaluated at an arbitrarily large value of the separation coordinate:
   \begin{equation}
     G(x_0|\infty,L)
   =
    L^3    
\vert   \langle 0 \vert j_z(0) \vert E_0, L \rangle  \vert^2 e^{- E_0(L) x_0} +   O \left ( e^{- E_1(L) x_0} \right ) \,,
\label{eq:stat:specdecom}
   \end{equation}
   where $\vert E_k, L \rangle$ is the $k$th finite-volume, zero-momentum state, that overlaps $\langle 0 \vert j_z(0)$, and $E_k(L)$ is the corresponding energy. As has been shown in refs.~\cite{\powerlawME}, both $E_k(L)$ and $\vert   \langle 0 \vert j_z(0) \vert E_0, L \rangle  \vert^2$ receive power-like finite-volume corrections. Thus, for very large $x_0$, though the leading corrections scale as $e^{- m L}$, the convergence becomes arbitrarily poor, with the multi-exponential series reproducing the $1/L$-behavior.\footnote{Many of these observations are due to Harvey Meyer, private discussion. See also ref.~\cite{\convergence}.} The results presented in this work provide a framework to study the $x_0$-dependent convergence of the $e^{-m_{\text{eff}} L}$ series, and its implications for $\Delta a(T,L)$.
   \item In direct analogy to the finite-$L$ corrections, finite-$T$ corrections to both $\amu(T,\infty)$ and $G(x_0 \vert T, \infty)$ can be written as an infinite series of
   exponentials of the type $e^{- m_\text{eff} T}$ with $m_\text{eff} \ge m$,
   multiplied by functions that grow at most polynomially in $T$. In the case of $\amu(T,\infty)$,
   the analysis is complicated by the fact that $T$ appears also in the
   integration domain in eq.~\eqref{eq:stat:amu-def}. One finds %
   \begin{gather}
      \Delta a(T,\infty)
      =
      \Delta a_\text{t}(T) + O\left(e^{-\frac{3}{2}mT}\right)
      \ ,
      \label{eq:stat:Dat}
   \end{gather}
   where $\Delta a_\text{t}(T)$ is expressed in terms of 
   the infinite-volume two-point function
   as well as a function that is related to the forward Compton scattering amplitude
   of the pion by analytic continuation.
   This formula is presented and discussed
   in section~\ref{sec:statement:analytical}, in particular in eq.~\eqref{eq:stat:DeltaGt} and eqs.~\eqref{eq:stat:DeltaIt}-\eqref{eq:stat:DeltaIt-WP}.    %
\end{enumerate}

   Given a model or an experimental determination of the pion Compton amplitude, our result for $\Delta a_\text{s}(L)$ [eqs.~\eqref{eq:stat:DeltaGs} and \eqref{eq:stat:DeltaIs}] can be used to estimate
   finite-$L$ effects up to contributions of order $e^{-1.93 mL}$. 
   We discuss a possible modeling strategy in section \ref{sec:statement:finiteL} where we find that, for $mL=4$,
   finite-$L$ corrections lead to a $3\%$ reduction as compared to the infinite-volume value of $\amu$. 
   We additionally argue that
   finite-$L$ effects are dominated by the single-pion exchange in the Compton
   scattering amplitude, see figure~\ref{fig:diags}(c), which is completely described by the pion
   electromagnetic form factor in the space-like region. We see explicitly that
   the subleading exponential $e^{-\sqrt{2}mL}$ contributes as much as the leading $e^{-mL}$ when $mL=4$, likely due to the poor convergence of the finite-volume expansion for $G(x_0 \vert \infty, L)$ at large $x_0$. However, the
   relative contribution of higher exponentials
    drops rapidly, so that our strategy gives an infinite-volume estimate of $\amu$, with residual volume effects at the few per mille level.  %
   We also find that approximating the pion as a
   pointlike particle underestimates the finite-$L$ effects by about $30\%$ if
   $mL=4$.

   In realistic numerical calculations, the long-distance behaviour of the
   two-point function is dominated by noise. Different
   estimators for $\amu$ are then used, which often involve restricting the
   integral in eq.~\eqref{eq:stat:amu-def} to the interval $0 \le x_0 \le
   x_0^\text{cut}$, and reconstructing the long-distance behaviour of the
   two-point function with alternative techniques. 
   A leading approach is to independently determine the finite-volume energies and matrix elements, and use these to reconstruct the large $x_0$ behavior, along the lines of eq.~\eqref{eq:stat:specdecom} \cite{\TMR}. 
   Though this step can, in principle, be performed independently of the finite-volume correction, many groups advocate using the well-established formalism governing matrix elements and energies, to simultaneously estimate a finite-$L$ correction for the large $x_0$ region \cite{\specFVcorr}.

One can investigate the robustness of such procedures by using our expression for the finite-$L$ corrections to $G(x_0|T,L)$ at fixed $x_0$, eq.~\eqref{eq:stat:DeltaGs}. The expansion studied here and the spectral-decomposition-based method are complementary, as the series in $e^{- m_{\text{eff}}L}$ converges best for small $x_0$ while the sum over finite-volume states requires that the time coordinate is taken large. In fact, our full estimate for $\Delta a_\text{s}(L)$ shows remarkable consistency with the spectral method for $mL=4$ \cite{\sizeofFVC}, indicating that the relevant time slices are well described by both methods for this volume. As $mL$ is increased in future calculations, the $e^{- m_{\text{eff}}L}$ expansion will improve, whereas the spectral method will struggle with the need to extract more finite-volume matrix elements and energies. 

In a much shorter paper~\cite{Hansen:2019rbh} we have calculated the leading
$e^{-mL}$ contribution by means of completely different techniques, making use of the finite-volume Hamiltonian. The
Hamiltonian formalism and the spectral decomposition, with quantization along a spatial direction, allow one to bypass the
diagrammatic expansion, making the derivation dramatically more compact. One can
easily use the Hamiltonian formalism also to calculate the leading $e^{-mT}$
contribution. Unfortunately there is no obvious way to use the approach to prove the decomposition, eq.~\eqref{eq:stat:separation}, nor to obtain
the subleading exponentials $e^{-\sqrt{2}mL}$ and $e^{-\sqrt{3}mL}$. This
leads us to the lengthy derivation presented in this article.

The remainder of the manuscript is organized as as follows. In section~\ref{sec:statement} we summarize
and describe our formulae for finite-$L$ and finite-$T$ corrections to $G(x_0|T,L)$ and $\amu(T,L)$,
we discuss in detail strategies to evaluate them numerically, and we provide
estimates. The rest of the paper is devoted to the derivation of these formulae.
In section~\ref{sec:diag} we analyze the diagrammatic expansion of the integral
in eq.~\eqref{eq:stat:amu-def}, we separate the finite-volume corrections,
recognize that the infinite-volume limit is of the saddle-point type, provide
upper bounds for the asymptotic behaviour of each Feynman diagram, identify the
contributions to the leading and first subleading exponentials, and sum
these contributions in terms of proper vertices and dressed propagators. While
the exponential decay of the finite-volume effects is manifest in coordinate
space, this is not the case in momentum space. In section~\ref{sec:calculation}
we deform the contours of the momentum integrals in the formulae obtained in the
preceding section to make the exponential decay manifest. This procedure leads to
our final formulae, written in terms of the forward Compton scattering amplitude
of the pion, and a specific analytic continuation of the same. Several theorems need to be proven on the way. Most of them are only
stated in the main text, with the proof relegated to
various appendices.

\section{Statement of the results}
\label{sec:statement}

\subsection{Analytical formulae}
\label{sec:statement:analytical}

If $x_0$ is kept constant while $T$ and $L$ are sent to infinity, the
two-point function has the following expansion
\begin{align}
   \Delta G(x_0|T,L)
   & = \Delta G(x_0|\infty,L) + \Delta G(x_0|T,\infty) 
   + O\left(e^{-m\sqrt{L^2+T^2}}\right)
   \label{eq:stat:Gexp}
   \,, 
\end{align}
where
\begin{align}
 \Delta  G(x_0|T,L)
   & =
     G(x_0|T,L) - G(x_0|\infty) \,,
     \label{eq:stat:DGdef}
\end{align}
with the shorthand $G(x_0 \vert \infty) = \lim_{T,L \to \infty}   G(x_0|T,L)$.

The finite-$L$ correction to the two-point function is given by
\begin{align}
& \Delta G(x_0|\infty,L)   =  \Delta G_\text{s}(x_0|L) + O\left(e^{-\sqrt{2+\sqrt{3}}mL}\right)  \,,  \label{eq:stat:DGsexp} \\[3pt]
 & \Delta G_\text{s}(x_0|L)
   =
   - \sum_{ \vec{n} \neq \vec{0} }
   \int \frac{d p_3}{2\pi}
   \frac{
   e^{-|\vec{n}| L \sqrt{m^2+p_3^2}}
   }{
   24 \pi |\vec{n}| L
   }
   \int \frac{d k_3}{2\pi} \cos (k_3x_0)
   \Re T(-k_3^2,-p_3k_3)
   \label{eq:stat:DeltaGs}
   \ ,
\end{align}
where $\vec{n}$ takes value in $\mathbb{Z}^3$, and $T(k^2, p \cdot k)$ is the forward Compton
scattering amplitude of the pion, summed over the Lorentz indices of the
off-shell photon and over the charge of the pion \begin{gather}
   T(k_\mu k^\mu, p_\mu k^\mu)
   =
   i \lim_{\vec{p}' \to \vec{p}} \sum_{q=0,\pm 1} \int d^4 x \ e^{ik_\mu x^\mu}  \langle \vec{p}', q | \mathrm{T} \mathcal{J}_\rho(x) \mathcal{J}^\rho(0) | \vec{p}, q \rangle
   \ .
   \label{eq:stat:Compton}
\end{gather}
In this formula, $\mathcal{J}_\rho(x)$ is the Minkowskian electromagnetic
current in the Heisenberg picture, $| \vec{p}, q \rangle$ is the state of a pion
with momentum $\vec{p}$ and charge $q$, and $p_0 =
\sqrt{m^2+\vec{p}^2}$, normalized using the standard relativistic convention. Note that the $\textbf n$th term in
eq.~\eqref{eq:stat:DeltaGs} is of order $e^{-|\vec{n}| m  L}$. Therefore our
analysis captures the exponentials $e^{-mL}$, $e^{-\sqrt{2} mL}$ and
$e^{-\sqrt{3} mL}$, corresponding to the terms with $\vec{n}^2=1,2,3$, while the
terms with $\vec{n}^2 \ge 4$ are at least of order $e^{-2mL}$ and are hence
subleading with respect to the $e^{-\sqrt{2+\sqrt{3}}mL}$-term, already
neglected in eq.~\eqref{eq:stat:Gexp}.

The finite-$T$ correction to the two-point function is given by
\begin{align}
 \Delta G(x_0|T, \infty) & =  \Delta G_\text{t}(x_0|T) + O\left(e^{-2mT}\right)  \,,
 \label{eq:stat:DGtexp}
 \end{align}
 where
 \begin{align}
  & \Delta G_\text{t}(x_0|T)
    = - \frac{1}{3} \int \! \frac{d^3 p}{(2\pi)^3} \frac{ e^{-T \sqrt{m^2+\vec{p}^2}} }{ 2 \sqrt{m^2+\vec{p}^2} } \hat{S}(x_0,\vec{p}^2)
   \label{eq:stat:DeltaGt}
   \ , \\
 &  \hat{S}(x_0,\vec{p}^2)
   =
   \lim_{\vec{p}' \to \vec{p}} \sum_{\rho=1}^3 \sum_{q=0,\pm 1}   \int d^3x \, \langle \vec{p}',q | \text{T} j_\rho(x) j_\rho(0) | \vec{p},q \rangle
   \ ,
   \label{eq:stat:Sdef}
\end{align}
and where $j_\rho(x)$ is the Euclidean electromagnetic current in the Heisenberg
picture. 
We will see that the Fourier transform of $\hat{S}(x_0,\vec{p}^2)$ with
respect to $x_0$ is related to the forward Compton scattering amplitude by
analytic continuation. It is evident that $\Delta G_\text{t}(x_0|T)$, for any fixed $x_0$, is of order
$e^{-mT}$.

As discussed in the Introduction, analogs of eqs.~\eqref{eq:stat:Gexp}, \eqref{eq:stat:DGsexp} and \eqref{eq:stat:DGtexp} hold for $\Delta a(T,L)$ as well. [See eqs.~\eqref{eq:stat:separation}, \eqref{eq:spec:Das} and \eqref{eq:stat:Dat}.]
Note that, for $\Delta a(T,L)$, a sub-leading correction of $e^{-\frac{3}{2}mT}$ appears in \eqref{eq:stat:DGtexp}, in place of the
corresponding exponential, $e^{-2mT}$, in eq.~\eqref{eq:stat:DGtexp}. This
is a consequence of the fact that values of $x_0$ that scale proportionally to
$T$ contribute to the integral defining $\amu(T,L)$.

 It turns out
that the finite-$L$ corrections to the integral, $\amu(T,L)$, are given by the
finite-volume corrections of $G(x_0|L,T)$ integrated with the kernel \begin{gather}
   \label{eq:stat:DeltaIs}
   \Delta a_\text{s}(L) = \int_0^\infty dx_0 \  \mathcal{K}(x_0) \, \Delta G_\text{s}(x_0|L)
   \ ,
\end{gather}
with
\begin{gather}
   \label{eq:stat:K-def}
   \mathcal{K}(x_0)
   =
   \frac{2 \alpha^2}{ m_\mu^2} \left[
   t^2 -2 \pi  t + 8 \gamma_E - 2 +\frac{4}{t^2}+8 \log t
   - \frac{8 K_1(2 t)}{t}
   - 8 \int_0^\infty dv \frac{e^{- t \sqrt{ v^2+4} }}{ (v^2+4 )^{3/2}}
   \right]_{t = m_\mu x_0}
   \ ,
\end{gather}
where $m_\mu$ is the muon mass, $\alpha$ is the fine-structure constant and
$K_1(x)$ is a Bessel function. This result is not obvious and relies on the polynomial scaling of $\mathcal K(x_0)$ for large $x_0$.

The finite-$T$ corrections are given by three
different pieces
\begin{gather}
   \Delta a_\text{t}(T) = \Delta a_\text{t}^\text{OB}(T) + \Delta a_\text{t}^\text{BP}(T) + \Delta a_\text{t}^\text{WP}(T)
   \label{eq:stat:DeltaIt}
   \ ,
\end{gather}
where we have introduced the out-of-the-box (OB) contribution
\begin{gather}
   \Delta a_\text{t}^\text{OB}(T)
   =
   - \int_{\frac{T}{2}}^\infty dx_0 \ \mathcal{K}(x_0) G(x_0|\infty)
   \ ,
   \label{eq:stat:DeltaIt-OB}
\end{gather}
the backpropagating-pion (BP) contribution
\begin{gather}
   \Delta a_\text{t}^\text{BP}(T)
   =
   \int_0^{\frac{T}{2}} dx_0 \ \mathcal{K}(x_0) G(T-x_0|\infty)
   \label{eq:stat:DeltaIt-BP}
   \ ,
\end{gather}
and the wrapped-pion (WP) contribution
\begin{gather}
   \Delta a_\text{t}^\text{WP}(T)
   =
   \int_0^{\frac{T}{2}} dx_0 \ \mathcal{K}(x_0) \, \Delta G_\text{t}(x_0|T)
   \ .
   \label{eq:stat:DeltaIt-WP}
\end{gather}
Each of these terms is of order $e^{-mT}$. It follows that the finite-$T$
corrections can be neglected in eq.~\eqref{eq:stat:separation} in the commonly-used
setup $T=2L$.

It is interesting to compare our calculation to others available in
the literature. The formulae for the finite-volume corrections obtained in this
paper are valid in any theory under very general assumptions:~particles with
mass $m$ and with the quantum numbers of the pions exist, and no particle or
bound state exists with energy below $2m$ (an assumption that can be easily
relaxed). Besides full QCD, these formulae apply to two other relevant
frameworks,  for non-interacting pions and in chiral perturbation theory ($\chi$PT).
In the case of the free particle, finite-$L$ corrections have been calculated
in ref.~\cite{DellaMorte:2017dyu} (Appendix C). We note that our
neglected terms, $O(e^{-\sqrt{2+\sqrt{3}}mL})$, come only from contributions to
the two-point function with more than one pion loop,\footnote{We
stress that higher-loop diagrams contribute also to the leading exponentials.} and are
therefore absent in the free theory. If one uses the free-pion
Compton scattering amplitude, our eq.~\eqref{eq:stat:DeltaGs} describes the
free-particle finite-$L$ corrections entirely and coincides with eq.~(C.3)
in ref.~\cite{DellaMorte:2017dyu}. We have also checked this correspondence numerically.  

In the case of $\chi$PT, the finite-$L$ corrections to the two-point
function have been calculated to NLO in ref.~\cite{Aubin:2015rzx} and to N$^2$LO in refs.~\cite{Bijnens:2017esv,Aubin:2019usy}. If one uses the NLO
Compton scattering amplitude in our eq.~\eqref{eq:stat:DeltaGs}, one should
obtain the result of refs.~\cite{Bijnens:2017esv,Aubin:2019usy}, up to $O(e^{-\sqrt{2+\sqrt{3}}mL})$
terms. In fact, one could use the difference between the two approaches to
estimate those terms. However we point out that $\chi$PT in this context should
be used with care, since the integral~\eqref{eq:stat:DeltaIs} includes small
values of $x_0$, for which $\chi$PT breaks down.  At sufficiently high order, the
small-$x_0$ suppression provided by the kernel is not enough to balance the
increasing singularity of the two-point function, and the integral becomes
divergent.\footnote{For instance, it is straightforward to check that the
Compton scattering amplitude $\Re T(-k_3^2,-p_3k_3)$ at N$^3$LO contains terms
given by $k_3^6$ times a low-energy constant. Direct calculation shows the
contribution of this term to $\Delta a_\text{s}(L)$ is divergent.} Our formula does
not suffer of this problem since the behaviour of the fully non-perturbative
Compton scattering amplitude is such that the integral~\eqref{eq:stat:DeltaIs}
is finite.

\subsection{Estimate of finite-\texorpdfstring{$L$}{L} corrections}
\label{sec:statement:finiteL}

\subsubsection{Preliminaries}

Up to terms of order $e^{-\sqrt{2+\sqrt{3}}mL}$, the finite-$L$ corrections to
$\amu(T,L)$, encoded by $\Delta a_{\text{s}}(L)$, can be rewritten in a form that is more suitable for
numerical evaluation, assuming that an estimate of the forward Compton
scattering amplitude is available.

As detailed in appendix~\ref{sec:CompDecom}, general principles allow one to represent the forward Compton scattering amplitude
of a pion against a space-like photon as
\begin{align}
   T(-k_3^2,-p_3k_3) & =
   T_\text{pole}(-k_3^2,-p_3k_3)
   + T_\text{reg}(-k_3^2,-p_3k_3) \ ,
   \nonumber \\  & = %
   \frac{2(4m^2 + k_3^2) \, F(-k_3^2)^2}{ k_3^2 + 2k_3p_3 - i\epsilon }
   + \frac{2(4m^2 + k_3^2) \, F(-k_3^2)^2}{ k_3^2 - 2k_3p_3 - i\epsilon }
   + T_\text{reg}(-k_3^2,-p_3k_3)
   \label{eq:estL:Compton}
   \ ,
\end{align}
where $F(k_\mu k^\mu)$ is the electromagnetic form factor of the pion, i.e.
\begin{gather}
   \langle \vec{p}, q | \mathcal{J}_\rho(0) | \vec{p}',q \rangle = q \, (p+p')_\rho \ F \big( (p'-p)_\mu (p'-p)^\mu \big)
   \ ;
   \label{eq:pionFF}
\end{gather}
see, again, figure~\ref{fig:diags}(c).
The pole due to the single-pion exchange is made manifest in
eq.~\eqref{eq:estL:Compton} in a Lorentz-invariant fashion. The reminder,
$T_\text{reg}$, can be seen to be analytic in $p_3$ and $k_3$ as long as $p_3^2
< 3m^2$. The behaviour of the Compton scattering amplitude around $k_3=0$ is
dictated by the Ward identities,\begin{gather}
   F(0) = 1 \ , \qquad T_\text{reg}(0,0) = 8 \ .
\end{gather}
The behaviour as $|k_3| \to \infty$ is dictated by operator product
expansion, which yields $T(-k_3^2,-p_3k_3) \propto k_3^{-2}$ up to logarithms.
On the other hand, one has for the form factor~\cite{Farrar:1979aw}
\begin{gather}
   F(-k_3^2) \propto k_3^{-2}  \qquad \text{for} \qquad \vert k_3 \vert \to \infty \ ,
\end{gather}
which, together with eq.~\eqref{eq:estL:Compton}, also implies
implies that
\begin{gather}
   T_\text{reg}(-k_3^2,-p_3k_3) \propto k_3^{-2}  \qquad \text{for} \qquad \vert k_3 \vert \to \infty  \ ,
\end{gather}
always up to logarithms.

\subsubsection{$\Delta a_s(L)$}

By substituting the above decomposition of the Compton scattering amplitude into
eq.~\eqref{eq:stat:DeltaGs} and then into eq.~\eqref{eq:stat:DeltaIs}, and by
using the following representation for the kernel~\cite{DellaMorte:2017dyu},
\begin{gather}
   \mathcal{K}(x_0)
   =
   \frac{8 \alpha^2}{ m_\mu^2}
   \int_0^\infty \frac{d\omega}{\omega^2} \ g(\omega)
   \left[ \omega^2t^2 - 4 \sin^2 \left( \frac{\omega t}{2} \right) \right]_{t=m_\mu x_0}
   \ , \\
   g(\omega) = \frac{16}{\sqrt{\omega^2+4} \left( \sqrt{\omega^2+4} + \sqrt{\omega^2} \right)^4}
   \label{eq:estL:g}
   \ ,
\end{gather}
one obtains, after some lengthy but straightforward algebra,
\begin{gather}
   \label{eq:estL:first}
   \Delta a_\text{s}(L)
   =
   \sum_{ \vec{n} \neq \vec{0} }
   \frac{\alpha^2}{12 \pi n L}
   \left\{
   \mathcal{T}'( 0 | nL )
   - \frac{4}{m_\mu} \int_{0}^\infty dk_3
   \frac{ \mathcal{T}(k_3^2 | nL ) - \mathcal{T}(0 | nL ) }{ k_3^2 }
   g\left( \frac{k_3}{m_\mu} \right)
   \right\}
   \ ,
\end{gather}
where $\mathcal{T}'( k_3^2 | nL ) = \partial_{k_3^2}  \mathcal{T}(k_3^2|nL) $, with the definitions
\begin{gather}
   \label{eq:estL:Tcal}
   \mathcal{T}(k_3^2|nL)
   =
   \mathcal{T}_\text{pole}(k_3^2|nL) + \mathcal{T}_\text{reg}(k_3^2|nL)
   \ , \\[5pt]
   \label{eq:estL:Tcal-pole}
   \mathcal{T}_\text{pole}(k_3^2|nL)
   =
   2(4m^2 + k_3^2) \, F(-k_3^2)^2 \, \zeta(k_3^2|nL)
   \ , \\
   \label{eq:estL:Tcal-reg}
   \mathcal{T}_\text{reg}(k_3^2 | nL )
   =
   \int%
   \frac{d p_3}{2\pi}
   e^{-nL \sqrt{m^2+p_3^2}}
   \Re T_\text{reg}(-k_3^2,-p_3k_3)
   \,,
\end{gather}
discussed in detail in appendix~\ref{sec:CompDecom}.
The function $ \zeta(k_3^2|nL)$ can be represented in a number of ways, e.g.
\begin{gather}
   \label{eq:estL:zeta}
 \hspace{-10pt}  \zeta(k_3^2|nL)
   =
   \int \frac{d p_3}{2\pi}
   \left[
   \frac{ e^{-nL Z_+} }{ Z_+ }
   \frac{
   \sinh \left( nL Z_- \right)
   }{
   2 Z_-
   }
   \right]_{
   Z_\pm = \sqrt{ \frac{1}{2} \left(m^2 + p_3^2 + \frac{k_3^2}{4}\right) \pm \frac{1}{2} \sqrt{ \left(m^2 + p_3^2 + \frac{k_3^2}{4}\right)^2 - p_3^2k_3^2 } }
   }
   \ ,
\end{gather}
which we find numerically stable [see also eq.~\eqref{eq:zetadef}]. Note that this representation allows one to take
the $\epsilon \to 0^+$ limit in eq.~\eqref{eq:estL:Compton} analytically.
A more compact formula can also be obtained integrating by parts twice
\begin{gather}
   \label{eq:estL:second}
   \Delta a_\text{s}(L)
   =
   - \sum_{ \vec{n} \neq \vec{0} }
   \frac{\alpha^2 m_\mu}{3 \pi n L}
   \int_{0}^\infty dk_3 \,
   \hat{g}\left( \frac{k_3}{m_\mu} \right) \,
   \mathcal{T}''(k_3^2 | nL )
   \ ,
\end{gather}
where the following auxiliary function has been introduced
\begin{align}
   \label{eq:estL:ghat}
   \hat{g}(\omega) & = \omega \int_{\omega^2}^\infty dy \int_{y}^\infty dz \frac{g(z^{1/2})}{ z^{3/2} } \,, \\
  &  =
\frac{\omega}{4}  \left[      \omega^3 \sqrt{\omega^2+4}-\omega^4 +2 \omega^2  -4 \omega \sqrt{\omega^2+4}   -8 \omega^2 \log \left(\frac{2 \omega}{\sqrt{\omega^2+4}+\omega}\right)+2\right]
   \ .
\end{align}
The function $\hat{g}(\omega)$ is positive, has a maximum at $\omega \simeq
0.37$, vanishes linearly as $\omega \to 0$, and has a long tail at infinity,
vanishing proportionally to $\omega^{-3}$. It turns out that only
values of $k_3$ smaller than a few times the muon mass contribute significantly to the
integral in eq.~\eqref{eq:estL:second}. In practice we find
eq.~\eqref{eq:estL:first} to be numerically more stable for the pole
contribution, while eq.~\eqref{eq:estL:second} is more convenient for the
regular part.

\begin{table}
\begin{center}
$-100 \times \Delta a_\text{s}(L) / a_\mu^\text{HVP,LO}$ \\[8pt]
\begin{tabular}{c||c|c|c}
   $mL$ & pole, $F(k^2)=1$ & pole, $F(k^2)=\frac{1}{1-k^2/M^2}$ & regular, NLO $\chi$PT \\
   \hline
   4 & 2.12 & 3.17 &    0.0242 \\
   5 & 1.05 & 1.42 &    0.00576 \\
   6 & 0.491 & 0.630 &  0.00151 \\
   7 & 0.223 & 0.274 &  0.000423 \\
   8 & 0.0986 & 0.118 & 0.000124
\end{tabular}
\end{center}
\caption{Numerical estimates of the pole and regular contributions to $\Delta a_\text{s}(L)$, normalized to $\amu=700 \times 10^{-10}$ and summed over the full series of $e^{- \vert \vec n \vert m L}$ arising from single-pion wrappings. For the pole contribution we take two models of the spacelike pion form factor, $F(k^2)$, and for the regular contribution we take the $\chi$PT expression given in eq.~\eqref{eq:stat:chpt}. The following inputs are used: $m_\mu/m = 106/137$, $m/M = 137/727$, $f_\pi/m=132/137$, $\alpha = 1/137$. We stress that our reference value of $\amu$ does not enter the determination of $\Delta a_{\text{s}}(L)$ and thus multiplying out $700 \times 10^{-10}$ gives a prediction based solely on the four inputs listed.\label{tab:summary}}
\end{table}

In table~\ref{tab:summary} we report numerical calculations of the pole and
regular part of $-100 \times \Delta a_\text{s}(L) / a_\mu^\text{HVP,LO}$, for
physical values of the pion and muon mass, and various values of $L$. Two models
have been used for the electromagnetic pion form factor, the point-like
limit, $F(k^2)=1$, and the so-called \textit{monopole} model \begin{gather}
   F(k^2) = \frac{1}{1 - {k^2}/{M^2}} \,,
\end{gather}
which is argued in ref.~\cite{\monopole} to describe quite well both the experimental and lattice data in
the space-like region $k^2<0$, with $M \simeq 727 \text{ MeV}$. The regular
part has been calculated using the NLO $\chi$PT result (we do not report the
calculation here)
\begin{gather}
   T_\text{reg}(k^2,pk)
   =
   8
   + c k^2
   - \frac{21m^2-16k^2}{18 \pi^2 f_\pi^2}
   + \frac{7m^2-4k^2}{6 \pi^2 f_\pi^2}
   \left[
   \sigma \coth^{-1} \sigma
   \right]_{\sigma = \sqrt{1 - \frac{4m^2}{k^2}}}
   + \dots
   \ ,   \label{eq:stat:chpt}
\end{gather}
where the convention $f_\pi \simeq 132 \text{ MeV}$ has been used, and $c$ is a %
low-energy constant (related to $\ell_5$ and $\ell_6$ in the notation
of ref.~\cite{Gasser:1983yg}). Note that %
$c$ drops out in
eq.~\eqref{eq:estL:second} because of the second derivative with respect to
$k^2$. 
The pole contribution turns out to be greatly enhanced with respect to
the regular contribution, which is perhaps not surprising.

In order to make these numerical predictions
 solid, one would need to assess the systematic errors due to the
chosen models, which is outside of the scope of this paper. Intuitively, the
dominant contribution comes from a single pion loop wrapping around the torus.
The structure of the pion matters quantitatively; for example at $mL=4$ we see that
the point-like approximation accounts only for $\sim 2/3$ of the full finite-$L$
corrections (table \ref{tab:summary}). At larger values, the point-like
approximation is increasingly more accurate. The summation indices $n_{k=1,2,3}$
in eqs.~\eqref{eq:estL:first} and~\eqref{eq:estL:second} can be interpreted as
the wrapping numbers of the pion loop around the $k$th direction of the spatial
torus.

\begin{table}
\begin{center}
\addtolength{\leftskip} {-2cm} %
\addtolength{\rightskip}{-2cm}
$-100 \times \Delta a_\text{s}(L) / a_\mu^\text{HVP,LO}$ \\[8pt]
\begin{tabular}{c||c|c|c|c|c|c|c|c||c}
   $mL$ &
   $|\vec{n}|=1$ & $\sqrt{2}$ & $\sqrt{3}$ &
   $2$ & $\sqrt{5}$ & $\sqrt{6}$ & $2\sqrt{2}$ & $3$ &
   $\sum_{\vec{n}}$ \\
   \hline
   4 &
   1.26 & 1.16 & 0.317 &
   0.104 & 0.193 & 0.0944 & 0.0128 & 0.0174 &
   3.17 \\
   5 &
   0.852 & 0.428 & 0.0813 &
   0.0199 & 0.0287 & 0.0112 & 0.00102 & 0.00117 &
   1.42 \\
   6 &
   0.461 & 0.141 & 0.0189 &
   0.00349 & 0.00394 & 0.00124 & 0.0000764 & 0.0000735 &
   0.630 \\
   7 &
   0.226 & 0.0433 & 0.00417 &
   0.000582 & 0.000515 & 0.000130 & 5.46 $\times 10^{-6}$ & 4.41 $\times 10^{-6}$ &
   0.274 \\
   8 &
   0.104 & 0.0128 & 0.000883 &
   0.0000936 & 0.0000652 & 0.0000132 & 3.79 $\times 10^{-7}$ & 2.57 $\times 10^{-7}$ &
   0.118
\end{tabular}
\end{center}
\caption{Numerical estimates of the pole contributions to $\Delta a_\text{s}(L)$, normalized to $\amu=700 \times 10^{-10}$ for various fixed values of $\vert \vec n \vert$, corresponding to contributions scaling as $e^{- \vert \vec n \vert m L}$. Here we use the monopole model for the spacelike pion form factor, with input parameters as in 
table \ref{tab:summary}.\label{tab:pole-various-n}}
\end{table}

In table~\ref{tab:pole-various-n} we show the breakdown of the pole
contribution to the finite-$L$ corrections for individual fixed values of $|\vec{n}|$. At
$mL=4$, pions wrapping more than once account for more than half of the
finite-$L$ corrections. As already noted in ref.~\cite{\convergence}, terms with
$|\vec{n}|=1$, which generate the leading $e^{-mL}$ exponentials, are
insufficient to capture the finite-$L$ corrections at $mL \lesssim 6$. Terms
with $|\vec{n}| \ge 2$ are subleading with respect to the neglected
$O\left(e^{-\sqrt{2+\sqrt{3}}mL}\right)$ terms at asymptotically large $L$.
However, if the single-pion pole produces an enhancement of prefactors, then it is well motivated to include those terms in the estimate of the finite-$L$ corrections at
moderate values of $L$, as has been done in table %
\ref{tab:summary}.

Here we also comment that these results have already been compared to a full lattice calculation as presented by the BMW collaboration in ref.~\cite{\BMW}. There the authors perform a dedicated analysis of finite-$L$ effects for two volumes $m L_{\text{big}} = 7.35$ and $m L_{\text{ref}} = 4.29$. They compare their results to various estimates including our predictions, in the form
\begin{gather}
[- \Delta a_{\text{s}}^{\text{pole}}(L_{\text{big}} )]  \times 10^{10} = 1.4 \,, \\
[\Delta a_{\text{s}}^{\text{pole}}(L_{\text{big}} )   - \Delta a_{\text{s}}^{\text{pole}}(L_{\text{ref}} )  ] \times 10^{10} = 16.3 \,,
\end{gather}
where we have confirmed their numerical estimates here, in both cases summed over the infinite series of single-pion wraps. The first result matches NNLO $\chi$PT to the digits reported and the second is consistent with the dedicated lattice study from BMW, leading to a reported value of 18.1(2.0)(1.4), as well as with alternative estimates, as discussed in ref.~\cite{\BMW}.

Our numerical results also show excellent consistency with approaches based in models (or direct lattice calculations) of the finite-volume matrix elements and energies entering the spectral decomposition of $G(x_0 \vert \infty, L)$. These give a reasonable approximation of $G(x_0 \vert \infty, L)$, especially for moderate to large current separations, and combining this with a corresponding description of $G(x_0 \vert \infty)$ allows one to estimate $\Delta a(\infty, L)$. 
To compare to values quoted in the literature, we sum the second and third columns of table~\ref{tab:summary} and multiply out the reference value, $\amu$, to reach the following estimates for $\{m L, - \Delta a(\infty, L) \times 10^{10}\}:$
 \[ \Big  \{ \{4, \ 22.4(3.1) \},     \ \          \{5, \ 10.0(0.4) \},    \ \          \{6, \ 4.42(0.06) \},    \ \          \{7, \ 1.924(0.009)  \},    \ \          \{8, \ 0.826 (0.001)   \} \Big \} \,, \]
where we have included an incomplete uncertainty estimate, determined by summing the pole contribution over the modes that we do not fully control, i.e.~from $\vert \textbf n \vert = 2$ to $\infty$. These numbers are consistent with values reported in refs.~\cite{Giusti:2018mdh,Gerardin:2019rua}. 
For example, from table 7 of ref.~\cite{Gerardin:2019rua}, one can infer that the Mainz group estimates a correction is in the range 21.0--22.6 for the E250 CLS ensemble with parameters $m=130(1) \, {\text{MeV}}$, $mL=4.1$. Similarly, in a private correspondence, the authors of ref.~\cite{Giusti:2018mdh} report an estimate of 24.0(1.0) for $mL=4$ at physical masses using the model described in that reference. This value is also plotted in figure 19 of ref.~\cite{Giusti:2018mdh}. 
We find the same clear consistency with preliminary estimates from the RBC/UKQCD collaboration as presented, for example, in ref.~\cite{Lehner2019}.

%}

\subsubsection{$\Delta G_\text{s}(x_0|L)$}

In realistic numerical calculations, the long-distance behavior of the
electromagnetic-current two-point function is dominated by noise. Different
estimators for $\amu$ are then used, which often involve restricting the
integral in eq.~\eqref{eq:stat:amu-def} to the interval $0 \le x_0 \le
x_0^\text{cut}$ and reconstructing the long-distance behaviour of the two-point
function with different techniques. In these approaches it is useful to
understand which region of the $x_0$ integral contributes the most to the
finite-$L$ effects. This is done by calculating the finite-$L$ corrections to the
two-point function in coordinate space. In the following analysis we focus only
on the pole contribution and set $T_\text{reg}=0$, which yields
\begin{equation}
   \Delta G_\text{s}^\text{pole}(x_0|L)
   =
   - \sum_{ \vec{n} \neq \vec{0} }
   \int \frac{d p_3}{2\pi}
   \frac{
   e^{-|\vec{n}| L \sqrt{m^2+p_3^2}}
   }{
   6 \pi |\vec{n}| L
   }
   \Re
   \int \frac{d k_3}{2\pi} e^{i k_3 x_0} 
   \frac{(4m^2 + k_3^2) \, F(-k_3^2)^2}{ k_3^2 + 2k_3p_3 - i\epsilon }
   \label{eq:stat:DeltaGs-pole-1}
   \ .
\end{equation}
The form factor $F(k^2)$ extends to an analytic function on the whole complex
plane except on a cut on the real axis located at $k^2>(2m)^2$. Therefore, the
integral in $k_3$ can be shifted in the complex plane to $\mathbb{R}+i\mu$ where
$\mu$ is any mass with $0 < \mu < 2m$. In doing so, the pole at $k_3 = - p_3 +
\sqrt{p_3^2 + i \epsilon}$ is encircled. After some lengthy algebra one obtains the
representation
\begin{gather}
   \Delta G_\text{s}^\text{pole}(x_0|L)
   = \nonumber \\
    \sum_{ \vec{n} \neq \vec{0} }
   \frac{1}{ 6 \pi |\vec{n}| L }
   \int \frac{d p_3}{2\pi}
   \Bigg\{
   \frac{ (m^2 + p_3^2) \, F(-4 p_3^2)^2
   e^{-|\vec{n}| L \sqrt{m^2+p_3^2}} - m^2 e^{- |\vec{n}| L m}
   }{p_3}
   \sin ( 2 p_3 |x_0| )
   \nonumber \\ \hspace{3.5cm}
   - \Re
   \int_{\mathbb{R}+i\mu} \frac{d k_3}{2\pi} e^{i k_3 |x_0|} 
   \frac{(4m^2 + k_3^2) \, F(-k_3^2)^2 e^{-|\vec{n}| L \sqrt{m^2+p_3^2}}}{ k_3^2 - 4p_3^2 }
   \Bigg\}
   \ ,
   \label{eq:stat:DeltaGs-pole-2}
\end{gather}
which is numerically stable for small and moderate values of $x_0$. However at
large values of $x_0$ the first term in the curly brackets become rapidly
oscillating. In fact, the above representation hides the fact that $\Delta
G_\text{s}^\text{pole}(x_0|L)$ decays exponentially at large $x_0$. This problem
can be cured by representing the sine as imaginary part of the complex
exponential and by shifting the $p_3$ integral for the first term only to
$\mathbb{R}+i\frac{\mu}{2}$. Some algebra yields
\begin{gather}
   \Delta G_\text{s}^\text{pole}(x_0|L)
   =
    \sum_{ \vec{n} \neq \vec{0} }
   \frac{1}{ 6 \pi |\vec{n}| L }
   \Im \int_{\mathbb{R}+i\mu} \frac{d k_3}{2\pi}
   e^{i k_3 |x_0|} (4m^2 + k_3^2) \, F(-k_3^2)^2
   \nonumber \\ \hspace{5cm} \times
   \Bigg\{
   \frac{
   e^{-|\vec{n}| L \sqrt{m^2+\frac{k_3^2}{4}}}
   }{4 k_3}
   -i \int \frac{d p_3}{2\pi}
   \frac{e^{-|\vec{n}| L \sqrt{m^2+p_3^2}}}{ k_3^2 - 4p_3^2 }
   \Bigg\}
   \ .
   \label{eq:stat:DeltaGs-pole-3}
\end{gather}
For the monopole ansatz, one of the integrals can be analytically calculated
\begin{gather}
   \Delta G_\text{s}^\text{pole}(x_0|L) 
   =
   \sum_{ \vec{n} \neq \vec{0} }
   \frac{1}{ 6 \pi |\vec{n}| L }
   \Bigg\{
   \Im
   \int_{\mathbb{R}+i\mu} \frac{d k_3}{2\pi}
   \frac{ e^{i k_3 |x_0|} (4m^2 + k_3^2) M^4}{(M^2+k_3^2)^2}
   \frac{
   e^{-|\vec{n}| L \sqrt{m^2+\frac{k_3^2}{4}}}
   }{4 k_3}
   \nonumber \\ \hspace{3cm}
   + \int \frac{d p_3}{2\pi}
   e^{-|\vec{n}| L \sqrt{m^2+p_3^2}}
   \frac{d}{dz} \left[
   \frac{ e^{- z |x_0|} (z^2-4m^2) M^4}{(z+M)^2 (z^2 + 4p_3^2)}
   \right]_{z = M}
   \Bigg\}
   \ .
   \label{eq:stat:DeltaGs-pole-4}
\end{gather}
In figure~\ref{fig:KDeltaG} we plot the function
\begin{gather}
   \rho(x_0|L) = 
   \frac{
   \mathcal{K}(x_0) \, \Delta G_\text{s}(x_0|L)
   }{\Delta a_\text{s}(L)}
   \ ,
\end{gather}
which is normalized such that $\int_0^\infty dx_0 \, \rho(x_0|L) = 1$.
The quantity $\rho(x_0|L) \, dx_0$ gives the relative contribution of the
interval $[x_0 , x_0+dx_0]$ to the integral defining the finite-$L$ effects. For
$mL=4$ most of the finite-$L$ corrections come from the region $m_\pi x_0
\lesssim 3$ and $25 \%$ of the correction arises from $m x_0 \lesssim 1.3$ corresponding to $x_0 \lesssim  1.9 \, \text{fm} \,$. The dominant region shifts to somewhat larger values of $x_0$ for
larger volumes.

\begin{figure}
   \centering
   \includegraphics[scale=.5]{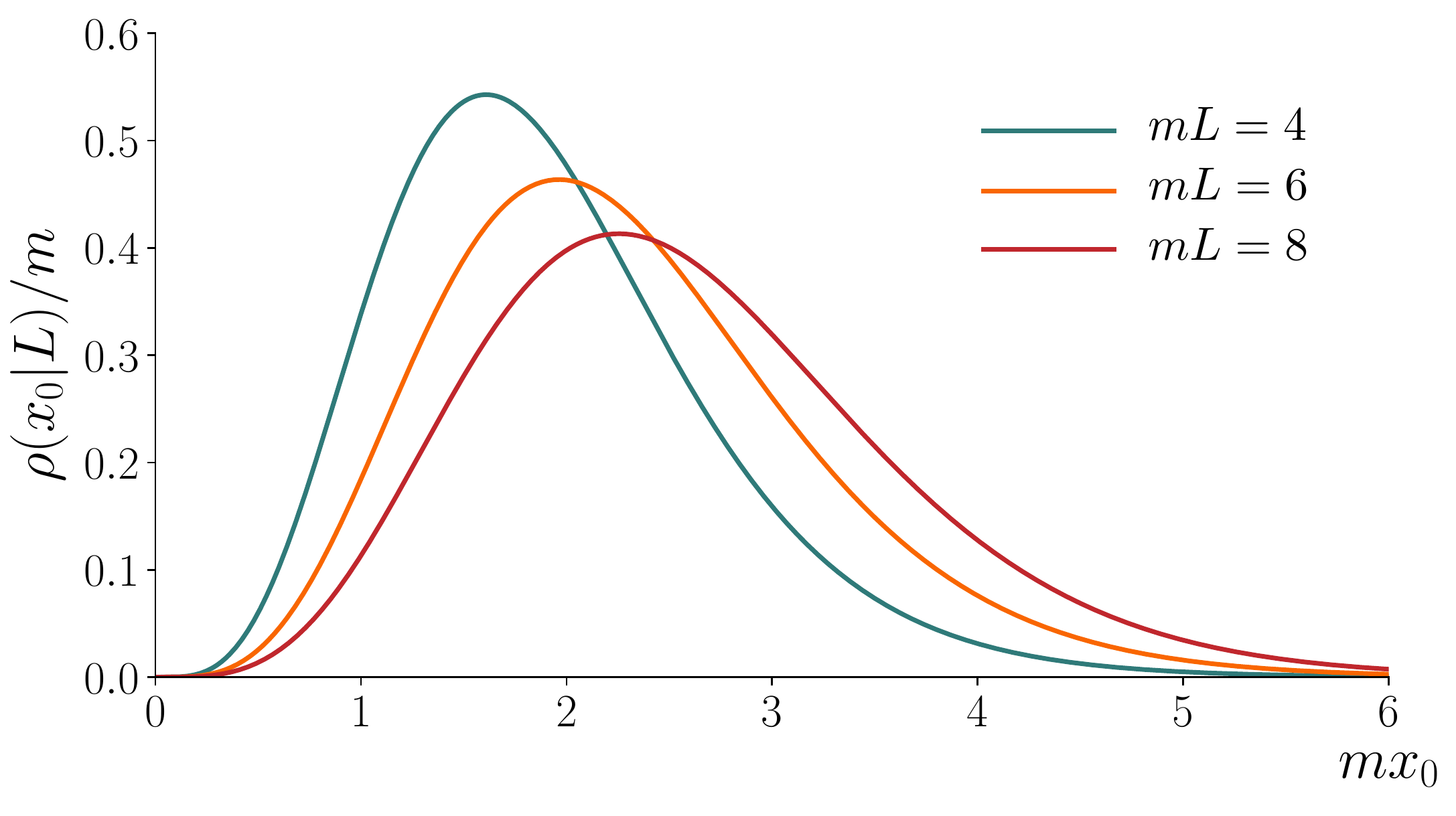}
   \caption{Plot of $\rho(x_0|L)/m$ vs.~$m x_0$ for three values of
   $mL$. The quantity $\rho(x_0|L) \, dx_0$ gives the relative contribution of
   the interval $[x_0 , x_0+dx_0]$ to the integral in
   eq.~\eqref{eq:stat:DeltaIs} defining the finite-$L$ corrections. Note that the area under each curve is equal to 1, by construction.}
   \label{fig:KDeltaG}
\end{figure}

\subsection{Estimate of finite-\texorpdfstring{$T$}{T} corrections}
\label{sec:statement:finiteT}

\begin{table}
\begin{center}
$ 100 \times \Delta  a_{\text{t}}^{\text{WP} \text{--\,} \text{SP}}( x_0^{\text{cut}} \, \vert\, T) /\amu$\\[8pt]
\addtolength{\leftskip} {-2cm} %
\addtolength{\rightskip}{-2cm}
\begin{tabular}{c||c|c|c|c|c|c|c}
& \multicolumn{7}{c}{\hspace{0cm} $x_0^{\text{cut}} \  (m x_0^{\text{cut}})$}   \\[2pt]
$mT$ & $1\, \text{fm} \, (0.69)$ & $2\, \text{fm} \, (1.4)$ & $3\, \text{fm} \,(2.1)$ & $4\, \text{fm}$ \,(2.8) & $5\, \text{fm} \,(3.5)$ & $6\, \text{fm} \,(4.2) $ & $T/2$\\ \hline
 4 & 0.0104 & 0.2944 &  &  &   &  & 1.6262 \\
 5 & 0.0022 & 0.0623 & 0.4149 & &   &  & 0.9539 \\
 6 & 0.0005 & 0.0146 & 0.0972 & 0.3585 & & & 0.5054 \\
 7 & 0.0001 & 0.0037 & 0.0244 & 0.0901 & 0.2411 &  & 0.2498 \\
 8 & $< 10^{-4}$ & 0.0010 & 0.0065 & 0.0238 & 0.0638 &   & 0.1175 \\
 9 & $< 10^{-4}$ & 0.0003 & 0.0018 & 0.0066 & 0.0175 & 0.0384 & 0.0532 \\
 10 & $< 10^{-4}$ &  $0.0001$ & 0.0005 & 0.0019 & 0.0050 & 0.0109 & 0.0234 
\end{tabular}
\end{center}
\caption{Numerical estimates of the single-particle contribution to $\Delta a^{\text{WP}}_\text{t}(T)$ (the wrapped-pion term), normalized to $\amu=700 \times 10^{-10}$ for various fixed values of the cutoff on the $x_0$ integral, denoted by $x_0^{\text{cut}}$. These values, given by eq.~\eqref{eq:stat:WPSP}, depend only on the pion mass and muon mass, taken as $m=137\,\text{MeV}$ and $m_\mu = 106\, \text{MeV}$. \label{tab:sp-various-cuts}}
\end{table}

Our final task for this section is to briefly consider the finite-$T$ corrections, focusing on $\Delta G_\text{t}(x_0 \vert T)$, defined in eq.~\eqref{eq:stat:DeltaGt} above, as well as its contribution to $\amu(T,L)$, the wrapped-pion contribution denoted by $\Delta a_\text{t}^\text{WP}(T)$ and defined in eq.~\eqref{eq:stat:DeltaIt-WP}. We do not discuss the other two contributions, $\Delta a_\text{t}^\text{OB}(T)$ (out of the box) and $\Delta a_\text{t}^\text{BP}(T)$ (backpropagating pion) defined in eqs.~\eqref{eq:stat:DeltaIt-OB} and \eqref{eq:stat:DeltaIt-BP} respectively, as these depend only on the two-point function and thus can be corrected using the lattice data, as is done, for example, in refs.~\cite{\corrOBandBP}. Following an approach analogous to the previous subsection, here we again divide the contributions into a single- and a multi-particle piece (SP and MP, respectively):
\begin{align}
\hat S(x_0, \vec p^2)  & =  \hat S^{\text{SP}}(x_0, \vec p^2)  + \hat S^{\text{MP}}(x_0, \vec p^2)  \,, \\[5pt]
\hspace{-60pt} \hat S^{\text{SP}}(x_0, \vec p^2) & = 2 \sum_{\rho=1}^3  \int d^3 x \int \! \frac{d^3 k}{(2 \pi)^3}  \frac{\langle \textbf p, + \vert j_\rho (x) \vert \textbf k, + \rangle \, \langle \textbf k, + \vert j_\rho (0) \vert \textbf p, + \rangle}{2 \sqrt{m^2 + \vec k^2}} = - 8     \frac{ \vec p^2  }{2 \sqrt{m^2 + \vec p^2}}\,, \end{align}
where the multi-particle contribution is defined by removing $\hat S^{\text{SP}}(x_0, \vec p^2) $ from the full function, defined in eq.~\eqref{eq:stat:Sdef}.

Focusing on the single-particle piece, we next note
\begin{align}
\Delta G^{\text{SP}}_{\text{t}}(x_0 \vert T) & = - \frac{1}{3} \int \! \frac{d^3 p}{(2\pi)^3} \frac{ e^{-T \sqrt{m^2+\vec{p}^2}} }{ 2 \sqrt{m^2+\vec{p}^2} } \ \hat{S}^{\text{SP}}(x_0,\vec{p}^2)  =   \frac{1}{3\pi^2}  \int_0^\infty \! \! dp  \  \frac{p^4  }{   m^2+p^2  }   e^{- T  \sqrt{m^2+ {p}^2}} \,.
\end{align}
Observe that both $\hat S^{\text{SP}}(x_0, \vec p^2) $ and $\Delta G^{\text{SP}}_{\text{t}}(x_0 \vert T)$ are in fact $x_0$ independent, due to an exact cancellation between the single-pion energies of the external state (arising from the thermal trace) and the internal state (arising from inserting a complete set).
The corresponding contribution to $\amu$ is then given by
\begin{equation}
\Delta  a_{\text{t}}^{\text{WP} \text{--\,} \text{SP}}( x_0^{\text{cut}} \, \vert\, T) =   \frac{1}{3\pi^2}  \int_0^\infty \! \! dp  \  \frac{p^4  }{   m^2+p^2  }   e^{- T  \sqrt{m^2+ {p}^2}}  \int_0^{x_0^{\text{cut}}} d x_0 \,  \mathcal K(x_0) \,.
\label{eq:stat:WPSP}
\end{equation}
Here we have extended the definition given in eq.~\eqref{eq:stat:DeltaIt-WP} above by allowing the cutoff on the $x_0$ integral, denoted by $x_0^{\text{cut}}$, to vary independently of $T$. The results, for various values of $T$ and $x_0^{\text{cut}}$, are given in Table \ref{tab:sp-various-cuts}. For realistic parameter choices the correction is found to be at the per-mille level and, for this reason, the multi-particle contribution to $\Delta a^{\text{WP}}_{\text{t}}(T)$ is not considered.

%\clearpage
 
\section{Diagrammatic analysis}
\label{sec:diag}

We proceed now to derive all results described in the previous sections.
As already mentioned in the introduction, we work in an effective
theory of pions with generic local interactions in the isospin-symmetric limit.
We also make use of an all-orders %
 {renormalized} perturbative expansion in an infinite-volume on-shell
renormalization scheme. In particular, this means that the interaction Lagrangian
does not depend on the volume, and the bare propagator in momentum space is given by
$(p^2+m^2)^{-1}$ where $m$ is the (pole) mass of the pion. Periodic boundary
conditions are assumed for the pion fields, but all results can readily be
generalized to the case of phase-periodic boundary conditions, as described in ref.~\cite{\ShortPaper}. Results
concerning finite-$L$ effects also hold with non-periodic boundary
conditions in time, e.g.~with open boundaries. However the structure of
the finite-$T$ corrections is expected to be quite different in this case.

The electromagnetic-current two-point function, and consequently the integral, $\amu(T,L)$,
are represented as a sum of finite-volume Feynman integrals in coordinate space.
The general structure of these integrals is presented in
subsection~\ref{subsec:diag:saddle}. One can show that the large-volume
expansion is of the saddle-point type, which exposes the exponential suppression
of the finite-volume corrections. It turns out that Feynman integrals can be
written as a sum of terms labeled by integers $(n)$ which, roughly speaking,
correspond to winding numbers of pion loops around the various periodic
directions. The asymptotic behaviour of each of these terms is found to be of
the form $e^{-m \mathcal{E}_n(T,L)}$, where the function $\mathcal{E}_n(T,L)$ depends
on $T$, $L$ and the winding numbers $n$, as shown, as well as some geometric properties of the
considered Feynman graph. These exponentials are characterized in
subsection~\ref{subsec:diag:asymptotic}. The next step is to classify the terms
that contribute to the largest exponentials. This is done in two steps, for
convenience of presentation. First, in section~\ref{subsec:diag:separation} it is
shown that, up to subleading terms, finite-spacetime corrections are given by
the sum of finite-$L$ and finite-$T$ contributions, as displayed in eqs.~\eqref{eq:stat:separation} and \eqref{eq:stat:Gexp}. Then, in
sections~\ref{subsec:diag:leading-L} and~\ref{subsec:diag:leading-T}, the
terms contributing to the leading exponentials are completely characterized.

By the end of section~\ref{subsec:diag:leading-T}, the leading large-volume corrections are completely identified, 
but are still represented as formal sums over Feynman diagrams. The final task, then, is to find a representation in terms of 1PI proper vertices, thereby effectively resumming the all-orders perturbative expansion.
Before turning to this, in section~\ref{subsec:diag:2pt-function} we briefly comment on differences that arise in the scaling of finite-volume corrections to the two-point function, as compared to the integral of the latter with the kernel. We conclude section~\ref{sec:diag}, in section~\ref{subsec:diag:skeleton}, by identifying compact, momentum-space expressions for both the correlator, $G(x_0 \vert T, L)$, and the estimator, $\amu(T,L)$.
The derivation remains incomplete, as the resulting formulae hide the exponential decay
of the finite-volume corrections. This is the starting point in
section~\ref{sec:calculation} in the derivation of our final results.

\subsection{Finite-volume diagrammatic expansion}%
\label{subsec:diag:saddle}

The two-point function~\eqref{eq:stat:G-def} can be represented as a sum of Feynman
integrals
\begin{gather}
   G(x_0|T,L) = \sum_{\mathcal{D}} \mathcal{F}_{T,L}(x_0|\mathcal{D})
   \label{eq:G-feynman-diagram}
   \ ,
\end{gather}
where $\mathcal{F}_{T,L}(x_0|\mathcal{D})$ is the contribution to
$G(x_0|T,L)$ associated to the Feynman diagram $\mathcal{D}$. Let
$\mathcal{G}$ be the abstract graph associated to $\mathcal{D}$. Let
$\mathcal{V}$ and $\mathcal{L}$ be the sets of its vertices and lines
respectively. Following refs.~\cite{\MartinStable} and~\cite{\GraphBook}, we
consider all lines with arbitrary orientation. Given a line, $\ell$, we denote by
$i(\ell)$ and $f(\ell)$ its initial and final vertices, respectively. Two special
vertices, $a$ and $b$, are identified to correspond to the insertion of the
operators $j_\mu(0)$ and $j_\mu(x)$ respectively. We introduce a number of
definitions which will be useful throughout this paper.

\textit{Paths.} A path $P$ is defined as a set of lines with the property that
a sequence $v_1, v_2, \dots , v_N$ of pairwise distinct vertices exists,
together with a labelling, $\ell_1,\ell_2,\dots,\ell_{N-1}$, of all lines in $P$,
such that $v_k$ and $v_{k+1}$ are the endpoints of $\ell_k$. An orientation of the
path, $P$, is specified by assigning a number $[P:\ell] \in \{-1,1\}$ for every
line, $\ell \in P$, with the following properties:
\begin{gather*}
   [P:\ell] = [P:\ell'] \ , \quad \text{if } i(\ell) = f(\ell') \ , \\
   [P:\ell] = -[P:\ell'] \ , \quad \text{if } i(\ell) = i(\ell') \text{ or } f(\ell)=f(\ell')\,, \qquad (\ell \neq \ell') \ .
\end{gather*}
We also define $[P:\ell]=0$ if $\ell \not\in P$. Given an
oriented path, its initial and final vertices are defined in the natural way. If
$v$ and $w$ are respectively the initial and final vertices of $P$, we
will say that $P$ is a path from $v$ to $w$.

\textit{Loops.} A loop $C$ is defined as a set of lines, with the
property that a sequence of pairwise distinct vertices, $v_1,\dots,v_{N}$, and a
labelling, $\ell_1,\dots,\ell_N$, of the elements of $C$ exist, such that
$v_k$ and $v_{k+1}$ are the endpoints of $\ell_k$ for $1 \le k < N$, and $v_N$ and $v_1$ are
the endpoints of $\ell_N$. The orientation of loops is defined in complete analogy
to the orientation of paths.

\textit{Connected graphs.} The graph $\mathcal{G}$ is said to be connected if
and only if, for any pair of distinct vertices $v$ and $w$, a path exists from
$v$ to $w$. The graph $\mathcal{G}$ is said to be disconnected if and only if it
is not connected. A disconnected graph can be split in connected components in
the natural way. (See ref.~\cite{\GraphBook}, sections 1.1 and 1.2, for more
details.)

\textit{Deletion of lines.} Given a set of lines, $A \subseteq \mathcal{L}$, we
denote by $\mathcal{G} - A$ the graph obtained from $\mathcal{G}$ by deleting
the lines in $A$. $\mathcal{G} - A$ and $\mathcal{G}$ have the same set of
vertices, $\mathcal{V}$, and $\mathcal{L} \setminus A$ is the set of lines of
$\mathcal{G} - A$. Further, $\mathcal{G} - A$ inherits the incidence relation
from $\mathcal{G}$ in the natural way. (See ref.~\cite{\GraphBook}, sections 1.1
and 1.2, for more details.) Note that even if $\mathcal{G}$ is connected,
$\mathcal{G} - A$ may be disconnected.

\textit{Cut-sets.} A cut-set $S$ is a minimal set of lines with the property
that one connected component in $\mathcal{G}$ becomes disconnected in
$\mathcal{G}-S$. Note that, if $S$ is a cut-set, then $\mathcal{G}-S$ has
exactly one more connected component than $\mathcal{G}$. If $\mathcal{G}$ is
connected, given two vertices $v$ and $w$, we say that $S$ disconnects $v$ and
$w$, if and only if they belong to different connected components of $\mathcal{G} - S$. An
orientation of the cut-set $S$ is specified by assigning a number, $[S:\ell] \in
\{-1,1\}$, for every line $\ell \in S$ with the requirement that $[S:\ell] =
[S:\ell']$ if $i(\ell)$ and $i(\ell')$ are in the same connected component of $G
- S$, and $[S:\ell] = -[S:\ell']$ otherwise. We also define $[S:\ell]=0$ if
$\ell \not\in S$.

\textit{$n$-particle (ir)reducibility.} The connected graph $\mathcal{G}$ is
said the be $n$-particle irreducible ($n$PI) between its distinct vertices $v$
and $w$ if and only if $v$ and $w$ always belong to the same connected component
of $\mathcal{G} - \{ \ell_1 , \dots, \ell_n\}$, no matter which lines $\ell_1 ,
\dots, \ell_n$ are deleted. $\mathcal{G}$ is said the be $n$PI if and only if it
is $n$PI between any distinct pair of its vertices. $\mathcal{G}$ is said the be
$n$-particle reducible if and only if it is not $n$PI. As noted in
ref.~\cite{Luscher:1985dn}, theorem 2.2, if $\mathcal{G}$ is $n$PI between $v$
and $w$, there exists $n+1$ disjoint paths connecting $v$ and $w$. This fact
will be used often in this paper, and can be proven using the techniques
described in ref.~\cite{\GraphBook}, chapter 5 (in particular theorem 5-4).

\bigskip

\begin{proposition}[One-particle irreducibility]\label{prop:1PI}
   
   In a generic effective theory of pions with exact isospin symmetry, the
   Feynman diagrams contributing to eq.~\eqref{eq:G-feynman-diagram} are
   connected and 1PI. This is true also in a finite periodic box with
   periodic boundary conditions.
   
\end{proposition}

\begin{proof}
   Connectedness of the Feynman diagrams follows from the fact that $\langle
   j_\mu \rangle_{T,L}=0$, e.g. because of charge-conjugation (C) invariance.
      Let us prove one-particle irreducibility. Let $\pi^a(x)$ be the real pion
   field with isospin index $a$, and let $S(\pi)$ be the action. We introduce a
   source, $J^a(x)$, for the pion field and a source, $A_\mu(x)$, for the Euclidean electromagnetic
   current $j_\mu(x)$. The Euclidean-signature generating functional for connected
   graphs is then defined as
   \begin{gather}
      e^{W(J,A)} = \int [d\pi] \, e^{-S(\pi) + (A,j) + (J,\pi)} \ .
   \end{gather}
   Let $\Gamma(\pi,A)$ be the Legendre transform of $W(J,A)$ with respect to $J$, i.e.
   \begin{gather}
      \Gamma(\pi,A) = \left[
      W(J,A) - (J,\pi)
      \right]_{
      \pi = \frac{\delta W}{\delta J}(J,A)
      }
      \ .
   \end{gather}
   The 1PI irreducible effective action has the form
   \begin{gather}
      \Gamma(\pi,A) = - \frac{1}{2} \left( \pi , \{-\partial^2+m^2\} \pi \right) + \Gamma_\text{I}(\pi,A) \ ,
   \end{gather}
   where $\Gamma_\text{I}(\pi,A)$ can be seen as a sum of 1PI Feynman diagrams
   in which the external lines correspond to insertions of fields $\pi^a(x)$ (rather
   than propagators). In this framework the term $(A,j)$ is seen as yet another
   interaction term which breaks isospin symmetry. Therefore, the field, $A$, is
   attached to the internal vertices of the Feynman diagrams. Taking $n$ derivatives
   with respect to $A$ and setting $A=0$ amounts to selecting only the Feynman
   diagrams in which the interaction $(A,j)$ has been inserted $n$ times. The
   derivative with respect to $A$ does not change the property of
   $\Gamma_\text{I}(\pi,A)$ to be 1PI.
   Standard algebra yields
   \begin{gather}
      \langle j_\mu(x) j_\nu(0) \rangle
      =
      \left. \frac{\delta^2 W}{\delta A_\mu(x) \delta A_\nu(0)} \right|_{J=A=0}
      = \nonumber \\ \qquad
      \left. \frac{\delta^2 \Gamma}{\delta A_\mu(x) \delta A_\nu(0)} \right|_{\pi=A=0}
      + \int_z \left. \frac{\delta^2 \Gamma}{\delta A_\mu(x) \delta \pi^a(z)} \right|_{\pi=A=0}
      \left. \frac{\delta^2 W}{\delta J^a(z) \delta A_\nu(0)} \right|_{J=A=0}
      \ ,
   \end{gather}
   where we have already used that the equation of motion, $\frac{\delta
   \Gamma}{\delta \pi} =0$ for $A=0$, is solved by $\pi=0$ because of isospin
   symmetry. We also note that
   \begin{gather}
      \left. \frac{\delta^2 W}{\delta J^a(z) \delta A_\nu(0)} \right|_{J=A=0}
      =
      \langle \pi^a(z) j_\mu(0) \rangle_\text{c}
      =
      0 \ ,
   \end{gather}
   where the subscript, c, indicates that only the connected contribution is to be kept.
   This is better understood in the electric-charge basis. Let $\pi_q(x)$ be the
   pion field with charge $q=0,\pm 1$. We observe that electric charge
   conservation implies that $ \langle  \pi_q(z)  j_\mu(0) \rangle$ vanishes if $q
   \neq 0$. Since $j_\mu(0)$ is C-odd and $\pi_0(x)$ is C-even, by C-invariance
   $ \langle \pi_0(x)  j_\mu(0) \rangle=0$ as well. We deduce %
   \begin{gather}
      \langle j_\mu(x) j_\nu(0) \rangle
      =
      \left. \frac{\delta^2 \Gamma_\text{I}}{\delta A_\mu(x) \delta A_\nu(0)} \right|_{\pi=A=0}
      \ ,
   \end{gather}
   i.e.~only 1PI diagrams contribute to the current two-point function. We stress that, although the $T,L$ labels have been suppressed here, all arguments used hold in a finite-volume spacetime.
\end{proof}

\bigskip

Let us discuss the structure of the Feynman integral,
$\mathcal{F}_{T,L}(x_0|\mathcal{D})$. In coordinate space, each line $\ell$
contributes with a finite-volume propagator
$\Delta_{T,L}\big (x[f(\ell)]-x[i(\ell)] \big )$ to the Feynman integrand.\footnote{We use $x(v)$ and $x[v]$ interchangeably, to improve the readability of equations.}
Here the propagator can indicate both a charged an neutral state and, as this has no effect on its functional form, we do not include a charge label on $\Delta_{T,L}$. 
The interaction
vertices are associated to homogeneous partial differential operators with
constant coefficients, acting on the arguments of the propagators in the
diagram. These differential operators do not depend on $T$ and $L$. We denote by
$\mathbb{V}$ the product of all these differential operators associated to the
diagram $\mathcal{D}$. The Feynman integral can then be written as
\begin{gather}
   \mathcal{F}_{T,L}(x_0|\mathcal{D})
   =
   \int_\L d^3x(b)
   \left[ \prod_{v \in \mathcal{V} \setminus \{a,b\} }
   \int_\L d^4x(v)
   \right]
   \mathbb{V} \left\{
   \prod_{\ell \in \mathcal{L}} \Delta_{\L} \big (x[f(\ell)] - x[i(\ell)]   \big )
   \right\}_{\substack{x(a)=0\\x_0(b)=x_0}} \ .
   \label{eq:F-D-1}
\end{gather}
The symbols $\int_\L d^3x$ and $\int_\L d^4x$ are shorthand notations for $\prod_{k=1}^3
\int_0^L dx_k$ and $\int_0^T dx_0 \prod_{k=1}^3
\int_0^L dx_k$, respectively.  

The finite-volume propagator is conveniently represented in coordinate space by
making use of the Poisson summation formula
\begin{gather}
   \Delta_{\L}(x ) = \sum_{n \in \mathbb{Z}^4} \Delta(x + \L n) \ ,
   \label{eq:prop-poisson}
\end{gather}
in terms of the infinite-volume propagator, $\Delta(x)$,
\begin{gather}
   \Delta(x) = \int \frac{d^4 p}{(2 \pi)^4} \frac{e^{i p x}}{p^2 + m^2} \ ,
\end{gather}
and the diagonal matrix
\begin{gather}
   \L = \text{diag}(T,L,L,L) \ .
\end{gather}
For every line, $\ell$, we then have a new summation variable, $n(\ell)
\in \mathbb{Z}^4$. By exchanging summations and integrations, one can split
$\mathcal{F}_{T,L}(x_0|\mathcal{D})$ into a sum of terms: one for each possible
configuration $(n(\ell))_{\ell \in \mathcal{L}}$ of Poisson parameters, i.e.
\begin{multline}
   \label{eq:F-D-2}
   \mathcal{F}_{T,L}(x_0|\mathcal{D}) = 
   \sum_n \int_\L d^3x(b)
   \left[ \prod_{v \in \mathcal{V} \setminus \{a,b\} }
   \int_\L d^4x(v)
   \right]
   \\
   \times
   \mathbb{V}
   \left\{
   \prod_{\ell \in \mathcal{L}} \Delta \big (x[f(\ell)] - x[i(\ell)] + \L n(\ell) \big )
   \right\}_{\substack{x(a)=0\\x_0(b)=x_0}} \ .
\end{multline}
One shows easily that the integrand, but not the integration domain, is
invariant under the \textit{gauge transformation}
\begin{gather}
   n(\ell) \to n^{\lambda}(\ell) = n(\ell) + \lambda[f(\ell)] - \lambda[i(\ell)] \ , \\
   x(v) \to x(v) - \L \lambda(v) \ ,
\end{gather}
where $\lambda(v) \in \mathbb{Z}^4$, with the \textit{boundary conditions}
\begin{gather}
   \lambda(a) = 0 \ , \qquad
   \lambda_0(b) = 0 \ ,
\end{gather}
which are needed in order to preserve the constraints $x(a)=0$ and $x_0(b)=x_0$.
We will call $\lambda(v)$ an \textit{admissible gauge transformation} if it
satisfies the above boundary conditions. The definition of the gauge
transformation motivates the identification of $n(\ell)$ as a $\mathbb{Z}^4$
gauge field over the graph $\mathcal{D}$. Two gauge fields are said to be
\textit{gauge-equivalent} if they are related by an admissible gauge transformation.

\bigskip

\begin{proposition}\label{prop:factorization}
   Let $\mathfrak{G}$ be the set of gauge fields, let $\Lambda$ be the set of
   admissible gauge transformations, and let $\mathfrak{G}/\Lambda$ be the set
   of gauge orbits. Let $\mathfrak{R}$ be a section of $\mathfrak{G}/\Lambda$,
   i.e.~a set constructed by picking exactly one representative per gauge orbit.
   Given a gauge field, $n \in \mathfrak{G}$, a unique pair $(n_\mathfrak{R} ,
   \lambda) \in \mathfrak{R} \times \Lambda$ exists, such that $n =
   n_\mathfrak{R}^\lambda$.
\end{proposition}

\begin{proof}
   Given $n \in \mathfrak{G}$, by definition of $\mathfrak{R}$, there is a
   unique $n_\mathfrak{R} \in [n]$ (where $[n]$ denotes the gauge orbit of $n$).
   This clearly implies that $\lambda \in \Lambda$ exists such that $n =
   n_\mathfrak{R}^\lambda$, i.e.
   \begin{gather}
      n(\ell) = n_\mathfrak{R}(\ell) + \lambda[f(\ell)] - \lambda[i(\ell)] \ .
   \end{gather}
   We need to show only that such $\lambda$ is unique. Given a vertex $v$, since
   $\mathcal{G}$ is connected, a path $P(v)$ exists from $a$ to $v$.
   Using $\lambda(a)=0$, one finds
   \begin{gather}
      \lambda(v)
      =
      \sum_{\ell \in P(v)} [ P(v) : \ell ] \, \{ \lambda[f(\ell)] - \lambda[i(\ell)] \}
      =
      \sum_{\ell \in P(v)} [ P(v) : \ell ] \, \{ n(\ell) - n_\mathfrak{R}(\ell) \}
      \ ,
   \end{gather}
   i.e.~$\lambda$ is uniquely determined by $n$ and $n_\mathfrak{R}$.
\end{proof}

\bigskip

By using proposition~\ref{prop:factorization}, one can make the following replacement in
eq.~\eqref{eq:F-D-2}
\begin{gather}
   \sum_{n \in \mathfrak{G}} f(n) \to \sum_{n \in \mathfrak{R}} \sum_{\lambda \in \Lambda} f(n^\lambda)
   \ .
\end{gather}
This corresponds to replacing the argument of each propagator as follows:
\begin{multline}
   x[f(\ell)] - x[i(\ell)]+ \L n(\ell)
  \ \    \to \ \ 
   x[f(\ell)] - x[i(\ell)]+ \L n^\lambda(\ell) =   \\
   \Big ( x[f(\ell)] + \L\lambda[f(\ell)] \Big ) - \Big ( x[i(\ell)] + \L\lambda[i(\ell)] \Big ) + \L n(\ell)
   \ .
\end{multline}
Then one can apply a change variables in the $x$ integrals
\begin{gather}
   x(v) \to x(v) - \L\lambda(v) \ ,
\end{gather}
which shifts the integration domain by $\L\lambda(v)$, and the sum over
$\lambda(v)$ reconstructs the integral over the whole $\mathbb{R}^4$ (or
$\mathbb{R}^3$ if $v=b$), i.e.
\begin{gather}
   \mathcal{F}_{T,L}(x_0|\mathcal{D}) = \sum_{n \in \mathfrak{R}} \mathcal{F}_{T,L}(x_0|\mathcal{D},n)
   =
   \sum_{[n] \in \mathfrak{G}/\Lambda} \mathcal{F}_{T,L}(x_0|\mathcal{D},n)
   \label{eq:F-D-3}
   \ ,
\end{gather}
where we have introduced
\begin{multline}
   \label{eq:F-Dn-1}
   \mathcal{F}_{T,L}(x_0|\mathcal{D},n)
   =
   \int_{\mathbb{R}^3} d^3x(b)
   \Bigg[ \prod_{v \in \mathcal{V} \setminus \{a,b\} }
   \int_{\mathbb{R}^4} d^4x(v)
   \Bigg] \\ \times
   \mathbb{V} \left\{
   \prod_{\ell \in \mathcal{L}} \Delta \big (x[f(\ell)] - x[i(\ell)]+ \L n(\ell) \big )
   \right\}_{\substack{x(a)=0\\x_0(b)=x_0}} \ .
\end{multline}
In the second equality in eq.~\eqref{eq:F-D-3}, we have used the observation
that $\mathcal{F}_{T,L}(x_0|\mathcal{D},n)$ is gauge invariant and therefore one
can replace the sum over the section $\mathfrak{R}$ with the sum over the set
 of gauge orbits $\mathfrak{G}/\Lambda$. \emph{The advantage of this
representation lies in the fact that $\L$ appears now only in the argument of
the propagator.}

The decomposition~\eqref{eq:F-D-3} induces the corresponding decomposition for
the integrals
\begin{gather}
   \amu(T,L)
   =
   \sum_{\mathcal{D}} \sum_{[n] \in \mathfrak{G}/\Lambda} \mathcal{I}_{T,L}(\mathcal{D},n)
   \label{eq:I-feynman-diagram}
   \ , \\
   \mathcal{I}_{T,L}(\mathcal{D},n) =
   \int_0^{\frac{T}{2}} dx_0 \,  { \mathcal K}(  x_0) \, \mathcal{F}_{T,L}(x_0|\mathcal{D},n)
   \ .
   \label{eq:I-Dn}
\end{gather}
 
\subsection{Large-volume asymptotic behaviour}
\label{subsec:diag:asymptotic}

In this section we study the asymptotic behaviour of
$\mathcal{I}_{T,L}(\mathcal{D},n)$ in the $T,L \to \infty$ limit, with the ratio
\begin{gather}
   r = \frac{T}{L} \,,
\end{gather}
held constant. Specifically, we show that the asymptotic behaviour is
given by a saddle-point expansion and provide an upper bound for the asymptotic
scaling with $T$ and $L$.

Using the heat-kernel representation of the propagator,
\begin{gather}
   \Delta(x) = \int_0^\infty ds \ (4\pi s)^{-2} \exp \left\{ - \frac{x^2}{4s} - m^2 s \right\}
   \ ,
\end{gather}
and rescaling $x(v) \to \L x(v)$, $s \to L s/2m$, the
integral~\eqref{eq:I-Dn} assumes the general form
\begin{gather}
   \mathcal{I}_{T,L}(\mathcal{D},n)
   =
   \int_0^{1/2} dx_0(b) \   {\mathcal K}( L_0 x_0(b))
   \int_{\mathbb{R}^3} d^3x(b) \Bigg [ \prod_{v \in \mathcal{V} \setminus \{a,b\} } \int_{\mathbb{R}^4} d^4x(v) \Bigg ]
   \nonumber \\ \hspace{3cm}
   \times  \left[ \prod_{\ell \in \mathcal{L} } \int_0^\infty ds_\ell \right] 
   (\det \L)^{|\mathcal{V}|-1} \frac{P(s,\L x,\L n)}{Q(s)}
   \exp \big [ -m L S(s,x,n) \big ]
   \ ,
   \label{eq:I-Dn-2}
\end{gather}
where $P(s,x,n)$ and $Q(s)$ are polynomials and
\begin{gather}
   S(s,x,n)
   =
   \frac{1}{2}
   \sum_{\ell \in \mathcal{L}}
   \left\{
   s_\ell
   + \frac{1}{s_\ell}
   r^2 \left[ \delta x_0(\ell) + n_0(\ell) \right]^2
   + \frac{1}{s_\ell}
   \left\| \delta \vec{x}(\ell) + \vec{n}(\ell) \right\|_2^2
   \right\}
   \ ,
\end{gather}
with the definition
\begin{gather}
   \delta x_\mu(\ell) = x_\mu[f(\ell)] - x_\mu[i(\ell)] \ .
   \label{eq:deltaxDef}
\end{gather}
Note that $S$ does not depend on $L$, and the $L$-dependence is explicit in
the exponential in eq.~\eqref{eq:I-Dn-2}. The $L \to \infty$ limit of
integral~\eqref{eq:I-Dn-2} is of the saddle-point type, obtained by expanding
around the minima of $S$. We use crucially the fact that the kernel
$\widehat{\mathcal K}(m_\mu x_0)$ diverges polynomially at $x_0 \to +\infty$,
and therefore it does not take part in the determination of the saddle point. In
particular
\begin{gather}
   \ln \mathcal{I}_{T,L}(\mathcal{D},n) = - m L \epsilon_r(n) + O(\ln L)
   \ ,
   \label{eq:saddle-1}
\end{gather}
with the definition
\begin{gather}
   \epsilon_r(n) = \min_{\substack{s \in (0,\infty) , \\ x \text{ with } x(a)=0 \\ x_0(b) \in [0,1/2]}} S(s,x,n)
   =
   \min_{\substack{x \text{ with } x(a)=0 \\ x_0(b) \in [0,1/2]}}
   \sum_{\ell \in \mathcal{L}}
   \sqrt{
   r^2 \left[ \delta x_0(\ell) + n_0(\ell) \right]^2
   + \left\| \delta \vec{x}(\ell) + \vec{n}(\ell) \right\|_2^2
   }
   \ ,
   \label{eq:epsilon-1}
\end{gather}
where the minimum with respect to the variables $s_\ell$ has been explicitly
calculated. The function $\epsilon_r(n)$ determines the asymptotic behaviour of
$\ln \mathcal{I}_{T,L}(\mathcal{D},n)$. We will write eq.~\eqref{eq:saddle-1} as
\begin{gather}
   \mathcal{I}_{T,L}(\mathcal{D},n) = O \left( e^{ - m L \epsilon_r(n) } \right)
   \ ,
   \label{eq:saddle-2}
\end{gather}
where we ignore multiplicative powers of $L$ in the $O( \cdot )$ symbol.

A complicated aspect of the function to be minimized in eq.~\eqref{eq:epsilon-1}
is that it intertwines various components of the gauge field. We aim at a
formula that allows us to study the temporal and spatial directions separately.
The Cauchy-Schwartz inequality implies
\begin{gather}
   \sqrt{
   r^2 \left[ \delta x_0(\ell) + n_0(\ell) \right]^2
   + \left\| \delta \vec{x}(\ell) + \vec{n}(\ell) \right\|_2^2
   }
   \ge
   \frac{r^2 \left| \delta x_0(\ell) + n_0(\ell) \right|
   + \left\| \delta \vec{x}(\ell) + \vec{n}(\ell) \right\|_2 }
   { \sqrt{r^2+1} }
   \ .\label{eq:cauchy-schwartz}
\end{gather}
After summing over the lines and minimizing with respect to $x$ one finds
\begin{gather}
   \epsilon_r(n)
   \ge
   \frac{ r^2 \hat{\epsilon}_0(n_0) + \hat{\epsilon}_\text{s}(\vec{n}) }
   { \sqrt{r^2+1} }
   \ , \label{eq:epsilon-2}
\end{gather}
with the definitions
\begin{gather}
   \hat{\epsilon}_0(n_0) = 
   \min_{\substack{x_0 \text{ with } x_0(a)=0 \\ x_0(b) \in [0,1/2]}}
   \sum_{\ell \in \mathcal{L}} \left| \delta x_0(\ell) + n_0(\ell) \right|
   \ , \label{eq:hatepsilon-0} \\
   \hat{\epsilon}_\text{s}(\vec{n}) = 
   \min_{\vec{x} \text{ with } \vec{x}(a)=\vec{0}}
   \sum_{\ell \in \mathcal{L}} \left\| \delta \vec{x}(\ell) + \vec{n}(\ell) \right\|_2
   \ . \label{eq:hatepsilon-s}
\end{gather}
Finally, combining with the trivial inequalities $\epsilon_r(n) \ge r
\hat{\epsilon}_0(n_0)$ and $\epsilon_r(n) \ge \hat{\epsilon}_\text{s}(\vec{n})$,
and multiplying by $L$, one obtains
\begin{gather}
   L \epsilon_r(n)
   \ge
   \max \left\{
   T \hat{\epsilon}_0(n_0)
   ,
   L \hat{\epsilon}_\text{s}(\vec{n})
   ,
   \frac{ T^2 \hat{\epsilon}_0(n_0) + L^2 \hat{\epsilon}_\text{s}(\vec{n}) }
   { \sqrt{T^2+L^2} }
   \right\}
   \ . \label{eq:epsilon-3}
\end{gather}
This inequality provides a lower bound for the function $\epsilon_r(n)$, i.e.~an
upper bound for the asymptotic behaviour of $\mathcal{I}_{T,L}(\mathcal{D},n)$. In
order to estimate the asymptotic behaviour of $\mathcal{I}_{T,L}(\mathcal{D},n)$,
one needs to understand the possible values of the functions
$\hat{\epsilon}_0(n_0)$ and $\hat{\epsilon}_\text{s}(\vec{n})$. In particular,
the leading asymptotic behaviour is governed by the gauge orbits $[n]$ for which
those functions take the smallest possible values.\footnote{
An important point to keep in mind
is that the functions $\hat{\epsilon}_0(n_0)$ and
$\hat{\epsilon}_\text{s}(\vec{n})$ are gauge invariant, and can thus be
equivalently thought of as functions of gauge orbits rather than gauge fields.
} Appendix~\ref{app:epsilon} is devoted to the analysis of the functions
$\hat{\epsilon}_0(n_0)$ and $\hat{\epsilon}_\text{s}(\vec{n})$. In the next
sections, we will report only a small number of selected results as we need
them, focusing on their consequences for the asymptotic behaviour of
$\mathcal{I}_{T,L}(\mathcal{D},n)$.
 
\subsection{Separation of finite-\texorpdfstring{$L$}{L} and finite-\texorpdfstring{$T$}{T} corrections}
\label{subsec:diag:separation}

The goal of this section is to prove eq.~\eqref{eq:stat:separation}, i.e.~that $\Delta a(T,L)$ separates into a sum of finite-$L$ and finite-$T$ corrections up to subleading exponentials that mix the two scales.
It is convenient to first introduce the following definition.

\textit{Pure gauge field.} The gauge field $n_\mu$ (or $n$, $\vec{n}$) is said
to be pure if and only if it is gauge equivalent to zero.

Let us separate the contributions of pure gauge fields in
eq.~\eqref{eq:I-feynman-diagram}, which yields
\begin{gather}
   \amu(T,L)
   =
   \sum_{\mathcal{D}} \mathcal{I}_{T,L}(\mathcal{D},0)
   + 
   \sum_{\mathcal{D}}
   \Bigg\{
   \sum_{ \substack{ [n_0] = 0 \\ [\vec{n}] \neq \vec{0}}}
   + \sum_{ \substack{ [n_0] \neq 0 \\ [\vec{n}] = \vec{0}}}
   + \sum_{ \substack{ [n_0] \neq 0 \\ [\vec{n}] \neq \vec{0}}}
   \Bigg\}
   \mathcal{I}_{T,L}(\mathcal{D},n)
   \ .
   \label{eq:I-leading-0}
\end{gather}
We next rewrite the $n=0$ term on the right-hand side of this equation. Using
eq.~\eqref{eq:I-Dn} gives
\begin{gather}
   \sum_{\mathcal{D}} \mathcal{I}_{T,L}(\mathcal{D},0)
   =
   \int_0^{\frac{T}{2}} dx_0 \ {\mathcal K}(x_0) \sum_{\mathcal{D}} \mathcal{F}_{T,L}(x_0|\mathcal{D},0)
   = \nonumber \\ \qquad
   \int_0^{\frac{T}{2}} dx_0 \ {\mathcal K}(x_0) \sum_{\mathcal{D}} \mathcal{F}_{\infty}(x_0|\mathcal{D},0)
   =
   \int_0^{\frac{T}{2}} dx_0 \ {\mathcal K}(x_0) G(x_0|\infty)
   \ ,
\end{gather}
where we have used that $\mathcal{F}_{T,L}(\mathcal{D},n)$ does not depend on either $T$ or $L$ if $n=0$,
as follows from eq.~\eqref{eq:F-Dn-1}. We have additionally substituted
eq.~\eqref{eq:G-feynman-diagram} for $T=L=\infty$. 
Splitting the integrals in
$x_0$ as $\int_0^{\frac{T}{2}} = \int_0^{\infty} - \int_{\frac{T}{2}}^\infty$,
we readily get
\begin{gather}
   \sum_{\mathcal{D}} \mathcal{I}_{T,L}(\mathcal{D},0)
   =
   \amu(\infty)
   +   \Delta a_\text{t}^\text{OB}(T)
   \ ,
   \label{eq:asympt:nzero}
\end{gather}
where the out-of-the-box (OB) contribution is defined in eq.~\eqref{eq:stat:DeltaIt-OB} above,\footnote{Throughout the manuscripts we will repeat equations for convenience a number of times and, in all subsequent cases, simply do this without stating it explicitly, using an unnumbered equation for the repeated instance.}
\begin{gather}
     \Delta a_\text{t}^\text{OB}(T)
   =
   - \int_{\frac{T}{2}}^\infty dx_0 \ {\mathcal K}(x_0) G(x_0|\infty)
   \, . \nonumber
\end{gather}
Since $G(x_0|\infty) = O(e^{-2mx_0})$, and for large $x_0$ and the kernel diverges
polynomially, it follows that
\begin{gather}
    \Delta a_\text{t}^\text{OB}(T)
   =
   O \left( e^{ - m T} \right)
   \ .
   \label{eq:asympt-0}
\end{gather}

At this stage we have recast eq.~\eqref{eq:I-leading-0} as
\begin{gather}
   \Delta a(T,L)
   =
   \Delta a_\text{t}^\text{OB}(T)
   + 
   \sum_{\mathcal{D}}
   \Bigg \{
   \sum_{ \substack{ [n_0] = 0 \\ [\vec{n}] \neq \vec{0}}}
   + \sum_{ \substack{ [n_0] \neq 0 \\ [\vec{n}] = \vec{0}}}
   + \sum_{ \substack{ [n_0] \neq 0 \\ [\vec{n}] \neq \vec{0}}}
   \Bigg \}
   \mathcal{I}_{T,L}(\mathcal{D},n)
   \ .
  \label{eq:I-leading-0-vB} 
\end{gather}
In order to analyze the remaining terms,
we use the
following results about pure gauge fields.
\begin{proposition} \label{prop:pure-t}
   $\hat{\epsilon}_0(n_0) = 0$ if and only if $n_0$ is a pure gauge field. If
   $n_0$ is not a pure gauge field, then $\hat{\epsilon}_0(n_0) \ge 1$.
   \emph{(Corollary of theorem~\ref{theo:onehalf}.)}
\end{proposition}
\begin{proposition} \label{prop:pure-s}
   $\hat{\epsilon}_\text{s}(\vec{n}) = 0$ if and only if $\vec{n}$ is a pure
   gauge field. If $\vec{n}$ is not a pure gauge field, then
   $\hat{\epsilon}_\text{s}(\vec{n}) \ge 1$.
   \emph{(Theorem~\ref{theo:pure-s}.)}
\end{proposition}
A direct consequence of these facts and inequality~\eqref{eq:epsilon-3} is that,
if neither $n_0$ nor $\vec{n}$ are pure gauge fields, then $L \epsilon_r(0,\vec{n})
\ge L$,  $L \epsilon_r(n_0,\vec{0}) \ge T$ and $L \epsilon_r(n) \ge
\sqrt{T^2+L^2}$, which implies the following for the terms in the curly brakets in
eq.~\eqref{eq:I-leading-0-vB}
\begin{gather}
   \sum_{\mathcal{D}}
   \sum_{ \substack{ [n_0] = 0 \\ [\vec{n}] \neq \vec{0}}}
   \mathcal{I}_{T,L}(\mathcal{D},n)
   =
   O \left( e^{ - m L} \right)
   \ ,
   \label{eq:asympt-1}
   \\
   \sum_{\mathcal{D}}
   \sum_{ \substack{ [n_0] \neq 0 \\ [\vec{n}] = \vec{0}}}
   \mathcal{I}_{T,L}(\mathcal{D},n)
   =
   O \left( e^{ - m T} \right)
   \ ,
   \label{eq:asympt-2}
   \\
   \sum_{\mathcal{D}}
   \sum_{ \substack{ [n_0] \neq 0 \\ [\vec{n}] \neq \vec{0}}}
   \mathcal{I}_{T,L}(\mathcal{D},n)
   =
   O \left( e^{ - m \sqrt{T^2+L^2}} \right)
   \ .
   \label{eq:asympt-3}
\end{gather}
In the following, we neglect terms of order $O ( e^{ - m
\sqrt{T^2+L^2}} )$.

Focus first on the term in eq.~\eqref{eq:I-leading-0-vB} with $[n_0] \neq 0$
and $[\vec{n}] = \vec{0}$, i.e.~on th left-hand side of eq.~\eqref{eq:asympt-2}.  
From eqs.~\eqref{eq:F-Dn-1} and \eqref{eq:I-Dn} it is
evident that $\mathcal{F}_{T,L}(\mathcal{D},n)$, and hence
$\mathcal{I}_{T,L}(\mathcal{D},n)$, do not depend on $L$ if $\vec{n}=0$. In
particular, one can replace $L \to \infty$ in $\mathcal{I}_{T,L}(\mathcal{D},n)$
without changing its value. Combining this term with the OB contribution, we
define the finite-$T$ corrections as
\begin{gather}
   \Delta a(T, \infty)
   =
   \Delta a^\text{OB}_\text{t}(T)
   + 
   \sum_{\mathcal{D}}
   \sum_{ \substack{ [n_0] \neq 0 \\ [\vec{n}] = \vec{0}}}
   \mathcal{I}_{T,\infty}(\mathcal{D},n)
   \ .
   \label{eq:DeltaIt}
\end{gather}
From eqs.~\eqref{eq:asympt-0} and \eqref{eq:asympt-2}, it follows clearly that
\begin{gather}
    \Delta a(T, \infty)
   =
   O \left( e^{ - m T} \right)
   \ .
   \label{eq:asympt-4}
\end{gather}

We next turn to the term in eq.~\eqref{eq:I-leading-0} with $[n_0] = 0$ and
$[\vec{n}] \neq \vec{0}$, i.e.~the left-hand side of eq.~\eqref{eq:asympt-1}. This term depends on $L$ via
$\mathcal{F}_{T,L}(\mathcal{D},n)$ and on $T$ via the $x_0$ integral in
eq.~\eqref{eq:I-Dn}. Again, $\mathcal{F}_{T,L}(x_0|\mathcal{D},n)$ does not
depend on $T$ for $n_0=0$. By splitting the integrals in $x_0$ as
$\int_0^{\frac{T}{2}} = \int_0^{\infty} - \int_{\frac{T}{2}}^\infty$, and
replacing $T \to \infty$ in $\mathcal{F}_{T,L}(\mathcal{D},n)$, we readily get
\begin{gather}
   \sum_{\mathcal{D}}
   \sum_{ \substack{ [n_0] = 0 \\ [\vec{n}] \neq \vec{0}}}
   \mathcal{I}_{T,L}(\mathcal{D},n)
   =
   \Delta a(\infty, L) + \Delta R(T,L)
   \ , \\
   \Delta a(\infty, L)
   =
   \sum_{\mathcal{D}}
   \sum_{ \substack{ [n_0] = 0 \\ [\vec{n}] \neq \vec{0}}}
   \mathcal{I}_{\infty,L}(\mathcal{D},n)
   \label{eq:DeltaIs}
   \ , \\
   \Delta R(T,L)
   =
   \sum_{\mathcal{D}}
   \sum_{ \substack{ [n_0] = 0 \\ [\vec{n}] \neq \vec{0}}}
   \int_{\frac{T}{2}}^\infty dx_0 \ {\mathcal K}(x_0) \sum_{\mathcal{D}} \mathcal{F}_{\infty,L}(x_0|\mathcal{D},n)
   \label{eq:DeltaR}
   \ .
\end{gather}
We want to prove that the term $\Delta R(T,L)$ is subleading. A trivial
extension of the anaysis presented in section~\ref{subsec:diag:asymptotic} yields
\begin{gather}
   \ln \int_{\frac{T}{2}}^\infty dx_0 \ {\mathcal K}(x_0)
   \mathcal{F}_{T,L}(x_0|\mathcal{D},n)
   =
   - m L \epsilon'_r(n) + O(\ln L)
   \label{eq:saddleprime}
   \ , \\
   L \epsilon'_r(n)
   \ge
   \frac{
   T^2 \hat{\epsilon}'_{0}(n_0) + L^2 \hat{\epsilon}_\text{s}(\vec{n})
   }{
   \sqrt{T^2+L^2}
   }
   \ ,
   \\
   \hat{\epsilon}'_{0}(n_0) = 
   \min_{\substack{x_0 \text{ with } x(a)=0 \\ x_0(b) \in [1/2,\infty)}}
   \sum_{\ell \in \mathcal{L}} \left| \delta x_0(\ell) + n_0(\ell) \right|
   \ .
\end{gather}
Note that $\hat{\epsilon}'_{0}(n_0)$ differs from
$\hat{\epsilon}_{0}(n_0)$ for the domain over which the minimum is taken,
and this reflects the fact that the $x_0$ integral in $\Delta R(T,L)$ runs over
$[T/2,\infty)$. We specialize to the case of interest, $n_0=0$ and $[\vec{n}]
\neq \vec{0}$. Thanks to proposition~\ref{prop:pure-s},
$\hat{\epsilon}_\text{s}(\vec{n}) \ge 1$. Moreover, as we will see in a moment,
$\hat{\epsilon}'_{0}(0) \ge 1$, implying
\begin{gather}
   L \epsilon'_r(n)
   \ge
   \sqrt{T^2 + L^2}
   \ .
   \label{eq:boundprime}
\end{gather}
To prove the inequality $\hat{\epsilon}'_{0}(0) \ge 1$ note that, since
$\mathcal{D}$ is 1PI, two disjoint paths $P_1$ and $P_2$ from $a$ to $b$ exist.
Then, by using the triangular inequality,
\begin{gather}
   \sum_{\ell \in \mathcal{L}} \left| x_0[f(\ell)] - x_0[i(\ell)] \right|
   \ge
   \sum_{j=1,2} \sum_{\ell \in P_j} \left| x_0[f(\ell)] - x_0[i(\ell)] \right|
   \ge
   2 \left| x_0(b) - x_0(a) \right|
   \ ,
\end{gather}
we deduce
\begin{gather}
   \hat{\epsilon}'_{0}(0)
   \ge
   \min_{x_0(b) \in [1/2,\infty)} 2|x_0(b)|
   =
   1
   \ .
\end{gather}
Finally, using inequality \eqref{eq:boundprime} in eqs.~\eqref{eq:DeltaR} and
\eqref{eq:saddleprime}, we obtain
\begin{gather}
   \Delta R(T,L)
   =
   O\left( e^{ - m \sqrt{T^2+L^2}} \right)
   \ ,
\end{gather}
which is of the same order of terms that we have already neglected in
eq.~\eqref{eq:I-leading-0}. We also note that eq.~\eqref{eq:asympt-1} implies
\begin{gather}
   \Delta a(\infty,L)
   =
   O\left( e^{ - m L} \right)
   \ .
\end{gather}
This completes our demonstration of items 1 and 2 as listed in the Introduction, and in particular of eqs.~\eqref{eq:stat:separation} and \eqref{eq:stat:LTscaling}.

\subsection{Leading finite-\texorpdfstring{$L$}{L} corrections}
\label{subsec:diag:leading-L}

We come now to the task of isolating and characterizing the leading exponentials
in $\Delta a(\infty, L)$, defined in eq.~\eqref{eq:DeltaIs} %
\begin{gather}
   \Delta a(\infty, L)
   =
   \sum_{\mathcal{D}}
   \sum_{ \substack{ [n_0] = 0 \\ [\vec{n}] \neq \vec{0}}}
   \mathcal{I}_{\infty,L}(\mathcal{D},n)
   \ .
   \nonumber
\end{gather}
In the following analysis, a special role will be played by the following class
of gauge fields and orbits.

\textit{Gauge fields localized on a line. Simple gauge fields and orbits.} Given
a line $\ell$, the gauge field $\vec{n}$ is said to be localized on $\ell$ if
and only if it is nonzero on $\ell$ and zero on any other line. The gauge field
$\vec{n}$ and the gauge orbit $[\vec{n}]$ are said to be simple if and only if
$\vec{n}$ is gauge equivalent to a gauge field which is localized on a line.

\begin{proposition} \label{prop:sqrt}
   If $1 \le \hat{\epsilon}_\text{s}(\vec{n}) < \sqrt{2+\sqrt{3}}$ then
   $\vec{n}$ is simple. \emph{(Theorem~\ref{theo:sqrt}.)}
\end{proposition}
Recall that $\hat{\epsilon}_\text{s}(\vec{n}) <1$ if and only if $\vec{n}$ is
pure (proposition~\ref{prop:pure-s}). A direct consequence of these facts and
inequality~\eqref{eq:epsilon-3} is that, if $[\vec{n}] \neq \vec{0}$ is not
simple then $L \epsilon_r(0,\vec{n}) \ge L \sqrt{2+\sqrt{3}}$. By restricting
the sum over $[\vec{n}] \neq \vec{0}$ to the set of simple gauge orbits, we
obtain
\begin{align}
   \Delta a(\infty, L)
   & =
   \Delta a_{\text s}(L)
   + O \left( e^{ - \sqrt{2+\sqrt{3}} m L } \right)
   \, , \\
    \Delta a_{\text s}(L) & =   \sum_{\mathcal{D}}
   \sum_{ [\vec{n}] \text{ simple} }
   \mathcal{I}_{\infty,L}(\mathcal{D},n) \Big|_{n_0=0}  \,.
      \label{eq:leading-L-simple-0}
\end{align}
This formula is not yet useful in practice. Ideally, we would like
to replace the sum over simple gauge orbits with the sum over the gauge fields
localized on a line. However, this would lead to some double counting, since
gauge fields localized on different lines may still be gauge equivalent to one
another. It turns out that gauge equivalence of gauge fields localized on a
single line can be understood in geometrical terms. We introduce a notion of
equivalence between lines:

\textit{$s$-equivalence.} The lines $\ell_1$ and $\ell_2$ are said to be
$s$-equivalent if and only if every loop containing $\ell_1$ contains also $\ell_2$ and
vice versa. One proves that this is indeed an equivalence relation (theorem
\ref{theo:s-equivalence}). The $s$-equivalence classes are denoted by
$[\ell]_\text{s}$.

The following proposition provides the exact relation between gauge fields
localized on a single line, up to a gauge transformation, and the concept of
$s$-equivalence.

\begin{proposition} \label{prop:s-localization}
   With no loss of generality, the orientation of the lines of $\mathcal{G}$ can
   be chosen in such a way that, if $\ell$ and $\ell'$ are $s$-equivalent and
   $C$ is a loop that contains both, then $[C:\ell]=[C:\ell']$, i.e.~either both
   $\ell$ and $\ell'$ have the same orientation as $C$, or they both have the
   opposite orientation of $C$. This choice of orientation is assumed
   here.\\[1mm]
   Let $\vec{n}$ and $\vec{n}'$ be gauge fields localized on $\ell$ and $\ell'$
   respectively. $\vec{n}$ and $\vec{n}'$ are gauge equivalent if and only if
   $\ell$ and $\ell'$ are $s$-equivalent and $\vec{n}(\ell)=\vec{n}'(\ell')$.
   \emph{(Theorem~\ref{theo:s-localization}.)}
\end{proposition}

A consequence of this proposition is that the simple gauge orbits are in
one-to-one correspondence with the pairs $([\ell],\vec{u})$ with $\vec{u} \in
\mathbb{Z}^3$, in the following way. Given a pair $([\ell],\vec{u})$ the
corresponding gauge orbit $[\vec{n}]$ is constructed by considering the gauge field
$\vec{n}$ localized on $\ell$ with $\vec{n}(\ell)=\vec{u}$. Therefore the sum in
eq.~\eqref{eq:leading-L-simple-0} can be represented as follows
\begin{gather}
     \Delta a_{\text s}(L) 
   =
   \sum_{\mathcal{D}}
   \sum_{[\ell]_\text{s}} \sum_{\vec{u} \in
   \mathbb{Z}^3}
   \mathcal{I}_{\infty,L}(\mathcal{D},n) \Big|_{n(\ell') = (0,\vec{u}) \delta_{\ell,\ell'}}
   \label{eq:DeltaIs-simple}
   \ .
\end{gather}
This completes our demonstration of items 3 and 4 as listed in the Introduction, and gives a precise definition to $\Delta a_{\text s}(L)$.

\subsection{Leading finite-\texorpdfstring{$T$}{T} corrections}
\label{subsec:diag:leading-T}

We come now to the task of isolating and characterizing the leading exponentials
in $ \Delta a(T, \infty) $, using eq.~\eqref{eq:DeltaIt} \begin{gather}
 \Delta a(T, \infty) - \Delta a^{\text{OB}}_{\text t}(T)  =  \sum_{\mathcal{D}}
   \sum_{ \substack{ [n_0] \neq 0 \\ [\vec{n}] = \vec{0}}}
   \mathcal{I}_{T,\infty}(\mathcal{D},n)
   \ . \label{eq:isolatingBPandWP}
\end{gather}

\begin{proposition} \label{prop:epshat-t-one}
   If $n_0$ is not pure then either $\hat{\epsilon}_{0}(n_0) = 1$ or else
   $\hat{\epsilon}_{0}(n_0) \ge \frac{3}{2}$. \emph{(Corollary of
   theorem~\ref{theo:onehalf}.)}
\end{proposition}

We use this proposition and eq.~\eqref{eq:epsilon-3} to estimate the error made
when the sum over $n_0$ is restricted to the contributions with
$\hat{\epsilon}_{0}(n_0)=1$,
\begin{gather}
   \sum_{\mathcal{D}}
   \sum_{ \substack{ [n_0] \neq 0 \\ [\vec{n}] = \vec{0}}}
   \mathcal{I}_{T,\infty}(\mathcal{D},n)
   =
   \sum_{\mathcal{D}}
   \sum_{ \substack{ [n_0] \text{ with } \hat{\epsilon}_{0}(n_0)=1 \\ [\vec{n}] = \vec{0}}}
   \mathcal{I}_{T,\infty}(\mathcal{D},n)
   +
   O \left( e^{- \frac{3}{2} mT} \right)
   \ .
   \label{eq:leading-T-0}
\end{gather}
The analysis of the leading exponentials is more involved in this case, because
some non-simple gauge fields contribute as well. We generalize the definition of
localized gauge fields.

\textit{Localized gauge fields. Localizable and simple gauge fields and orbits.}
Given a subset of lines $A \subset \mathcal L$, the gauge field $n_0$ is said to
be localized on $A$ if and only it is zero for each line in $\mathcal{L}
\setminus A $, and non-zero on each line in $A$. The gauge field $n_0$ and the
gauge orbit $[n_0]$ are said to be localizable on $A$ if and only if $n_0$ is gauge
equivalent to a gauge field which is localized on $A$. The gauge field $n_0$ and
the gauge orbit $[n_0]$ are said to be simple if and only if they are
localizable on a single line.

\begin{proposition} \label{prop:one-0}
   If $\hat{\epsilon}_0(n_0) = 1$, then one of the following two possibilities
   is realized:
   \begin{enumerate}[topsep=0pt,itemsep=1ex]
      \item (\textit{Type-1} gauge fields.) Up to a gauge transformation, $n_0$
      is localized on a line $\ell$ and $|n_0(\ell)| = 1$ (in particular, $n_0$
      is simple).
      \item (\textit{Type-2} gauge fields.) Up to a gauge transformation, $n_0$
      is localized on a cut-set $S=\{\ell_1,\ell_2\}$ which disconnects $a$ and
      $b$. Assuming that, with no loss of generality, $i(\ell_{1,2})$ is
      connected to $a$ in $\mathcal G - S$, then $n_0(\ell_{1,2}) = -1$.
   \end{enumerate}
   \emph{(Theorem~\ref{theo:one-0}.)}
\end{proposition}

The first crucial observation is that type-1 and type-2 gauge orbits are distinct,
i.e.~no type-1 gauge field is gauge equivalent to a type-2 gauge field and vice
versa. This fairly obvious statement follows from theorem~\ref{theo:simple-0b}.
The sum over $[n_0]$ with $\hat{\epsilon}_{0}(n_0)=1$ can thus be split
into two pieces, schematically
\begin{gather}
   \sum_{[n_0] \text{ with } \hat{\epsilon}_{0}(n_0)=1}
   =
   \sum_{[n_0] \text{ is type-1}}
   + \sum_{[n_0] \text{ is type-2}}
   \ .
\end{gather}
A less obvious fact is that all type-2 gauge fields are gauge equivalent.
This follows from theorem~\ref{theo:simple-0c}. Therefore two exclusive
possibilities are given:
\begin{enumerate}[topsep=0pt,itemsep=1ex]
   \item $\mathcal{G}$ is 2PI between $a$ and $b$, i.e.~no cut-set with two
   elements exists. In this case no type-2 gauge orbit exists.
   \item $\mathcal{G}$ is two-particle reducible (2PR) between $a$ and $b$. In this case
   exactly one type-2 gauge orbit exists, and the sum over the type-2 gauge
   orbits reduces to only one term, for instance
   \begin{gather}
      \sum_{ \substack{ [n_0] \text{ is type-2} \\ [\vec{n}] = \vec{0}}}
      \mathcal{F}_{T,\infty}(x_0 \vert \mathcal{D},n)
      =
      \mathcal{F}_{T,\infty}(x_0 \vert \mathcal{D},n) \Big|_{n_\mu(\ell') = - (\delta_{\ell_1,\ell'} + \delta_{\ell_2,\ell'}) \delta_{\mu,0} }
      \label{eq:leading:type2-0}
      \ ,
   \end{gather}
   where $S=\{\ell_1,\ell_2\}$ is a cut-set as in proposition~\ref{prop:one-0}.\footnote{Note here we have switched to a discussion in terms of $\mathcal F_{T,L}(x_0 \vert \mathcal D, n)$ rather than $\mathcal I_{T,L}( \mathcal D, n)$. Recall that these are related according to eq.~\eqref{eq:I-Dn}, in particular that the former defines a contribution to $G(x_0 \vert T, L)$ and the latter a contribution to  $\amu$.}
   The right-hand side of this equation can be dramatically simplified with the
   following trick. Let $\mathcal{V}_b$ be the set of vertices connected to $b$
   in $\mathcal{G}-S$. By changing variables $x_0(v) \to x_0(v) + L_0$ for $v
   \in \mathcal{V}_b$ in the integrals which define
   $\mathcal{F}_{T,L}(x_0|\mathcal{D},n)$ in eq.~\eqref{eq:F-Dn-1}, one easily
   proves that
   \begin{gather}
      \mathcal{F}_{T,L}(x_0|\mathcal{D},n) \Big|_{n_\mu(\ell') = - (\delta_{\ell_1,\ell'} + \delta_{\ell_2,\ell'}) \delta_{\mu,0} }
      =
      \mathcal{F}_{T,L}(x_0-T|\mathcal{D},0)
      =
      \mathcal{F}_\infty(x_0-T|\mathcal{D})
      \label{eq:leading:type2-1}
      \ ,
   \end{gather}
   where we have used that $\mathcal{F}_{T,L}(x_0-T|\mathcal{D},0)$ does
   not depend on $L$, and depends on $T$ only through the argument, $x_0 - T$.
\end{enumerate}
When combining eqs.~\eqref{eq:leading:type2-0} and \eqref{eq:leading:type2-1},
and integrating over the kernel, one obtains
\begin{gather}
  \sum_{\mathcal{D}} \sum_{ \substack{ [n_0] \text{ is type-2} \\ [\vec{n}] = \vec{0}}}
  \mathcal{I}_{T,\infty}(\mathcal{D},n)
  =
  \int_0^{\frac{T}{2}} dx_0 \ {\mathcal K}(x_0) \sum_{\substack{\mathcal{D} \text{ 2PR} \\ \text{between $a$ and $b$}}} \mathcal{F}_\infty(x_0-T|\mathcal{D})
  \ .
  \label{eq:leading:type2-2}
\end{gather}
Note that, if the diagram $\mathcal{D}$ is 2PI between $a$ and $b$, then
$\mathcal{F}_\infty(x_0-T|\mathcal{D}) \le O\left(e^{- \frac{3}{2} mT}\right)$
if $x_0 \in [0,T/2]$, since, in this case, at least three pions propagate in between the
two electromagnetic currents. Therefore the restriction to 2PR diagrams in
eq.~\eqref{eq:leading:type2-2} can be lifted up to a term of an order we are
already neglecting, and the sum over all diagrams reconstructs the
infinite-volume two-point function, $G(x_0-T|\infty)$, inside the integrand,
yielding
\begin{gather}
   \sum_{\mathcal{D}} \sum_{ \substack{ [n_0] \text{ is type-2} \\ [\vec{n}] = \vec{0}}}
   \mathcal{I}_{T,\infty}(\mathcal{D},n)
   =
   \Delta a_\text{t}^\text{BP}(T) + O\left(e^{- \frac{3}{2} m T}\right)
   \ ,
   \label{eq:leading:type2-3}
\end{gather}
where we have introduced the backpropagating-pion (BP) contribution, defined in eq.~\eqref{eq:stat:DeltaIt-BP} above, \begin{gather}
\nonumber
   \Delta a_\text{t}^\text{BP}(T)
   =
   \int_0^{\frac{T}{2}} dx_0 \ {\mathcal K}(x_0) G(T-x_0|\infty)
   \ .
\end{gather}
In this equation we have also used the fact that $G(x_0|\infty)$ is even in $x_0$.

The sum over type-1 gauge orbits is manipulated in a very similar way to the
spatial case. The relevant notion of equivalence between lines needs to be
slightly modified to account for the different structure of gauge
transformations for the temporal component.

\textit{$t$-equivalence.} The lines $\ell_1$ and $\ell_2$ are said to be
$t$-equivalent if and only if every loop containing $\ell_1$ contains also
$\ell_2$ and vice versa, \textit{and} every path from $a$ to $b$ containing
$\ell_1$ contains also $\ell_2$ and vice versa. One proves that this is indeed
an equivalence relation (theorem \ref{theo:t-equivalence}). The $t$-equivalence
classes are denoted by $[\ell]_\text{t}$.

The following proposition provides the exact relation between gauge fields
localized on a single line up to a gauge transformation and the concept of
$t$-equivalence.

\begin{proposition} \label{prop:t-localization}
   The orientation of the lines of $\mathcal{G}$ can be chosen in such a way
   that, if $\ell$ and $\ell'$ are $t$-equivalent and $C$ is a loop that
   contains both, then $[C:\ell]=[C:\ell']$, i.e.~either both $\ell$ and $\ell'$
   have the same orientation as $C$, or they both have the opposite orientation
   of $C$. This choice of orientation is assumed here.\\[1mm]
   Let $n_0$ and $n'_0$ be gauge fields localized on $\ell$ and $\ell'$
   respectively. $n_0$ and $n'_0$ are gauge equivalent if and only if $\ell$ and
   $\ell'$ are $t$-equivalent and $n_0(\ell)=n'_0(\ell')$.
   \emph{(Theorem~\ref{theo:t-localization}.)}
\end{proposition}

We refer to the sum over type-1 gauge orbits in eq.~\eqref{eq:leading-T-0}
as the wrapped-pion (WP) contribution. As for the spatial directions, this sum can
be represented as follows
\begin{gather}
    \Delta a_\text{t}^\text{WP}(T)
   =
   \sum_{\substack{[n_0] \text{ is type-1} \\ [\vec{n}] = \vec{0} }} \mathcal{I}_{\infty,T}(x_0|\mathcal{D},n)
   =
   \sum_{\mathcal{D}} \sum_{[\ell]_\text{t}} \sum_{\alpha= \pm 1} \mathcal{I}_{\infty,T}(x_0|\mathcal{D},n)
   \Big|_{n(\ell') = (\alpha,\vec{0}) \delta_{\ell,\ell'}}
   \label{eq:DeltaIt-WP}
   \ .
\end{gather}

Combining all equations in this subsection gives our final decomposition for the leading finite-$T$ corrections to $\Delta a(T,L)$, eq.~\eqref{eq:stat:DeltaIt},
\begin{gather}
   \Delta a_\text{t}(T) = \Delta a_\text{t}^\text{OB}(T) + \Delta a_\text{t}^\text{BP}(T) + \Delta a_\text{t}^\text{WP}(T) 
   \ , \nonumber
\end{gather}
providing also the proof of eq.~\eqref{eq:stat:Dat}, i.e.~that the difference between $\Delta a_\text{t}(T) $ and $\Delta a_\text{t}(T) $ scales as $O(e^{-\frac32 m T})$. At this stage we have demonstrated all 8 items listed in the Introduction.
 
\subsection{Differences in the expansions of $G(x_0 \vert T, L)$ and $\amu(T,L)$}%
\label{subsec:diag:2pt-function}

Before describing the skeleton expansion required to bring the preceding decompositions into a useful form, we comment here on the differences between finite-volume
corrections of the electromagnetic-current two-point function, $G(x_0|T,L)$, and the integral defining $\amu(T,L)$.

We stress that a term in the expansion of $\amu(T,L)$ with a given scaling (e.g.~$e^{- m L}$ or $e^{- m T}$) \textit{cannot}, in general, be determined by identifying the contribution to $G(x_0|T,L)$ with the same scaling, and then integrating with the kernel.
 This is simply because values of $x_0$
that scale proportionally to $T$ contribute to the integral so that
for any given diagram, $\mathcal D$, and gauge orbit, $n$, the two contributions, $\mathcal I_{T,L}(\mathcal D, n)$ and $\mathcal F_{T,L}(x_0 \vert \mathcal D, n)$, can have different asymptotic behavior.
It is for this reason that in
section~\ref{subsec:diag:asymptotic} the analysis of the saddle-point was
performed directly on the full integral. {As we have already noted, the fact
that the divergence of the kernel as $x_0 \to +\infty$ is bounded by a positive
power of $x_0$, implies that the location of the saddle-point does not depend on
the kernel itself. %
For any choice of kernel with power-like asymptotic behavior, the same class of terms will contribute. }

By contrast, to identify corrections to $G(x_0 \vert T, L)$ directly,
one can repeat the entire analysis %
as described for $\amu$ but with a Dirac %
delta function in place of $\mathcal K(x_0)$. This %
does modify the location of the saddle point, such that some terms
that survive with a power-like kernel turn out to be subleading in this special case. We
give here only the final results, leaving to the reader the task to sort out the
details.

The first key result from a direct analysis of the correlation function is eq.~\eqref{eq:stat:Gexp}, 
\begin{align}
   \Delta G(x_0|T,L)
   & = \Delta G(x_0|\infty,L) + \Delta G(x_0|T,\infty) 
   + O\left(e^{-m\sqrt{L^2+T^2}}\right)
   \,,  \nonumber
\end{align}
i.e.~that a separation formula holds for $\Delta G(x_0 \vert T, L)$ directly matching that for $\amu(T,L)$ [eq.~\eqref{eq:stat:separation}].
We stress that, in this expression, $x_0$ is kept fixed (i.e.~not scaled with
the volume).

The finite-$L$ corrections to the two-point functions are obtained
by replacing the kernel with a delta function in eq.~\eqref{eq:DeltaIs-simple},
yielding eq.~\eqref{eq:stat:DGsexp},
\begin{equation}
  \Delta G(x_0| \infty, L) =   \Delta G_\text{s}(x_0|L)  + O \left( e^{ - \sqrt{2+\sqrt{3}} m L } \right) \,, \nonumber
\end{equation}
with $  \Delta G_\text{s}(x_0|L)$ defined as
\begin{gather}
    \Delta G_\text{s}(x_0|L)
   =
   \sum_{\mathcal{D}}
   \sum_{[\ell]_s} \sum_{\vec{u} \in
   \mathbb{Z}^3}
   \mathcal{F}_{\infty,L}(\mathcal{D},n) \Big|_{n(\ell') = (0,\vec{u}) \delta_{\ell,\ell'}}
   \label{eq:DeltaGs-simple}
   \ .
\end{gather}

Turning to the finite-$T$ corrections, when the same analysis is performed for the out-of-the-box \eqref{eq:stat:DeltaIt-OB} and backpropagating-pion  \eqref{eq:stat:DeltaIt-BP} contributions, one finds that the former is zero while the latter contributes at $O(e^{-2mT})$, and should thus be dropped to the order we work.
Thus, the only leading finite-$T$ scaling comes from
the wrapped-pion contribution of eq.~\eqref{eq:DeltaIt-WP}. This yields eq.~\eqref{eq:stat:DGtexp}
\begin{equation}
  \Delta G(x_0| T, \infty ) =   \Delta G_\text{t}(x_0 | T)  + O \left( e^{ - 2 m T } \right) \,, \nonumber
\end{equation}
with $ \Delta G_\text{t}(x_0|T)$ defined as
\begin{gather}
    \Delta G_\text{t}(x_0|T)
   =
   \sum_{\mathcal{D}} \sum_{[\ell]_t} \sum_{\alpha= \pm 1} \mathcal{F}_{T,\infty}(x_0|\mathcal{D},n)
   \Big|_{n(\ell') = (\alpha,\vec{0}) \delta_{\ell,\ell'}}
   \label{eq:DeltaGt-simple}
   \ .
\end{gather}
We stress that the new analysis shows that the neglected contributions are of order $e^{ -
2 m T }$ in this case and not $e^{ - \frac{3}{2} m T }$ as for $\Delta a(T,\infty)$.
This is intuitive as the half-integer
factor in the scaling comes from integrating
 up to $T/2$. Since this distinction does not affect any conclusion drawn in
this paper, we will omit its proof. %

We close the subsection by noting that our decomposition of the leading finite-$L$ correction to $\amu(T,L)$  [$\Delta a_{\text s}(L)$ in eq.~\eqref{eq:DeltaIs-simple}] and that of $G(x_0 \vert T, L)$  [$\Delta G_{\text s}(x_0 \vert L)$ in eq.~\eqref{eq:DeltaGs-simple}], leads to the simple relation expressed by eq.~\eqref{eq:stat:DeltaIs}
\begin{gather}
   \Delta a_\text{s}(L) = \int_0^\infty dx_0 \  \mathcal{K}(x_0) \, \Delta G_\text{s}(x_0|L) \nonumber
   \,.
\end{gather}
And similarly, while this correspondence fails in general for the finite-$T$ expressions, it does hold for the wrapped-pion. Combining eqs.~\eqref{eq:DeltaIt-WP} and \eqref{eq:DeltaGt-simple} yields 
\eqref{eq:stat:DeltaIt-WP}
\begin{gather}
   \Delta a_\text{t}^\text{WP}(T)
   =
   \int_0^{\frac{T}{2}} dx_0 \ \mathcal{K}(x_0) \, \Delta G_\text{t}(x_0|T) \nonumber
   \ .
\end{gather}

\subsection{Skeleton expansion}
\label{subsec:diag:skeleton}

\newcommand{\wild}{\eta}

We now derive more useful expressions for $\Delta G_{\text s}(x_0 \vert L)$ and $\Delta G_{\text t}(x_0 \vert T)$ [defined in eqs.~\eqref{eq:DeltaGs-simple} and \eqref{eq:DeltaGt-simple} respectively] in terms of proper vertices and a dressed propagators.
The two quantities can be written in a unified form 
\begin{gather}
   \Delta G_\wild(x_0 \vert T, L)
   =
   \sum_{\mathcal{D}} \sum_{[\ell]_\wild}
   \sum_{u \in U_\wild} \mathcal{F}_{T,L}(x_0|\mathcal{D},n)
   \bigg|_{n(\ell') = u \delta_{\ell',\ell}}
   \label{eq:skeleton:S0}
   \ , \qquad \qquad
   \text{for } \wild = \text{s}, \text{t} \ ,
\end{gather}
with the definitions
\begin{gather}
   U_\text{s} =  \big \{ (0,\vec{u}) \ \vert \ \vec{u} \in \mathbb{Z}^3 \setminus \{0\} \big \}
   \ , \qquad 
   U_\text{t} = \big \{ (-1,\vec{0}) , (+1,\vec{0}) \big \}
   \ .
\end{gather}

The aim is to evaluate the sums over diagrams and
equivalence classes, in order to obtain a compact representation for $\Delta G_\wild(x_0 \vert T, L)$, independent of the details of the effective-field theory. %
We begin with an
overview of the main logical steps. First it is convenient to define an auxiliary quantity, obtained
from the right-hand side of eq.~\eqref{eq:skeleton:S0} by replacing the sum over the
equivalence classes with a sum over all lines $\ell$,
\begin{gather}
   \Delta G^{\mathcal L}_\wild(x_0 \vert T, L)
   =
   \sum_{\mathcal{D}} \sum_{\ell \in \mathcal{L}}
   \sum_{u \in U_\wild} \mathcal{F}_{T,L}(x_0|\mathcal{D},n)
   \bigg|_{n(\ell') = u \delta_{\ell',\ell}}
   \label{eq:skeleton:Sprime}
   \, ,
\end{gather}
where the superscript indicates that the sum runs over the set $\mathcal L$.
Each term in $\Delta G_\wild(x_0 \vert T, L)$ is counted $N_\wild(\ell)$ many times in
$\Delta G^{\mathcal L}_\wild(x_0 \vert T, L)$, where $N_\wild(\ell)$ is the number of
elements in the equivalence class $[\ell]_\wild$. In fact, the sum over the
equivalence classes in $\Delta G_\wild(x_0 \vert T, L)$ can be equivalently replaced by a
sum over all lines $\ell$ divided by $N_\wild(\ell)$, i.e.
\begin{gather}
   \Delta G_\wild(x_0 \vert T, L)
   =
   \sum_{\mathcal{D}} \sum_{\ell \in \mathcal{L}} \frac{1}{N_\wild(\ell)}
   \sum_{u \in U_\wild} \mathcal{F}_{T,L}(x_0|\mathcal{D},n)
   \bigg|_{n(\ell') = u \delta_{\ell',\ell}}
   \label{eq:skeleton:S1}
   \ .
\end{gather}
As we will see, the auxiliary quantity $\Delta G^{\mathcal L}_\wild(x_0 \vert T, L)$ has a
straightforward skeleton expansion in terms of proper vertices and dressed
propagators. By analysing the various terms of the expansion, one can
identify the multiplicity factor, $N_\wild(\ell)$, and divide it out by hand,
obtaining in this way the desired skeleton expansion for $\Delta G_\wild(x_0 \vert T, L)$.

Before jumping into the core of the calculation, we find convenient to switch to
the momentum representation of the Feynman integrals. This proceeds as in the
 infinite-volume case, with only small modifications. We introduce the Fourier
transform of $\mathcal{F}_{T,L}(x_0|\mathcal{D},n)$
\begin{gather}
   \widetilde{\mathcal{F}}_{T,L}(k_0|\mathcal{D},n) =  \int_{-\infty}^{\infty} d x_0 \, e^{- i k_0 x_0} \,  \mathcal{F}_{T,L}(x_0|\mathcal{D},n) \,,
\end{gather}
and, analogously, the Fourier transforms $\Delta \widetilde{G}_\wild(k_0 \vert T, L)$
and $\Delta \widetilde{G}^{\mathcal L}_\wild(k_0 \vert T, L)$, of $\Delta G_\wild(x_0 \vert T, L)$
and $\Delta G^{\mathcal L}_\wild(x_0 \vert T, L)$ respectively.

Let us examine the structure of
$\widetilde{\mathcal{F}}_{T,L}(k_0|\mathcal{D},n)$. A momentum four-vector, $p(\ell)$, is
associated to each line $\ell$. Each infinite-volume propagator in
eq.~\eqref{eq:F-Dn-1} is represented as
\begin{gather}
   \Delta \big (\delta x (\ell)+ \L n(\ell) \big)
   =
   \int \frac{d^4p(\ell)}{(2\pi)^4}
   e^{ip(\ell) \cdot \delta x (\ell)}
   \frac{
   e^{ip(\ell) \cdot \L n(\ell)}
   }{
   p(\ell)^2 + m^2
   }
   \ .
\end{gather}
The integrals over the coordinates in eq.~\eqref{eq:F-Dn-1} are calculated
explicitly, yielding
\begin{gather}
   \widetilde{ \mathcal{F}}_{T,L}(k_0|\mathcal{D},n) = 
   \int \left[ \prod_{\ell \in \mathcal{L}} \frac{d^4p(\ell)}{(2\pi)^4}
   \frac{
   e^{ip(\ell) \cdot \L n(\ell)}
   }{
   p(\ell)^2 + m^2
   }
   \right]
   \widetilde{\mathbb{V}}(p,k_0)
   \label{eq:skeleton:Fmom}
   \ ,
\end{gather}
where $\widetilde{\mathbb{V}}(p,k_0)$ is the product of the vertex functions,
which are (volume-independent) polynomials of the momenta, times the delta of
momentum conservation at each vertex. If $n=0$, this is nothing but the standard
infinite-volume Feynman integral in momentum space associated to $\mathcal{D}$.

We use the representation~\eqref{eq:skeleton:Fmom} in the definition of
$\Delta \widetilde{G}^{\mathcal L}_\wild(k_0 \vert T, L)$, obtained by taking the Fourier transform of
eq.~\eqref{eq:skeleton:Sprime}. Since the gauge fields, $n$, that contribute to
$\Delta \widetilde{G}^{\mathcal L}_\wild(k_0 \vert T, L)$ are localized on a single line $\ell$, all
propagators in the Feynman integral must be the infinite-volume ones except for
the one in $\ell$ which needs to be replaced according to the rule
\begin{gather}
   \frac{
   1
   }{
   p^2 + m^2
   }
   \ \rightarrow \ 
   K_\wild(p)
   =
   \sum_{u \in U_\wild}
   \frac{
   e^{i u \L p }
   }{
   p^2 + m^2
   }
   =
   \sum_{u \in U_\wild}
   \frac{
   \cos \left( u \L p \right)
   }{
   p^2 + m^2
   }
   \label{eq:skeleton:Kdef}
   \ ,
\end{gather}
where we have used the fact that the set, $U_\wild$, is invariant under $u \to
-u$. Therefore the sum $\Delta \widetilde{G}^{\mathcal L}_\wild(k_0 \vert T, L)$ can be
written as
\begin{gather}
   \Delta \widetilde{G}^{\mathcal L}_\wild(k_0 \vert T, L)
   =
   \sum_{\mathcal{D}} 
   \int \bigg[ \prod_{\ell' \in \mathcal{L}} \frac{d^4p(\ell')}{(2\pi)^4} \bigg]
   \widetilde{\mathbb{V}}(p,k_0)
   \sum_{\ell \in \mathcal{L}}
   \bigg[ \prod_{\ell' \neq \ell}
   \frac{1}{p(\ell')^2 + m^2}
   \bigg]
   K_\wild(p(\ell))
   \label{eq:skeleton:Sprime-2}
   \ .
\end{gather}

The operation of replacing the propagator in each line with a modified
propagator can be understood in terms of a simple deformation of the action. Let
$\pi_q(x)$ be the pion field with charge $q=0,\pm 1$. We use the notation
$\bar{q}=-q$ such that fields satisfy $\pi_q(x)^* = \pi_{\bar{q}}(x) =
\pi_{-q}(x)$. Expressing the infinite-volume action of our general effective theory as
\begin{gather}
   S(\pi) = \frac{1}{2} \sum_q \int \frac{d^4 p}{(2\pi)^4} \, (p^2+m^2) \left|\widetilde{\pi}_{q}(p) \right|^2  + S_\text{I}(\pi)
   \ ,
\end{gather}
we add a source term to define
\begin{gather}
   S(\pi,K)
   =
   S(\pi) - \frac{1}{2} \sum_q \int \frac{d^4 p}{(2 \pi)^4} \, (p^2+m^2)^2 K_\wild(p) \left|\widetilde{\pi}_{q}(p) \right|^2
   \ .
\end{gather}
The modified action, $S(\pi,K)$, has the same vertices as the original action,
$S(\pi)$, while the canonical propagator is replaced by a more complicated
function that depends on $K_\wild(p)$. Here we do not need the general form, but
only note that the source term is designed in such a way that, at leading order
in $K_\wild(p)$, the new propagator is simply
\begin{gather}
  \frac{1}{p^2 + m^2} \ \rightarrow \ \frac{1}{p^2 + m^2} + K_\wild(p) + O(K^2) \ .
\end{gather}
When this replacement rule is used in a Feynman diagram, the $O(K)$ term is
obtained by replacing the propagator with $K_\wild(p)$ in one line at a time, and
then summing over \textit{all} lines. For instance, let $\langle j_\rho(x)
j_\rho(0) \rangle^{[K]}$ be the infinite-volume expectation value of $j_\rho(x)
j_\rho(0)$, with the action $S(\pi,K)$, and define 
\begin{equation}
\widetilde G^{[K]}(k_0 \vert \infty) =  - \frac{1}{3} \sum_{\rho=1}^3 \int d^4 x \, e^{ik_0x_0} \langle j_\rho(x) j_\rho(0) \rangle^{[K]} \,.
\end{equation}
Recalling that
$\widetilde{\mathcal{F}}_{T,L}(k_0|\mathcal{D},0) =
\widetilde{\mathcal{F}}_\infty(k_0|\mathcal{D})$ is the general Feynman integral
contributing to the infinite-volume limit of the current two-point function
given in eq.~\eqref{eq:stat:G-def} (up to the Fourier transfrom), one finds
\begin{align}
\widetilde G^{[K]}(k_0 \vert \infty) %
   &  =
   \sum_{\mathcal{D}}  \widetilde{ \mathcal{F}}_\infty(k_0|\mathcal{D})  \bigg \vert_{(p^2+m^2)^{-1} \to (p^2+m^2)^{-1} + K_\wild(p) + O(K^2)} \,,
   \nonumber \\
   &=   
   \widetilde G(k_0 \vert \infty)
   + \Delta \widetilde{G}^{\mathcal L}_\wild(k_0 \vert T, L)
   + O(K^2)
   \label{eq:skeleton:jjK-1}
   \ ,
\end{align}
with $\widetilde G(k_0 \vert \infty)$ defined as the Fourier transform of $G(x_0 \vert \infty)$, equivalently as $\widetilde G^{[K=0]}(k_0 \vert \infty) $. %
In words, we have found that the $O(K)$ term
of the current two-point function in the modified theory is nothing but
$\Delta \widetilde{G}^{\mathcal L}_\wild(k_0 \vert T, L)$, as follows from the representation
given in eq.~\eqref{eq:skeleton:Sprime-2}. 

Alternatively, one can treat the
$K$-source term as a small interaction. The $O(K)$ term of the two-point
function then is obtained by inserting the interaction and can be written in
terms of a four-point function. Explicit calculation yields
\begin{multline}
\widetilde G^{[K]}(k_0 \vert \infty) 
   =
\widetilde G(k_0 \vert \infty) 
   \\
   - \frac{1}{6} \sum_{\rho=1}^3 \sum_q \int \frac{d^4 p}{(2 \pi)^4} \, K_\wild(p) \, (p^2 + m^2)^2 \, D(p)^2 \,  C^{\pi\gamma\gamma\pi}_{q \rho\rho \bar{q}}(p,k,-k,-p)
   +
   O(K^2)
   \label{eq:skeleton:jjK-2}
   \ ,
\end{multline}
where $k = (k_0,\vec{0})$. Here we have introduced $D(p)$ as the momentum-space
pion two-point function (or dressed propagator)
\begin{gather}
   \delta_{q_1 \bar{q}_2}
   D(p)
   = 
   \int d^4x
   \, e^{-ipx} \langle \pi_{q_1}(x) \pi_{q_2}(0) \rangle_\text{c}
   =
   \frac{\delta_{q_1 \bar{q}_2}}{p^2 + m^2 - \Sigma(p)}
   \ ,
\end{gather}
where the subscript c indicates that only the connected contribution is kept. 
We have also introduced the
connected, amputated four-point function, $C^{\pi\gamma\gamma\pi}_{q
\rho_1 \rho_2 \bar{q}}$, defined via
\begin{multline}
   \langle \pi_{q}(x_1) j_{\rho_1}(y_1) j_{\rho_2}(y_2) \pi_{\bar{q}}(x_2) \rangle_\text{c}
   =
   \int \frac{d^4 p_1}{(2\pi)^4} \frac{d^4 p_2}{(2\pi)^4} \frac{d^4 k_1}{(2\pi)^4} \frac{d^4 k_2}{(2\pi)^4}
   e^{i p_1x_1 + i k_1y_1 + i k_2y_2 + i p_2x_2 }
 \\ \times
   (2\pi)^4 \delta(p_1+k_1+k_2+p_2) \ 
   D(p_1) C^{\pi\gamma\gamma\pi}_{q \rho_1 \rho_2 \bar{q}}(p_1,k_1,k_2,p_2) D(p_2)
   \label{eq:skeleton:Cdef}
   \ .
\end{multline}
By expanding the dressed propagator in powers of the self-energy, $\Sigma(p)$, for the
combination appearing in eq.~\eqref{eq:skeleton:jjK-2} we obtain
\begin{gather}
   K_\wild(p) \, (p^2 + m^2)^2 \, D(p)^2
   =
   K_\wild(p) \, \left\{ \sum_{Q=0}^\infty \left[ \frac{\Sigma(p)}{p^2 + m^2}  \right]^Q \right\}^2
   =
   \sum_{N=1}^\infty N \, K_\wild(p) \left[ \frac{\Sigma(p)}{p^2 + m^2}  \right]^{N-1}
   \ .
\end{gather}
Each term of this expression has a simple diagrammatic interpretation: It is a
chain of alternating $N$ propagators and $N-1$ self-energy insertions, in which each
propagator is replaced once by $K_\wild(p)$ (yielding the overall factor $N$).
Comparing eqs.~\eqref{eq:skeleton:jjK-1} and \eqref{eq:skeleton:jjK-2}, and substituting the above identity, one immediately gets
\begin{gather}
   \Delta \widetilde{G}^{\mathcal L}_\wild(k_0 \vert T, L)
   =
   - \frac{1}{6} \sum_{\rho=1}^3 \sum_{N=1}^\infty N \ 
   \sum_q \int \frac{d^4 p}{(2 \pi)^4} \, K_\wild(p) \, \left[ \frac{\Sigma(p)}{p^2 + m^2}  \right]^{N-1} \,  C^{\pi\gamma\gamma\pi}_{q \rho\rho \bar{q}}(p,k,-k,-p)
   \label{eq:skeleton:Sprime-3}
   \ ,
\end{gather}
which is a representation for $\Delta \widetilde{G}^{\mathcal L}_\wild(k_0 \vert T, L)$ in which all
Feynman integrals have been resummed.

\begin{figure}
   \centering
  \includegraphics[width=\textwidth]{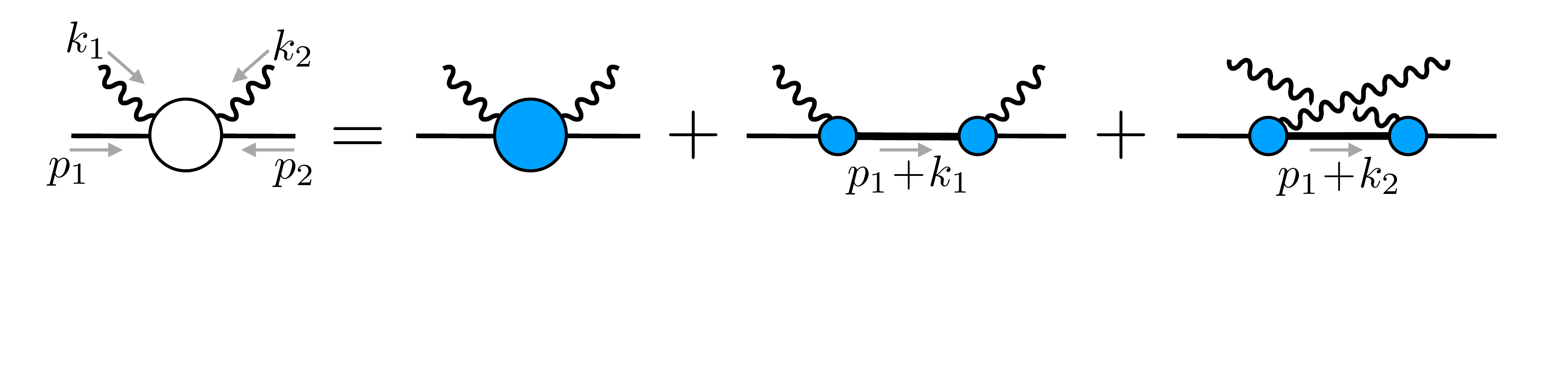}%
   \vspace{-50pt}
   \caption{Skeleton expansion of the $\pi\gamma\gamma\pi$ amputated connected 4-point
   function (white blob) in terms of the 1PI proper vertices (blue blobs) and
   dressed pion propagators (thick lines).}
   \label{fig:4pt-function}
\end{figure}

To complete the skeleton expansion, it remains only to use our compact form of the auxiliary quantity, $\Delta \widetilde{G}^{\mathcal L}_\wild(k_0 \vert T, L)$, 
to reach a form for $\Delta \widetilde{G}_\wild(k_0 \vert T, L)$, in which the sum runs only over the gauge orbits. We will find that the conversion from $\Delta \widetilde{G}^{\mathcal L}_\wild$ to $\Delta \widetilde{G}_\wild$ is \emph{almost} {given} by deleting the factor of $N$ appearing next to the sum in eq.~\eqref{eq:skeleton:Sprime-3}, up to a subtlety that we described in the remainder of this subsection.

The first step is to introduce the decomposition of
$C^{\pi\gamma\gamma\pi}_{\bar{q} \rho_1 \rho_2 q}$ in terms of 1PI vertices
$\Gamma^{\pi\gamma\pi}_{q \rho \bar{q}}$ and $\Gamma^{\pi\gamma\gamma\pi}_{q \rho
\rho \bar{q}}$ 
\begin{multline}
   \begin{split}
   C_{q \rho_1 \rho_2 \bar{q}}(p_1,k_1,k_2,p_2)
   & = 
   \Gamma^{\pi\gamma\pi}_{q \rho_1 \bar{q}}(p_1,k_1,-p_1-k_1) \, D(p_1+k_1) \, \Gamma^{\pi\gamma\pi}_{q \rho_2 \bar{q}}(p_1+k_1,k_2,p_2) \\[3pt]
   & + \Gamma^{\pi\gamma\pi}_{q \rho_2 \bar{q}}(p_1,k_2,-p_1-k_2) \, D(p_1+k_2) \, \Gamma^{\pi\gamma\pi}_{q \rho_1 \bar{q}}(p_1+k_2,k_1,p_2) \\[3pt]
   & + \Gamma^{\pi\gamma\gamma\pi}_{q \rho_1 \rho_2 \bar{q}}(p_1,k_1,k_2,p_2)
   \label{eq:skeleton:skeleton}
   \ .
   \end{split}
\end{multline}
This standard decomposition is represented diagrammatically in
figure~\ref{fig:4pt-function}. 
Substituting into eq.~\eqref{eq:skeleton:Sprime-3} then yields
$\Delta \widetilde{G}^{\mathcal L}_\wild(k_0 \vert T, L)$ as the sum of two terms: The first
contains the 1PI proper vertex, $\Gamma^{\pi\gamma\gamma\pi}_{q \rho \rho
\bar{q}}$, and is diagrammatically represented in figure~\ref{fig:first},
\begin{gather}
 \Delta \widetilde{G}^{\mathcal L}_\wild(k_0 \vert T, L) \supset   - \frac{1}{6} \sum_{\rho=1}^3 \sum_{N=1}^\infty N \ 
   \sum_q \int \frac{d^4 p}{(2 \pi)^4} \, K_\wild(p) \, \left[ \frac{\Sigma(p)}{p^2 + m^2}  \right]^{N-1} \,  \Gamma^{\pi\gamma\gamma\pi}_{q \rho \rho \bar{q}}(p,k,-k,-p)
   \label{eq:skeleton:Sprime-G4}
   \ ,
\end{gather}
and the second contains the 1PI proper vertex, $\Gamma^{\pi\gamma\pi}_{q \rho
\bar{q}}$, and is diagrammatically represented in figure~\ref{fig:second}
\begin{gather}
\Delta \widetilde{G}_\wild(k_0 \vert T, L) \supset   - \frac{1}{3} \sum_{\rho=1}^3 \sum_{N=1}^\infty N \ 
   \sum_{P=0}^\infty 
   \sum_q \int \frac{d^4 p}{(2 \pi)^4} \, K_\wild(p) \, \left[ \frac{\Sigma(p)}{p^2 + m^2}  \right]^{N-1} \,
   \Gamma^{\pi\gamma\pi}_{q \rho \bar{q}}(p,k,-p-k)
   \times \nonumber \\
   \hspace{3cm} \times
   \frac{1}{(p+k)^2 +m^2 } \left[\frac{\Sigma (p+k)}{(p+k)^2 +m^2 }  \right]^{P-1}
   \Gamma^{\pi\gamma\pi}_{q \rho \bar{q}}(p+k,-k,-p)
   \label{eq:skeleton:Sprime-G3}
   \ ,
\end{gather}
where also the dressed propagator appering in eq.~\eqref{eq:skeleton:skeleton}
has been expanded in powers of the self-energy.   Here we have used $\Gamma^{\pi\gamma\pi}_{q_1 \rho q_2}(p_1,k,p_2) = \Gamma^{\pi\gamma\pi}_{q_1 \rho q_2}(p_2,k,p_1)$ and $K(p) = K(-p)$ to combine the two terms, leading to the prefactor of $1/3$ rather than $1/6$.

We are now ready to reconstruct $\Delta \widetilde{G}_\wild (k_0 \vert T, L)$, from the
auxiliary function, $\Delta \widetilde{G}^{\mathcal L}_\wild(k_0 \vert T, L)$, by identifying and
dividing out the multiplicity factor, $N_\wild(\ell)$, appearing in
eq.~\eqref{eq:skeleton:S1}. In order to calculate $N_\wild(\ell)$, one needs to
consider the line, $\ell$, of the Feynman diagram to which the modified propagator,
$K_\wild(p)$, is associated, and then count all the lines which are equivalent to
$\ell$. As we will see, all lines in the equivalence classes, $[\ell]_\text{s}$ and
$[\ell]_\text{t}$, are associated to propagators which appear explicitly in
eqs.~\eqref{eq:skeleton:Sprime-G4} and \eqref{eq:skeleton:Sprime-G3}, and the
multiplicity factor, $N_\wild(\ell)$, is trivially related to the $N$ factor in the
same equations.

\begin{figure}
   \centering
\includegraphics[width=\textwidth]{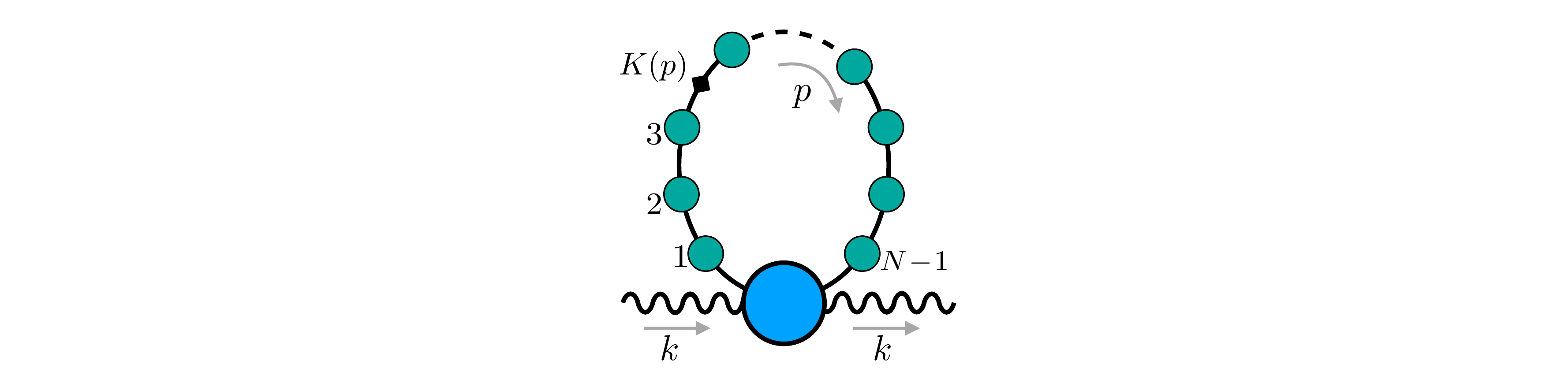}
   \caption{Diagrammatic representation of eq.~\eqref{eq:skeleton:Sprime-G4}. The
   number of self-energy insertions is $N-1$. There are in total $N$ lines in
   the loop, of which $N-1$ are standard propagators, and one corresponds to the
   insertion of $K$. There are in total $N$ possible ways to insert $K$. If
   $\ell$ is the line where $K$ is located, the equivalence class $[\ell]_\text{s} =
   [\ell]_\text{t}$ is the set of all pion lines explicitly drawn in this diagram.}
   \label{fig:first}
\end{figure}

To see how this works out, first consider the term
in~\eqref{eq:skeleton:Sprime-G4} for a given value of $N$, which is
diagrammatically represented in figure~\ref{fig:first}. Let $\ell$ be the line
where $K$ is located. By using the definition of $s$-equivalence classes (see
section \ref{subsec:diag:leading-L}) and the fact that the insertions are 1PI, it is
easy to show that the $s$-equivalence class, $[\ell]_\text{s}$, is the set of all pion
lines explicitly drawn in this diagram. By using the definition of
$t$-equivalence classes (see section \ref{subsec:diag:leading-T}) one further finds that, in
this case, $[\ell]_\text{s}=[\ell]_\text{t}$. Therefore the number, $N_\wild(\ell)$, of lines in
these equivalence classes is equal to $N$. Dropping the factor of $N$
in~\eqref{eq:skeleton:Sprime-G4} amounts to dividing out the multiplicity factor
$N_\wild(\ell)$, yielding the following contribution to $\Delta \widetilde{G}_\wild(k_0 \vert T, L)$:
\begin{gather}
 \Delta \widetilde{G}_\wild(k_0 \vert T, L) \supset  - \frac{1}{6} \sum_{\rho=1}^3 \sum_{N=1}^\infty
   \sum_q \int \frac{d^4 p}{(2 \pi)^4} \, K_\wild(p) \, \left[ \frac{\Sigma(p)}{p^2 + m^2}  \right]^{N-1} \,  \Gamma^{\pi\gamma\gamma\pi}_{q \rho \rho \bar{q}}(p,k,-k,-p)
   \ .
   \label{eq:skeleton:piece4}
\end{gather}

\begin{figure}
   \centering
  \includegraphics[width=\textwidth]{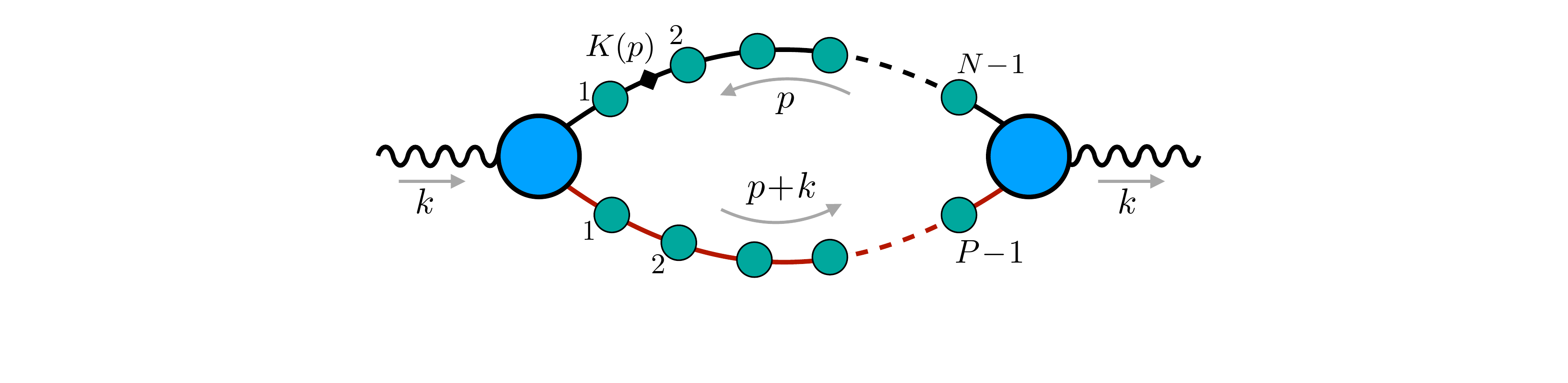}
  \vspace{-30pt}
   \caption{Diagrammatic representation of eq.~\eqref{eq:skeleton:Sprime-G3}.
   There are two self-energy chains, the black one and the red one. The number
   of black self-energy insertions is $N-1$. There are in total $N$ black lines
   in the loop, of which $N-1$ are standard propagators, and one corresponds to
   the insertion of $K$. The number of red self-energy insertions is $P-1$.
   There are in total $P$ red lines in the loop, all correspoinding to standard
   propagators. If $\ell$ is the line where $K$ is located, the equivalence
   class $[\ell]_\text{t}$ is the set of all black pion lines explicitly drawn in this
   diagram. By contrast, the equivalence class $[\ell]_\text{s}$ is the set of all (black
   and red) pion lines explicitly drawn in this diagram.}
   \label{fig:second}
\end{figure}

Consider now the term in~\eqref{eq:skeleton:Sprime-G3} for  given values of $N$
and $P$, diagrammatically represented in figure~\ref{fig:second}. Again,
let $\ell$ be the line where $K$ is located. As is the case for the $\Gamma^{\pi \gamma \gamma \pi}$ term,
discussed in the previous paragraph, here the $s$-equivalence class
$[\ell]_\text{s}$ is the set of all pion lines explicitly drawn in the diagram.
However, in this case the $s$-equivalence class splits into two $t$-equivalence
classes. In particular $[\ell]_\text{t}$ is the set of all black pion lines
explicitly drawn in the figure. This implies that the number, $N_\text{t}(\ell)$, of
lines in $[\ell]_\text{t}$ is equal $N$. Dropping the factor of $N$
in~\eqref{eq:skeleton:Sprime-G3} amounts to dividing out the multiplicity factor,
$N_\text{t}(\ell)$, yielding the following contribution to $\Delta
\widetilde{G}_\text{t}(k_0 \vert T)$
\begin{gather}
\Delta
\widetilde{G}_\text{t}(k_0 \vert T) \supset   - \frac{1}{3} \sum_{\rho=1}^3 \sum_{N=1}^\infty
   \sum_{P=0}^\infty 
   \sum_q \int \frac{d^4 p}{(2 \pi)^4} \, K_\wild(p) \, \left[ \frac{\Sigma(p)}{p^2 + m^2}  \right]^{N-1} \,
   \Gamma^{\pi\gamma\pi}_{q \rho \bar{q}}(p,k,-p-k)
   \times \nonumber \\
   \hspace{3cm} \times
   \frac{1}{(p+k)^2 +m^2 } \left[\frac{\Sigma (p+k)}{(p+k)^2 +m^2 }  \right]^{P-1}
   \Gamma^{\pi\gamma\pi}_{q \rho \bar{q}}(p+k,-k,-p)
   \ .
   \label{eq:skeleton:piece3}
\end{gather}

So, in the cases analyzed so far, we have found that $N_\eta(\ell) = N$ so that one can convert $ \Delta \widetilde{G}^{\mathcal L}_\wild(k_0 \vert T, L) $ to $ \Delta \widetilde{G}_\wild(k_0 \vert T, L) $, simply by dropping the explicit factor of $N$ appearing in eq.~\eqref{eq:skeleton:Sprime-G3}.
The remaining case to consider, the $\Gamma^{\pi \gamma \pi}$-dependent contribution to $\Delta \widetilde{G}_{\text s}(k_0 \vert  L) $, breaks this pattern.
To see why, first note that $[\ell]_\text{s}$ is the combined set of all black and red pion lines, explicitly drawn in
figure~\ref{fig:second}. It is evident that there is a symmetry between the black
and the red sets, after summing over $N$ and $P$. This can be made
manifest in eq.~\eqref{eq:skeleton:Sprime-G3}, by recasting the series as \begin{multline}
\Delta \widetilde{G}_{\wild}(k_0 \vert T, L)  \supset  - \frac{1}{3} \sum_{\rho=1}^3 \sum_{N=1}^\infty \frac{N+P}{2} \ 
   \sum_{P=0}^\infty 
   \sum_q \int \frac{d^4 p}{(2 \pi)^4} \, K_\wild(p) \, \left[ \frac{\Sigma(p)}{p^2 + m^2}  \right]^{N-1} \,
\\ \times 
    \Gamma^{\pi\gamma\pi}_{q \rho \bar{q}}(p,k,-p-k)  \frac{1}{(p+k)^2 +m^2 } \left[\frac{\Sigma (p+k)}{(p+k)^2 +m^2 }  \right]^{P-1}
   \Gamma^{\pi\gamma\pi}_{q \rho \bar{q}}(p+k,-k,-p)
   \ ,
   \label{eq:skeleton:Sprime-G3-NP}
\end{multline}
where the factor $N$ has been replaced by $(N+P)/2$. The equality to \eqref{eq:skeleton:Sprime-G3} is proven by
exchanging $P$ and $N$ in the term proportional to $P$, changing variables $p
\to -p-k$, using invariance under sign change of all momenta in the proper
vertices (for a detailed discussion of the symmetries of the proper vertices,
see sec. \ref{subsec:calculation:preliminaries}), and, finally, applying the definition of
$K_\wild(p)$. The number, $N_\text{s}(\ell)$, of lines in $[\ell]_\text{s}$ is
equal to $N+P$. Thus dropping the factor of $N+P$ in~\eqref{eq:skeleton:Sprime-G3-NP}
amounts to dividing out the multiplicity factor, $N_\text{s}(\ell)$, yielding a
contribution to $\Delta \widetilde{G}_\text{s}(k_0 \vert L)$ that is equal to
\eqref{eq:skeleton:piece3} multiplied by an extra factor of $1/2$.

Finally, we use the
explicit expression for $K_\wild(p)$, given in eq.~\eqref{eq:skeleton:Kdef}, and
 resum the geometric series for the dressed propagators in eqs.~\eqref{eq:skeleton:piece4} and \eqref{eq:skeleton:piece3}. This yields the skeleton expansion
\begin{multline}
   \label{eq:skeleton:final}
   \Delta \widetilde{G}_\wild(k_0 \vert T, L)
   =
   - \frac{1}{6} \sum_{u \in U_\wild} \sum_{\rho=1}^3
   \sum_q \int \frac{d^4 p}{(2 \pi)^4} \, 
   \cos \left( u \L p \right)
   D(p) \, 
   \\ \times \Big [
   \Gamma^{\pi\gamma\gamma\pi}_{q \rho \rho \bar{q}}(p,k,-k,-p)
   + \beta_\wild
   \Gamma^{\pi\gamma\pi}_{q \rho \bar{q}}(p,k,-p-k)
   D(p+k)
   \Gamma^{\pi\gamma\pi}_{q \rho \bar{q}}(p+k,-k,-p)
   \Big ]_{\vec{k}=\vec{0}}
   \ ,
\end{multline}
where $\beta_\text{t} = 2$ and $\beta_\text{s}=1$.

We have reached a compact form for $ \Delta \widetilde{G}_{\text{t}}(k_0 \vert T)$ and $ \Delta \widetilde{G}_{\text{s}}(k_0 \vert L)$, expressed in terms of three basic building blocks and completely independent of all details of the effective-field theory. This, however, has come at the cost that the exponentially suppressed scaling is now hidden. This is remedied in the next section.

\section{Relation to partially on-shell 4pt functions}
\label{sec:calculation}

For most of section \ref{sec:diag} we have worked in position space %
and with Euclidean-signature correlators. %
The advantage of this approach lies in the fact that the
exponential decay of the finite-volume corrections to each Feynman diagram can be
understood in terms of a saddle-point expansion. When finite-volume corrections
are written in momentum space, as in eq.~\eqref{eq:skeleton:final}, the
exponential decay is no longer manifest, as $L$ and $T$ appear in oscillatory
factors.

To make the underlying scaling manifest in momentum space, one relies on the fact that the leading large-volume corrections are dominated by single-particle
poles, located in the complex plane, away from the integration
domain. Thus, in this section, we perform the contour deformations of eq.~\eqref{eq:skeleton:final}, necessary to identify the
 $e^{- mL}$ and $e^{- mT}$ behavior. This section is divided into three subsections:
first we introduce some convenient notation and discuss symmetry properties of
various objects entering the calculation; then we evaluate, respectively, the
leading contributions from the finite spatial and temporal extents.

\subsection{Preliminaries}
\label{subsec:calculation:preliminaries}

We are concerned here with the leading finite-$L$ correction given in
eq.~\eqref{eq:stat:DeltaIs} and with the wrapped-pion piece of the leading
finite-$T$ correction given in eq.~\eqref{eq:stat:DeltaIt-WP}. These corrections
are written in terms of the quantity $\Delta G_\wild(x_0 \vert T, L)$, whose Fourier
transform has been conveniently written in eq.~\eqref{eq:skeleton:final} in
terms of infinite-volume dressed propagators and 1PI vertices. We introduce some
convenient notation to make the following equations more compact.

We parametrize the infinite-volume dressed propagator  as
\begin{gather}
   D(p) = \frac{Z(p)}{p^2 + m^2}
   \label{eq:calc:Zdef}
   \ ,
\end{gather}
where $Z(p)$ is a function of $p^2$ only, and in particular it is an analytic
function below the three-pion threshold, i.e.~for $\Re p^2 > -9m^2$. We define
also the functions
\begin{gather}
   \mathcal{A}_{\rho\sigma}(p,k) = \sum_q Z(p-\tfrac{k}{2}) \Gamma^{\pi\gamma\pi}_{q \rho \bar{q}}(p-\tfrac{k}{2},k,-p-\tfrac{k}{2}) Z(p+\tfrac{k}{2}) \Gamma^{\pi\gamma\pi}_{q \sigma \bar{q}}(p+\tfrac{k}{2},-k,-p-\tfrac{k}{2})
   \label{eq:calc:Adef}
   \ , \\
   \mathcal{M}_{\rho\sigma}(p,k) = \sum_q Z(p) \Gamma^{\pi\gamma\gamma\pi}_{q \rho \sigma \bar{q}}(p,k,-k,-p)
   \label{eq:calc:Mdef}
   \ ,
\end{gather}
in terms of which the inverse Fourier transform of $\Delta
\widetilde{G}_\wild(k_0 \vert T, L)$ %
is conveniently
written as
\begin{multline}
   \label{eq:calc:DeltaG-0}
   \Delta G_\wild(x_0 \vert T, L)
   =
   - \frac{1}{6} \sum_{u \in U_\wild} \sum_{\rho=1}^3
   \int \frac{d k_0}{2\pi} e^{ik_0x_0}
   \int \frac{d^4 p}{(2 \pi)^4} \, 
   \frac{
   \cos \left( u \L p \right)
   }{ p^2 + m^2 } \\ \times
   \left[
   \beta_\wild \frac{\mathcal{A}_{\rho\rho}(p+\frac{k}{2},k)}{(p+k)^2+m^2}
   + \mathcal{M}_{\rho\rho}(p,k)
   \right]_{\vec{k}=\vec{0}}
   \ .
\end{multline}
Recall that, thanks to the analysis of section \ref{subsec:diag:2pt-function},
the leading exponentials of  $\Delta a_\text{s}(L)$ and $\Delta
a_\text{t}^\text{WP}(T)$ can be obtained from the leading exponential of $\Delta
G_\wild(x_0 \vert T, L)$ at fixed $x_0$.

We list here the symmetry properties of the functions $\mathcal{A}$ and
$\mathcal{M}$. Invariance under exchange of the two electromagnetic currents
implies
\begin{gather}
   \mathcal{A}_{\rho\sigma}(p,-k) = \mathcal{A}_{\sigma\rho}(p,k)
   \ , \qquad
   \mathcal{M}_{\rho\sigma}(p,-k) = \mathcal{M}_{\sigma\rho}(p,k)
   \ ,
   \label{eq:calc:photonex}
\end{gather}
and invariance under Euclidean parity and Euclidean rotations implies, for $R \in
\text{O}(4)$,
\begin{gather}
   \mathcal{A}_{\rho\sigma}(Rp,Rk) = \sum_{\rho'\sigma'} R_{\rho\rho'} R_{\sigma\sigma'} \mathcal{A}_{\rho'\sigma'}(p,k)
   \ , \qquad
   \mathcal{M}_{\rho\sigma}(Rp,Rk) = \sum_{\rho'\sigma'} R_{\rho\rho'} R_{\sigma\sigma'} \mathcal{M}_{\rho'\sigma'}(p,k)
   \ .
   \label{eq:calc:rotations}
\end{gather}
Under CPT, local fields transform generally as
$\mathcal{O}_{\mu_1,\dots,\mu_n}^\text{CPT}(x) = (-1)^n
\mathcal{O}_{\mu_1,\dots,\mu_n}(-x)^*$ (it is straightforward to check that this
is true both in Minkowskian and Euclidean time). In particular
$\pi_q^\text{CPT}(x) = \pi_q(-x)^*$ and $j_\mu^\text{CPT}(x) = -j_\mu(-x)^*$.
Thus, the action density transforms like $\mathcal{L}_\text{E}^\text{CPT}(x) =
\mathcal{L}_\text{E}(-x)^*$, and the action like $S(\pi^\text{CPT}) =
S(\pi)^*$. One easily checks that this implies
\begin{gather}
   \mathcal{A}_{\rho\sigma}(p^*,k^*) = \mathcal{A}_{\rho\sigma}(p,k)^*
   \ , \qquad
   \mathcal{M}_{\rho\sigma}(p^*,k^*) = \mathcal{M}_{\rho\sigma}(p,k)^*
   \ .
   \label{eq:calc:cpt}
\end{gather}

We will find that the leading finite-$L$ correction of the current
two-point function can be expressed in terms of the forward Compton scattering
amplitude $T(k^2,p  k)$ %
of an off-shell photon with momentum
$k$ against an on-shell pion with momentum $p$. The finite-$T$ correction can
be expressed in terms of a function obtained by analytically continuing
$T(k^2,p  k)$. %
Here we want to write the Compton scattering amplitude
in terms of the functions $\mathcal{A}$ and $\mathcal{M}$.
We begin here by considering the relation of the connected
amputated four-point function $C^{\pi\gamma\gamma\pi}_{q \rho \sigma \bar{q}}$,
defined in eq.~\eqref{eq:skeleton:Cdef}, to its Minkowski-signature counterpart.
The Minkowskian 4-point function is obtained by Wick rotating the temporal
components of the momenta to the imaginary axis
\begin{gather}
   \widetilde{C}^{\pi\gamma\gamma\pi}_{q \rho \sigma \bar{q}}(p_1,k_1,k_2,p_2)
   =
   \eta_\rho \eta_\sigma \lim_{\epsilon \to 0^+} C^{\pi\gamma\gamma\pi}_{q \rho \sigma \bar{q}}(\tilde{p}_1^\epsilon,\tilde{k}_1^\epsilon,\tilde{k}_2^\epsilon,\tilde{p}_2^\epsilon)
   \ ,
\end{gather}
where the repeated indices are not summed. Here we have used \begin{gather}
   \tilde{v}^\epsilon = (i e^{i \epsilon} v_0 , \vec{v} )
   \ , \qquad
   \eta = (1 , -i , -i , -i)
   \ .
\end{gather}
The prefactors $\eta_\mu$ take into account that the Euclidean and Minkoskian
electromagnetic current are not equal, as one can easily check by applying
Noether's theorem to the Euclidean and Minkowskian actions. The forward Compton
tensor, summed over the pion charges, is obtained by means of the LSZ reduction
formula applied to $\widetilde{C}$ in a standard fashion, and is given by
\begin{align}
   T_{\rho\sigma}(k|\vec{p})
   &  =
   \bigg[
   Z(p) \sum_q \widetilde{C}^{\pi\gamma\gamma\pi}_{q \rho\sigma \bar{q}}(p,k,-k,-p)
   \bigg]_{p_0 = E(\vec{p})} \,, \\[5pt]
   & = 
   \lim_{\epsilon \to 0^+} \bigg[
   \eta_\rho \eta_\sigma Z(\tilde{p}^\epsilon)   \sum_q C^{\pi\gamma\gamma\pi}_{q \rho \sigma \bar{q}}(\tilde{p}^\epsilon,\tilde{k}^\epsilon,-\tilde{k}^\epsilon,-\tilde{p}^\epsilon)
   \bigg]_{p_0 = E(\vec{p})}
   \label{eq:calc:compton0}
   \ ,
\end{align}
where we have introduced $E(\vec{p}) = \sqrt{m^2+\vec{p}^2}$. By using the
skeleton expansion~\eqref{eq:skeleton:skeleton} for $C$, the definitions
\eqref{eq:calc:Adef} and \eqref{eq:calc:Mdef}, the observation that
\begin{gather}
   \left[
   \frac{1}{(\tilde{p}^\epsilon \pm \tilde{k}^\epsilon)^2+m^2}
   \right]_{p_0 = E(\vec{p})}
   =
   \frac{1}{
   - k_0^2 + \vec{k}^2 \mp 2 [ k_0 E(\vec{p}) - \vec{p}\vec{k} ] - i \epsilon
   }
   \ ,
\end{gather}
in the $\epsilon \to 0^+$ limit, and the definition
\begin{gather}
   \bar{p} = (i E(\vec{p}),\vec{p})
   \ ,
\end{gather}
one obtains (specializing to $\rho=\sigma$, which is not summed)
\begin{gather}
   g^{\rho\rho}   T_{\rho\rho}(k|\vec{p})
   =
   \bigg[
   \frac{\mathcal{A}_{\rho\rho}(\bar{p}+\frac{\tilde{k}}{2},\tilde{k})}{
   - k_0^2 + \vec{k}^2 - 2 ( k_0 E(\vec{p}) - \vec{p}\vec{k} ) - i \epsilon}
   + \nonumber \\ \hspace{3cm}
   + \frac{\mathcal{A}_{\rho\rho}(\bar{p}-\frac{\tilde{k}}{2},-\tilde{k})}{
   - k_0^2 + \vec{k}^2 + 2 ( k_0 E(\vec{p}) - \vec{p}\vec{k} ) - i \epsilon}
   + \mathcal{M}_{\rho\rho}(\bar{p},\tilde{k})
   \bigg]_{ \tilde{k}=(ik_0,\vec{k}) }
   \label{eq:calc:compton1}
   \ .
\end{gather}
Note that we have dropped the $\epsilon$ dependence in the arguments of
$\mathcal{A}$ and $\mathcal{M}$. This is possible if $E(\vec{p})<2m$, or
equivalently $\vec{p}^2 < 3m^2$, and $k_0$ real, since those functions are
analytic below the two-pion threshold (see theorem~\ref{theo:ana1}). In the following we will be interested in the analytic continuation to complex values
of $k_0$. However $\vec{p}$ and $\vec{k}$ are kept real at all stages. 

We next comment on the strange looking index structure, in which the object $ g^{\rho\rho} T_{\rho\rho}$, with no sum over $\rho$, appears on the left-hand side. 
Note first that, when we do sum over the index, defining
\begin{gather}
   T(k_\mu k^\mu, \bar{p}_\mu k^\mu)
   =
   \sum_{\rho=0}^3 g^{\rho\rho} T_{\rho\rho}(k|\vec{p})
   \ ,
\end{gather}
then this becomes the usual Lorentz invariant Minkowskian trace, equal to $g^{\mu \nu} T_{\mu \nu} = {T^\mu}_\mu$ in the standard index notation.
This allows us to write this quantity as a function of the Lorentz
invariants $k_\mu k^\mu = \sum_{\rho,\sigma} g^{\rho\sigma} k_\rho k_\sigma$ and
$\bar{p}_\mu k^\mu = \sum_{\rho,\sigma} g^{\rho\sigma} \bar{p}_\rho k_\sigma$, as we have done on the left-hand side. The only non-standard aspect is the factors of $A_{\rho \rho}$ and $\mathcal M_{\rho \rho}$, with no factor of $g^{\rho \rho}$, on the right-hand side of eq.~\eqref{eq:calc:compton1}. This mismatch is due to the fact that we are mixing Euclidean and Minkowskian conventions at this intermediate stage of the calculation.

\subsection{Finite \texorpdfstring{$L$}{L}}

We return now to eq.~\eqref{eq:calc:DeltaG-0} for $\wild=\text{s}$ and show that this can be rewritten as given in eq.~\eqref{eq:stat:DeltaGs}. Using the
reality of the functions $\mathcal{A}$ and $\mathcal{M}$ (CPT), and invariance
under parity along the 0 direction, one obtains the following representation
\begin{multline}
   \Delta G_\text{s}(x_0 \vert L)
   =
   - \frac{1}{6} \Re \sum_{ \vec{n} \neq \vec{0} } \sum_{\rho=1}^3
   \int \frac{d k_0}{2\pi} \cos (k_0x_0)
   \int \frac{d^4 p}{(2 \pi)^4} \, 
   \frac{
   e^{i L \vec{n} \vec{p}}
   }{ p^2 + m^2 }
\times  \\ \times \left[
   \frac{\mathcal{A}_{\rho\rho}(p+\frac{k}{2},k)}{(p+k)^2+m^2}
   + \mathcal{M}_{\rho\rho}(p,k)
   \right]_{\vec{k}=\vec{0}}
   \label{eq:scalc:DeltaG3-0}
   \ .
\end{multline}

We next rewrite this equation in a slightly different way, that
will be convenient at a later stage. First, we change variables $\vec{p} \to R^T
\vec{p}$ where $R$ is an SO($3$) matrix such that $R \vec{n} =
(0,0,|\vec{n}|)$. Using also the covariance of the functions $\mathcal{A}$ and
$\mathcal{M}$ under (spatial) rotations, and the invariance of the
$\vec{k}=\vec{0}$ constraint, this amounts simply to replacing
\begin{gather}
   e^{i L \vec{n} \vec{p}}
   \rightarrow
   e^{i L |\vec{n}| p_3} \,,
\end{gather}
in eq.~\eqref{eq:scalc:DeltaG3-0}. Then we exchange the labels $k_0$ and $k_3$ everywhere in eq.~\eqref{eq:scalc:DeltaG3-0} also
exchange $p_0 \leftrightarrow p_3$, which amounts to acting with a particular
O($4$) matrix. 
We reach
\begin{multline}
   \Delta G_\text{s}(x_0 \vert L)
   =
   - \frac{1}{6} \Re \sum_{ \vec{n} \neq \vec{0} } \sum_{\rho=0}^2
   \int \frac{d k_3}{2\pi} \cos(k_3x_0)
   \int \frac{d^4 p}{(2\pi)^4} \frac{e^{i |\vec{n}| L p_0}}{p^2+m^2}
\times \\ \times   \left[
   \frac{\mathcal{A}_{\rho\rho}(p+\frac{k}{2},k)}{(p+k)^2+m^2} + \mathcal{M}_{\rho\rho}(p,k)
   \right]_{k_0=k_1=k_2=0}
   \label{eq:scalc:DeltaG3-1}
   \ .
\end{multline}

Observe that the integrand has four explicit complex poles in $p_0$ and that
these become pairwise degenerate when $k_3=0$. In order to avoid these double
poles, we deform the integrand by replacing, with $\epsilon = 0^+$,
\begin{gather}
   \frac{1}{p^2+m^2} \to \frac{1}{p^2+m^2-i\epsilon} \ , \\
   \frac{1}{(p+k)^2+m^2} \to \frac{1}{(p+k)^2+m^2+i\epsilon} \ .
\end{gather}
We shift the $p_0$ integral in the complex plane to $\mathbb{R} + i m \sqrt{2 +
\sqrt{3}}$, which is allowed because $\mathcal{A}_{\rho\rho}(p+\frac{k}{2},k)$
and $\mathcal{M}_{\rho\rho}(p,k)$ are analytic in the strip $| \! \Im p_0| < 2m$
(see theorem~\ref{theo:ana1}) and $\sqrt{2 + \sqrt{3}} < 2$. 
Introducing the notation $p_\perp=(p_1,p_2)$, we observe that, 
in the shift, one
picks up the pole
\begin{gather}
   p_0 = i \sqrt{m^2-i\epsilon + p_\perp^2 + p_3^2}
   \label{eq:scalc:pole1}
   \ ,
\end{gather}
as long as $\vec{p}^2 = p_\perp^2 + p_3^2 < ( 1 + \sqrt{3} ) m$, and the pole
\begin{gather}
   p_0 = i \sqrt{m^2+i\epsilon + p_\perp^2 + (p_3+k_3)^2} \ ,
   \label{eq:scalc:pole2}
\end{gather}
as long as $p_\perp^2 + (p_3+k_3)^2 < ( 1 + \sqrt{3} ) m$.  The contribution from the shifted contour is seen
to be of order $e^{-\sqrt{2+\sqrt{3}} m L}$, i.e.~it is of the same order as
terms that we have already neglected. Therefore only the above poles contribute
to the leading order of the integral~\eqref{eq:scalc:DeltaG3-1}, yielding
\begin{gather}
  \Delta G_\text{s}(x_0 \vert L)
   =
   - \frac{1}{6} \Re \sum_{ \vec{n} \neq \vec{0} } \sum_{\rho=0}^2
   \int_{\vec{p}^2 < ( 1 + \sqrt{3} )m^2} \frac{d^3 p}{(2\pi)^3}
   \frac{
   e^{-|\vec{n}| E(\vec p) L}
   }{
   2 E(\vec p)
   } \int \frac{d k_0}{2\pi} \cos (k_3 x_0)
   \times \nonumber \\  \label{eq:scalc:DeltaG3-2}  \hspace{3cm} \times
   \left[
   \frac{\mathcal{A}_{\rho\rho}(\bar{p}+\frac{k}{2},k)}{k_3^2+2p_3k_3+2i\epsilon}
   + \frac{\mathcal{A}_{\rho\rho}(\bar{p}-\frac{k}{2},k)}{k_3^2-2p_3k_3-2i\epsilon}
   + \mathcal{M}_{\rho\rho}(\bar{p},k)
   \right]_{k_0=k_1=k_2=0}
  \\ \hspace{3cm} + O\left( e^{-\sqrt{2+\sqrt{3}} m L} \right) \nonumber
   \ .
\end{gather}%

The quantity appearing in square brackets here is almost the Compton tensor, given in eq.~\eqref{eq:calc:compton1},
for $k_0=k_1=k_2=0$ and $\vec{p}^2\le (1+\sqrt{3})m^2<3m^2$,
\begin{gather}
   g^{\rho\rho} T_{\rho\rho}(0,0,0,k_3|\vec{p})
   =
   \left[
   \frac{\mathcal{A}_{\rho\rho}(\bar{p}+\frac{k}{2},k)}{
   k_3^2 + 2 p_3k_3 - i \epsilon}
   + \frac{\mathcal{A}_{\rho\rho}(\bar{p}-\frac{k}{2},-k)}{
   k_3^2 - 2 p_3k_3 - i \epsilon}
   + \mathcal{M}_{\rho\rho}(\bar{p},k)
   \right]_{ k=(0,0,0,k_3) }
   \label{eq:scalc:compton0}
   \ ,
\end{gather}
except for the wrong $i\epsilon$ prescription in the first propagator. However
this difference is immaterial because only the real part contributes, and the
numerator is found to be real (this is a consequence of CPT and parity in the 0
direction). Therefore one obtains
\begin{multline}
  \Delta G_\text{s}(x_0 \vert L)
   =
   - \frac{1}{6} \sum_{ \vec{n} \neq \vec{0} }
   \int_{\vec{p}^2 < (1+\sqrt{3})m^2} \frac{d^3 p}{(2\pi)^3}
   \frac{
   e^{-|\vec{n}| L E(\vec{p})}
   }{
   2 E(\vec{p})
   }
   \int \frac{d k_0}{2\pi} \cos (k_0x_0) \times \\  \times 
   \sum_{\rho=0}^2 g^{\rho\rho} \Re T_{\rho\rho}(0,0,0,k_3|\vec{p})
   + O\left( e^{-\sqrt{2+\sqrt{3}} m L} \right)
   \label{eq:scalc:DeltaG3-3}
   \ .
\end{multline}

The Compton tensor satisfies the electromagnetic Ward identity
$\sum_{\rho \alpha} g^{\rho \alpha} k_\alpha T_{\rho\sigma}(k|\vec{p}) = 0$ 
\cite{\ward}, which specializes in the case of interest to
\begin{gather}
   k_3 T_{33}(0,0,0,k_3|\vec{p}) = 0 \ .
   \label{eq:scalc:ward}
\end{gather}
This must be interpreted as an equation for a tempered distribution, meaning
that $T_{33}$ may be proportional to $\delta(k_3)$. However, using the
analyticity properties of $\mathcal{A}$ and $\mathcal{M}$ in the kinematic region
of interest, one proves that $T_{33}(0,0,0,k_3|\vec{p})$ is analytic for
real values of $k_3$,\footnote{
Multiplying the Ward identity~\eqref{eq:scalc:ward} by $(k_3^2-4p_3^2)$, one
obtains the following equation for analytic functions
\begin{gather}
   0
   =
   - (k_3^2-4p_3^2) k_3 T_{33}(0,0,0,k_3|\vec{p})
   = \\ \nonumber \qquad =
   \left[
   (k_3-2p_3) \mathcal{A}_{33}(\bar{p}+\tfrac{k}{2},k)
   + (k_3+2p_3) \mathcal{A}_{33}(\bar{p}-\tfrac{k}{2},-k)
   + (k_3^2-4p_3^2) k_3 \mathcal{M}_{33}(\bar{p},k)
   \right]_{ k=(0,0,0,k_3) }
   \ .
\end{gather}
Evaluating this at $k_3= \mp 2p_3$, and using the fact that $\mathcal{A}$ is a
continuous function in both variables and in the kinematical region of interest,
one finds
\begin{gather}
   \left[ \mathcal{A}_{33}(\bar{p} \pm \tfrac{k}{2}, \pm k) \right]_{ k=(0,0,0,\mp 2p_3) }
   = 0 \ .
\end{gather}
This fact, with the analyticity of $\mathcal{A}_{33}$ for real values of $p_3$
and $k_3$, implies that
\begin{gather}
   k_3
   \left[
   \frac{
   \mathcal{A}_{\rho\rho}(\bar{p} \pm \frac{k}{2}, \pm k)
   }{
   k_3^2 \pm 2p_3k_3 -i \epsilon
   }
   \right]_{ k=(0,0,0,k_3) }
   =
   \frac{
   \left[ \mathcal{A}_{\rho\rho}(\bar{p} \pm \frac{k}{2}, \pm k) \right]_{ k=(0,0,0,k_3) }
   - \left[ \mathcal{A}_{33}(\bar{p} \pm \tfrac{k}{2}, \pm k) \right]_{ k=(0,0,0,\mp 2p_3)}}{
   k_3 \pm 2 p_3} \,,
\end{gather}
is an analytic function for real values of $p_3$ and $k_3$, and so is $k_3 T_{33}(0,0,0,k_3|\vec{p})$. Therefore eq.~\eqref{eq:scalc:ward} implies $T_{33}(0,0,0,k_3|\vec{p}) = 0$.
} and therefore the Ward identity implies $T_{33}(0,0,0,k_3|\vec{p}) = 0$. It
follows that the sum over $\rho$ in eq.~\eqref{eq:scalc:DeltaG3-3} can be
equivalently extended to $\rho=3$, i.e.
\begin{gather}
   \sum_{\rho=0}^2 g^{\rho\rho} T_{\rho\rho}(0,0,0,k_3|\vec{p})
   =
   \sum_{\rho=0}^3 g^{\rho\rho} T_{\rho\rho}(0,0,0,k_3|\vec{p})
   =
   T(-k_3^2,-p_3k_3)
   \ .
\end{gather}
Note that this quantity does not depend on $p_\perp$. Dropping the
restrictions over the $\vec{p}$ integration domain in
eq.~\eqref{eq:scalc:DeltaG3-3} amounts to an error of
$O(e^{-\sqrt{2+\sqrt{3}}mL})$. After this is done, the integral over $p_\perp$
can be calculated explicitly, yielding eq.~\eqref{eq:stat:DeltaGs}
\begin{multline}
 \Delta G_\text{s}(x_0 \vert L)
   =
   - \sum_{ \vec{n} \neq \vec{0} }
   \int \frac{d p_3}{2\pi}
   \frac{
   e^{-|\vec{n}| L \sqrt{m^2+p_3^2}}
   }{
   24 \pi |\vec{n}| L
   }
   \int \frac{d k_3}{2\pi} \cos (k_3x_0)
   \Re T(-k_3^2,-p_3k_3)
  \\ + O\left( e^{-\sqrt{2+\sqrt{3}} m L} \right)
\nonumber
   \ .
\end{multline}
Note here that we have used a slight abuse of notation in sections~\ref{sec:statement} and \ref{sec:diag}, giving $\Delta G_\text{s}(x_0 \vert L)$ (as well as $\Delta a_{\text{s}}(L)$) as the Compton-amplitude-only piece in the former section, but defining it as the sum over all simple gauge orbits in the latter. This amounts to shuffling the $O ( \cdots \! \, )$ term appearing here into the relation between $\Delta G(x_0 \vert \infty, L)$ and $\Delta G_\text{s}(x_0 \vert L)$. Since the two quantities differ by this scaling anyway, we thought it unnecessary to introduce a separate symbol.

\subsection{Finite \texorpdfstring{$T$}{T}}

We now complete the derivation by analysing the case $\eta=\text{t}$ and showing that %
eq.~\eqref{eq:calc:DeltaG-0} can be rewritten as eqs.~\eqref{eq:stat:DeltaGt} and \eqref{eq:stat:Sdef}.
Here it is convenient to work with the inverse Fourier transform of $\hat S(x_0, \vec p^2)$, defined in \eqref{eq:stat:Sdef}. We thus introduce
\begin{gather}
   S(k_0,\vec{p}^2)
   =
   \lim_{\vec{p}' \to \vec{p}}
   \sum_{\rho=1}^3 \sum_q \int d^4x \ e^{ik_0x_0} \langle \vec{p}',q | \text{T} j_\rho(x) j_\rho(0) | \vec{p},q \rangle
   \label{eq:tcalc:S0}
   \ ,
\end{gather}
where the Euclidean electromagnetic current in the Heisenberg picture
satisfies
\begin{gather}
   j_\rho(x) = e^{x_0H-i\vec{P}\vec{x}} j_\rho(0) e^{-x_0H+i\vec{P}\vec{x}}
   \ .
\end{gather}
Our first goal is to obtain an expression for $S(k_0,\vec{p})$ in terms of 1PI
proper vertices. This is done by relating this quantity to the Compton tensor,
and in particular to
\begin{gather}
   \bar{T}(k_0,\vec{p}^2)
   =
   \sum_{\rho=1}^3 g^{\rho\rho} T_{\rho\rho}(k|\vec{p})
   =
   i \lim_{\vec{p}' \to \vec{p}}
   \sum_{\rho=1}^3 \sum_q g^{\rho\rho}
   \int d^4x \ e^{ik_0x_0} \langle \vec{p}',q | \text{T} \mathcal{J}_\rho(x) \mathcal{J}_\rho(0) | \vec{p},q \rangle
   \ ,
   \label{eq:tcalc:T0}
\end{gather}
and by using the standard skeleton expansion~\eqref{eq:calc:compton1} for the
Compton tensor. We introduce the spectral density
\begin{gather}
   \rho(\omega|\vec{p}^2)
   =
   \lim_{\vec{p}' \to \vec{p}}
   \sum_{\rho=1}^3 \sum_q g^{\rho\rho}
   \langle \vec{p}',q | \mathcal{J}_\rho(0,\vec{x}) (2\pi) \delta(H-\omega) (2\pi)^3 \delta^3(\vec{P}) \mathcal{J}_\rho(0) | \vec{p},q \rangle
   \ .
\end{gather}
By using the relation between Euclidean and Minkowskian electromagnetic
currents, one sees that both $S$ and $\bar T$ have a spectral representation in
terms of the same spectral density, i.e.
\begin{gather}
   S(k_0,\vec{p}^2)
   =
   \int_0^\infty \frac{d\omega}{2\pi} \int dx_0 \ e^{ik_0x_0} e^{- |x_0| \, [\omega - E(\vec{p})]} \rho(\omega|\vec{p}^2)
   \label{eq:tcalc:S1}
   \ , \\
   \bar{T}(k_0^2,E(\vec{p}) k_0)
   =
   i \int_0^\infty \frac{d\omega}{2\pi} \int dx_0 \ e^{ik_0x_0} e^{-i |x_0| \, [\omega - E(\vec{p})]} \rho(\omega|\vec{p}^2)
   \ ,
   \label{eq:tcalc:T1}
\end{gather}
where, as usual, the Fourier transforms have to be interpreted in the sense of
tempered distributions. The one-pion contribution to the spectral density can be
calculated explicitly by using the electromagnetic Ward identities
\begin{gather}
   \rho(\omega|\vec{p}^2)
   =
   - \frac{8 \vec{p}^2}{2E(\vec{p})} 2\pi \delta\left(\omega-E(\vec{p})\right)
   + \theta\left( \omega-\sqrt{(2m)^2+\vec{p}^2} \right) \rho(\omega|\vec{p}^2)
   \ .
\end{gather}
By plugging this decomposition into eqs.~\eqref{eq:tcalc:S1} and
\eqref{eq:tcalc:T1}, and calculating the integrals in $x_0$, one gets
\begin{gather}
   S(k_0,\vec{p}^2)
   =
   - \frac{8 \vec{p}^2}{2E(\vec{p})} 2\pi \delta(k_0)
   + S_\text{MP}(k_0,\vec{p}^2)
   \label{eq:tcalc:S2}
   \ , \\
   S_\text{MP}(k_0,\vec{p}^2)
   =
   \int_{\sqrt{(2m)^2+\vec{p}^2}}^\infty \frac{d\omega}{2\pi} \ 
   \frac{
   2[\omega - E(\vec{p})]\rho(\omega|\vec{p}^2)
   }{[\omega - E(\vec{p})]^2 + k_0^2}
   \label{eq:tcalc:SMP}
   \ ,
\end{gather}
and analogously
\begin{gather}
   \bar{T}(k_0^2,E(\vec{p}) k_0)
   =
   - i \frac{8 \vec{p}^2}{2E(\vec{p})}  2\pi \delta(k_0)
   + \bar{T}_\text{MP}(k_0,\vec{p}^2)
   \label{eq:tcalc:T2}
   \ , \\
   \bar{T}_\text{MP}(k_0,\vec{p}^2)
   =
   \lim_{\epsilon \to 0^+} \int_{\sqrt{(2m)^2+\vec{p}^2}}^\infty \frac{d\omega}{2\pi} \ 
   \frac{
   2[\omega - E(\vec{p})]\rho(\omega|\vec{p}^2)
   }{
   [ \omega - E(\vec{p}) ]^2 - k_0^2 - i \epsilon }
   \label{eq:tcalc:TMP}
   \ .
\end{gather}
Our goal is to relate $S_\text{MP}(k_0,\vec{p}^2)$ and
$\bar{T}_\text{MP}(k_0,\vec{p}^2)$ by analytic continuation. In general this is problematic since $\bar{T}_\text{MP}(k_0,\vec{p}^2)$ is not an analytic function (it is a tempered distribution). However in the kinematical region
$\vec{p}^2<3m^2$, if $\omega$ is in the integration domain of
eqs.~\eqref{eq:tcalc:SMP} and \eqref{eq:tcalc:TMP},
\begin{gather}
   \omega-E(\vec{p}) \ge \sqrt{(2m)^2+\vec{p}^2} - \sqrt{m^2+\vec{p}^2}
   >
   \frac{m}{2}
   \ ,
\end{gather}
which implies that the limit $\epsilon \to 0^+$ in
eq.~\eqref{eq:tcalc:TMP} is obtained simply by setting $\epsilon=0$ for real
values of $k_0$ with $|k_0|<\frac{m}{2}$. Since $\bar{T}$ and $\bar{T}_\text{MP}$ differ only for a delta in $k_0=0$, the following equality holds for $0 < k_0 < \frac{m}{2}$
\begin{gather}
   \label{eq:tcalc:chain}
   \int_{\sqrt{(2m)^2+\vec{p}^2}}^\infty \frac{d\omega}{2\pi} \ 
   \frac{
   2[\omega - E(\vec{p})]\rho(\omega|\vec{p}^2)
   }{
   [ \omega - E(\vec{p}) ]^2 - k_0^2 }
   =
   \bar{T}_\text{MP}(k_0,\vec{p}^2)
   = \\  \nonumber \hspace{1cm} =
   \bar{T}(k_0^2,E(\vec{p}) k_0)
   =
   \sum_{\rho=1}^3 \left[
   \frac{\mathcal{A}_{\rho\rho}(\bar{p}+\frac{\tilde{k}}{2},\tilde{k})}{
   - k_0^2  - 2 k_0 E(\vec{p})}
   + \frac{\mathcal{A}_{\rho\rho}(\bar{p}-\frac{\tilde{k}}{2},-\tilde{k})}{
   - k_0^2 + 2 k_0 E(\vec{p})}
   + \mathcal{M}_{\rho\rho}(\bar{p},\tilde{k})
   \right]_{ \tilde{k}=(ik_0,\vec{0}) }
   \ ,
\end{gather}
where eq.\eqref{eq:calc:compton1} has been used, again with $\epsilon = 0$,
since the denominators never vanishes for $\vec{p}^2<3m^2$ and $\frac{m}{4} <
k_0 < \frac{m}{2}$. 

The above equality extends by analyticity to a much larger
domain. In particular, note that the expression in terms of $\mathcal{A}$ and
$\mathcal{M}$ is analytic for $| \! \Re k_0| = | \! \Im \tilde{k}_0| < 2m -
E(\vec{p})$ as a consequence of theorem~\ref{theo:ana1}, including $k_0=0$ since
the pole cancels in the sum of the two terms with $\mathcal{A}$. Thus the equality also holds on the imaginary axis. By setting $k_0 \to -i k_0$ in eq.~\eqref{eq:tcalc:chain}, and using eq.~\eqref{eq:tcalc:SMP}, one obtains
\begin{gather}
   S_\text{MP}(k_0,\vec{p}^2)
   =
   \sum_{\rho=1}^3 \left[
   \frac{\mathcal{A}_{\rho\rho}(\bar{p}+\frac{k}{2},k)}{
   k_0^2 + 2 i k_0 E(\vec{p})}
   + \frac{\mathcal{A}_{\rho\rho}(\bar{p}-\frac{k}{2},-k)}{
   k_0^2 - 2 i k_0 E(\vec{p})}
   + \mathcal{M}_{\rho\rho}(\bar{p},k)
   \right]_{ k=(k_0,\vec{0}) }
   \ ,
\end{gather}
which is valid in particular for every real value of $k_0$. In combination with
eq.~\eqref{eq:tcalc:S2}, this yields the desired expression for
$S(k_0,\vec{p}^2)$ in terms of 1PI vertices.

To reach a more compact result we use again the fact that the function inside the squared brackets is analytic at $k_0=0$. Applying the distribution identity
\begin{gather}
   \frac{1}{k_0^2 \pm 2iE(\vec{p})k_0+\epsilon} = \frac{1}{k_0 \pm 2iE(\vec{p})} \frac{\text{PV}}{k_0} + \frac{\pi}{2E(\vec{p})} \delta(k_0)
   \ ,
\end{gather}
where PV stands for Cauchy principal value, and the electromagnetic Ward
identity for the 1PI vertices in the form of equation
\begin{gather}
   \sum_{\rho=1}^3 \mathcal{A}_{\rho\rho}(\bar{p},0) = - 8 m^2 \vec{p}^2 \ ,
\end{gather}
one easily obtains our final formula
\begin{gather}
   S(k_0,\vec{p}^2)
   =
   \lim_{\epsilon \to 0^+}
   \sum_{\rho=1}^3 \left[
   \frac{\mathcal{A}_{\rho\rho}(\bar{p}+\frac{k}{2},k)}{
   k_0^2 + 2 i k_0 E(\vec{p}) +\epsilon}
   + \frac{\mathcal{A}_{\rho\rho}(\bar{p}-\frac{k}{2},-k)}{
   k_0^2 - 2 i k_0 E(\vec{p}) +\epsilon}
   + \mathcal{M}_{\rho\rho}(\bar{p},k)
   \right]_{ k=(k_0,\vec{0}) }
   \label{eq:tcalc:S-skeleton}
   \ .
\end{gather}
Note that the propagator in eq.~\eqref{eq:tcalc:S-skeleton} is defined with a
non-conventional $\epsilon$ prescription, as a consequence of the more involved
Wick rotation.

We go back now to the evaluation of integral~\eqref{eq:calc:DeltaG-0}, for
$\wild=\text{t}$. By using the symmetry properties of the functions
$\mathcal{A}$ and $\mathcal{M}$, this can be equivalently written as
\begin{gather}
   \label{eq:tcalc:int-10}
   \Delta G_\text{t}(x_0 \vert T)
   =
   - \frac{1}{3} \sum_{\rho=1}^3
   \int \frac{d k_0}{2\pi} e^{ik_0x_0}
   \int \frac{d^4 p}{(2 \pi)^4} \, 
   \frac{
   e^{i T p_0 }
   }{ p^2 + m^2 }
   \times \\ \nonumber \hspace{4cm} \times
   \left[
   \frac{\mathcal{A}_{\rho\rho}(p+\frac{k}{2},k)}{(p+k)^2+m^2}
   + \frac{\mathcal{A}_{\rho\rho}(p-\frac{k}{2},-k)}{(p-k)^2+m^2}
   + \mathcal{M}_{\rho\rho}(p,k)
   \right]_{\vec{k}=\vec{0}}
   \ .
\end{gather}
The calculation of the leading exponential proceeds in a very similar way as for
the finite-$L$ contribution, with some important technical differences.

The terms in the integrand with the function $\mathcal{A}$ have four explicit
complex poles in $p_0$, which become pairwise degenerate when $k_0=0$. In order
to avoid these double poles, we deform the integrand by replacing
\begin{gather}
   \frac{1}{(p \pm k)^2+m^2} \to \frac{1}{(p \pm k)^2+m^2+\epsilon} \,,
\end{gather}
with $\epsilon = 0^+$. We stress that here we have shifted with a real rather
than imaginary value and that the key point is \textit{not} to deform the
propagator $(p^2+m^2)^{-1}$. This non-standard approach is required because, for
example, the prescription used in the previous section would shift the double
pole but would not resolve it. Now shifting the $p_0$ integral to $\mathbb{R} +
2 i m^-$, one picks up the poles
\begin{gather}
   p_0 = i \sqrt{m^2+\vec{p}^2}
   \label{eq:tcalc:pole1}
   \ , \\
   p_0 = \mp k_0 + i \sqrt{m^2+\epsilon+\vec{p}^2} \ .
   \label{eq:tcalc:pole2}
\end{gather}
The contribution from the shifted contour is seen to be of order $e^{-2mT}$, and can be neglected. Therefore only the above poles contribute to the leading order of the integral~\eqref{eq:tcalc:int-10}. A straightforward calculation, together with eq.~\eqref{eq:tcalc:S-skeleton}, yields
\begin{multline}
   \label{eq:tcalc:int-20}
   \Delta G_\text{t}(x_0 \vert T)
   = - \frac{1}{3}
   \int_{\vec{p}^2 < 3 m^2} \frac{d^3 p}{(2\pi)^3} \frac{ e^{-T E(\vec{p})} }{ 2 E(\vec{p}) }
   \int \frac{d k_0}{2\pi} e^{ik_0x_0}
   S(k_0,\vec{p}^2)
   \\ 
   - \frac{1}{3} \int_{\vec{p}^2 < 3 m^2} \frac{d^3 p}{(2\pi)^3} \frac{ e^{-T E(\vec{p})} }{ 2 E(\vec{p}) }
   \int \frac{dk_0}{2\pi}
   \frac{\big \{ e^{i (T-x_0) k_0} + e^{i (T+x_0) k_0} \big \} \mathcal{A}(\bar{p}+\frac{k}{2},k)}{k_0^2+2ik_0E(\vec{p})-\epsilon}
   \bigg \vert_{\vec{k}=0}
   + O(e^{-2mT})
   \ .
\end{multline}
The integral in the second line is found to be also of order $e^{-2mT}$ (at
fixed $x_0$), and can thus be dropped. To see this, note first that the function
$\mathcal{A}(\bar{p}+\frac{k}{2},k)$ is analytic in $k_0$ in the strip $-m <
\Im k_0 < 2m - E(\vec{p})$. In particular, note that the upper bound is
positive, i.e.~$2m - E(\vec{p})>0$, because of the restriction $\vec{p}^2 < 3
m^2$. Observe further that the explicit propagator has poles only for $\Im k_0 <
0$. The $k_0$ integral can be shifted to $\mathbb{R} + i [ 2m - E(\vec{p}) ]^-$
and, after the shift, the exponentials in the integrand become
\begin{align}
\begin{split}
 \hspace{-20pt}  e^{-T E(\vec{p})} e^{i (T \mp x_0) k_0}
\ \ \     \to \  \ \ 
   e^{-T E(\vec{p})} e^{-(T \mp x_0) [ 2m - E(\vec{p}) ]} e^{i (T-x_0) k_0}
 &  =
   e^{-2m T } e^{\pm  [ 2m - E(\vec{p}) ] x_0} e^{i (T-x_0) k_0} \\
  &  =
   O(e^{-2mT})
   \ ,
   \end{split}
\end{align}
where we have used $E(\vec{p}) \ge m$, and the fact that $x_0$ is kept constant.

Finally we note that the restriction over the $\vec{p}$ integration domain of
the first integral in eq.~\eqref{eq:tcalc:int-20} can be dropped up to an error
of the same order that we are neglecting. We deduce eq.~\eqref{eq:stat:DeltaGt}  
\begin{gather}
   \label{eq:tcalc:WP-leading}
   \Delta G_\text{t}(x_0|T)
   = - \frac{1}{3}
   \int \frac{d^3 p}{(2\pi)^3} \frac{ e^{-T E(\vec{p})} }{ 2 E(\vec{p}) }
   \int \frac{d k_0}{2\pi} e^{ik_0x_0}
   S(k_0,\vec{p}^2)
   + O(e^{-2mT}) \nonumber
   \ .
\end{gather}
 As with $\Delta G_{\text{s}}(x_0,L)$, we have abused notation in section \ref{sec:statement} by absorbing the $O(e^{-2mT})$ term here into the relation between $\Delta G(x_0 \vert T, \infty)$ and $\Delta G_\text{t}(x_0|T)$ such that eq.~\eqref{eq:stat:DeltaGt} does not contain the subleading correction explicitly.

\acknowledgments
We warmly thank Mattia Bruno, Maarten Golterman, Christoph Lehner, and Harvey Meyer for useful discussions. We also acknowledge the CERN-TH Institute ``Advances in Lattice Gauge Theory'', which provided the opportunity to make significant progress on this manuscript. We thank Davide Giusti and Silvano Simula for providing comments on an earlier version of this work, including numerical estimates based in an alternative method for comparison. Finally we are grateful to 
Laurent Lellouch, Finn Stokes and Kalman Szabo, for an incredibly careful reading of this work in which they identified typos and a missing factor of $-(1/3)$ in the finite-$T$ correction that we previously reported.

\appendix

\section{Proofs of theorems concerning gauge fields on graphs}
\label{app:epsilon}

In this appendix we provide the proofs of all theorems referenced in
section~\ref{sec:diag}. Ultimately, the goal of these theorems is to study the
properties of the functions $\hat{\epsilon}_0(n_0)$ and
$\hat{\epsilon}_\text{s}(\vec{n})$ defined in eqs.~\eqref{eq:hatepsilon-0}
and~\eqref{eq:hatepsilon-s}, and to identify the gauge fields which give the
leading finite-volume corrections of $\amu$. We stress again that a large part of the core ideas in the following analysis are not original and have been borrowed from \cite{\MartinStable}.

A number of preparatory results need to be established and are organized in
various subsections. Before being able to study
$\hat{\epsilon}_\text{s}(\vec{n})$, we need to introduce the auxiliary
quantity
\begin{gather}
   \hat{\epsilon}_k(n_k) = 
   \min_{x_k \text{ with } x_k(a)=0}
   \sum_{\ell \in \mathcal{L}} | \delta x_k(\ell) + n_k(\ell) |
   \ , \label{eq:hatepsilon-k}
\end{gather}
in which one spatial direction is considered at a time. The solution of the
minimization problem appearing in the definition of $\hat{\epsilon}_\mu(n_\mu)$
is constructed in section~\ref{app:subsec:epshat}. After introducing the
concept of axial gauge for the gauge fields $n_\mu$ in
subsection~\ref{app:subsec:axial}, we characterize completely the set of
possible values for the functions $\hat{\epsilon}_\mu(n_\mu)$ in
subsection~\ref{app:susec:onehalf}. The gauge fields corresponding to the lower
possible values of $\hat{\epsilon}_k(n_k)$, $\hat{\epsilon}_\text{s}(\vec{n})$
and $\hat{\epsilon}_0(n_0)$ are characterized in the
subsection~\ref{app:subsec:sqrt} and
\ref{app:subsec:one-0} respectively.

In the following subsections, we also need some more graph-theory concepts
which we introduce here. We assume throughout that $\mathcal{G}$ is a connected
graph.

\bigskip

\begin{figure}
   \centering
  \includegraphics[scale=0.7]{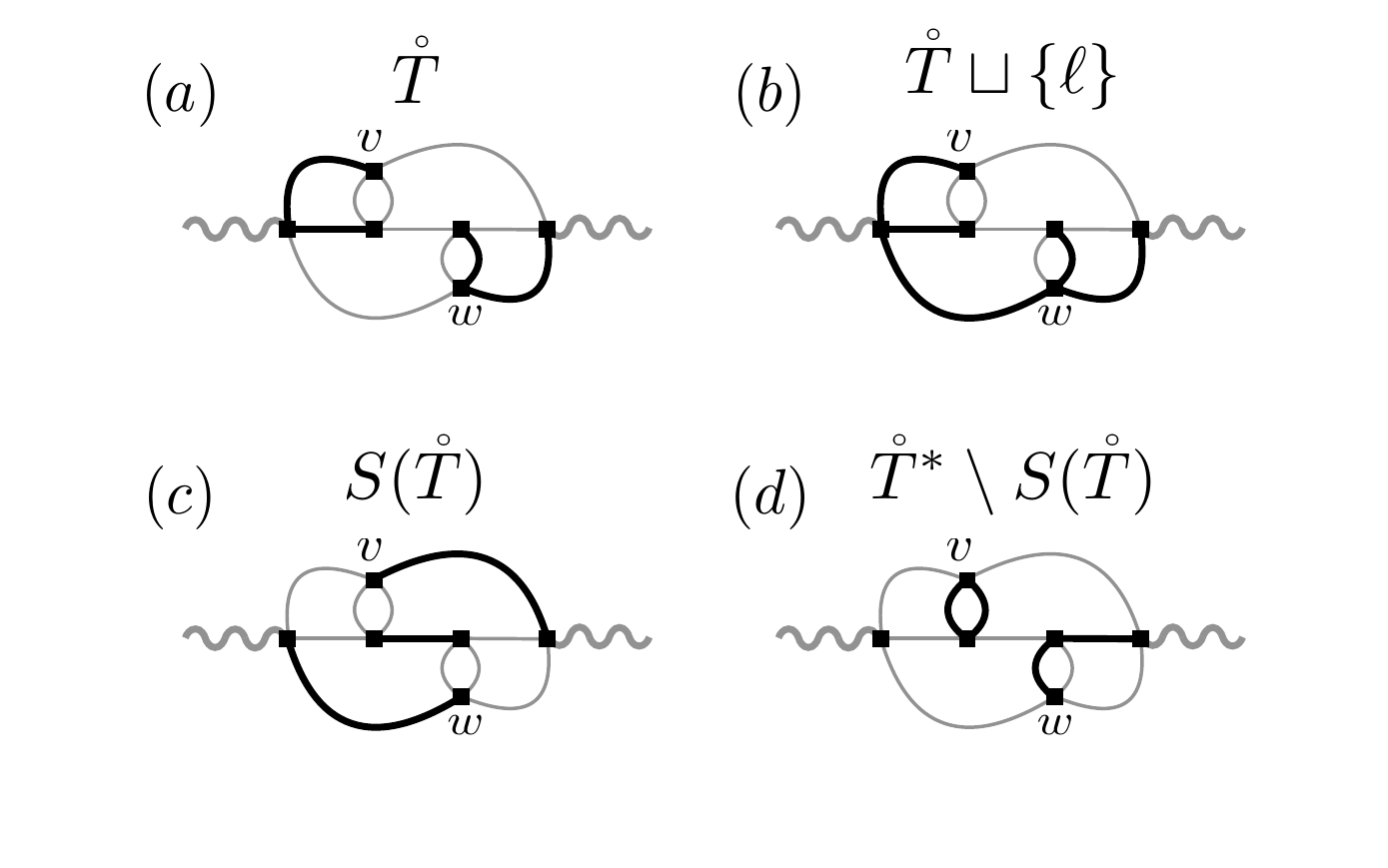} 
   \vspace{-15pt}
   \caption{Examples of sets entering the definition of a 2-tree. Figure (a) shows an example 2-tree $\mathring{T}$, and (b) illustrates the defining property, that at least one line $\ell$ exists in the complement, such that $\mathring{T}\sqcup \{\ell\}$ is a tree. Note that (b) is indeed a maximal subset with no loops, and that a unique path exists from $v$ to $w$. Figure (c) illustrates the cut-set $S(\mathring{T})$ assigned to $\mathring{T}$, and (d) illustrates the complement. As described in the text, each $ \ell \in \mathring{T}^* \setminus S(\mathring{T})$ corresponds to a unique loop in $\mathring{T} \sqcup \{\ell\}$.
   }
   \label{fig:twoTree}
\end{figure}

\textit{Trees.} A tree $T$ is a maximal subset of $\mathcal{L}$ containing no
loops. We denote by $T^*$ the complement of $T$ in $\mathcal{L}$. It can be
shown that, given any pair of vertices $v$ and $w$, there is exactly one path in
$T$ from $v$ to $w$ (\cite{nakanishi1971graph}, theorem 2-5). It follows that
for every line $\ell \in T^*$ there is a unique loop in $T \sqcup \{\ell\}$
(\cite{nakanishi1971graph}, theorem 2-22). The set of such loops is denoted by
$\mathbf{C}(T)$.

\textit{2-trees.} A 2-tree $\mathring{T}$ is a subset of lines with the property
that a line $\ell \in \mathring{T}^*$ exists such that $\mathring{T} \sqcup
\{\ell\}$ is a tree. [See figure~\ref{fig:twoTree}.] Note that there are at least two distinct vertices $v$
and $w$ such that there is no path from $v$ to $w$ in $\mathring{T}$. In this
case we say that $\mathring{T}$ disconnects $v$ and $w$. For each 2-tree
$\mathring{T}$, there is a unique cut-set $S(\mathring{T})$ with the property
that $S(\mathring{T}) \cap \mathring{T}$ is empty (\cite{nakanishi1971graph}, theorem
2-9). It is easy to show that $S(\mathring{T})$ is the set of lines $\ell \in
\mathring{T}^*$ with the property that $\mathring{T} \sqcup \{\ell\}$ is a tree.
If $\mathring{T}$ disconnects $v$ and $w$, for every line $\ell \in
S(\mathring{T})$ there is a unique path from $v$ to $w$ in $\mathring{T} \sqcup
\{\ell\}$. The set of such paths is denoted by $\mathbf{P}_{vw}(\mathring{T})$.
For every line $\ell \in \mathring{T}^* \setminus S(\mathring{T})$ there is a
unique loop in $\mathring{T} \sqcup \{\ell\}$. The set of such loops is denoted
by $\mathbf{C}(\mathring{T})$.

\textit{Parallel transports.} Given the oriented path or loop $P$, the parallel
transport along $P$ is defined by
\begin{gather}
   n(P) = \sum_{\ell \in P} [P:\ell] n(\ell) \ .
\end{gather}
If $C$ is a loop, then $n(C)$ is referred to as Wilson loop along $C$. If $P$ is
a path from $a$ to $b$, then $n_0(P)$ is referred to as Wilson line along $P$.
Note that Wilson loops and Wilson lines are gauge invariant.

\subsection{Solving the minimization problem for \texorpdfstring{$\hat{\epsilon}_\mu(n_\mu)$}{epsilon\_mu(n\_mu)}}
\label{app:subsec:epshat}

The functions $\hat{\epsilon}_\mu(n_\mu)$ defined by
eqs.~\eqref{eq:hatepsilon-0} and \eqref{eq:hatepsilon-k} determine the
asymptotic behaviour of the Feynamn integrals $\mathcal{I}_{T,L}(\mathcal{D},n)$ at
large $L$. Theorem~\ref{theo:epshat} at the end of this subsection provides a
detailed characterization of the solution of the minimization problem that
defines $\hat{\epsilon}_\mu(n_\mu)$. In order to prove this theorem, we will
need to discuss two lemmas.

\bigskip

\begin{lemma} \label{lemma:minumum-k}
   Let us fix $k \in \{1,2,3\}$. An assignment $\mathcal{V} \ni v \mapsto
   \bar{x}_k(v) \in \mathbb{R}$ of numbers to each vertex exists, with the
   following properties:
   \begin{enumerate}[topsep=0pt,itemsep=1ex]
      \item $\bar{x}_k$ satisfies
      \begin{gather}
         \bar{x}_k(a)=0
         \ .
      \end{gather}
      
      \item $\bar{x}_k$ realizes the minimum in eq.~\eqref{eq:hatepsilon-k}, i.e.
      \begin{gather}
         \hat{\epsilon}_k(n_k) = 
         \sum_{\ell \in \mathcal{L}} \left| \bar{x}_k[f(\ell)] - \bar{x}_k[i(\ell)] + n_k(\ell) \right|
         \ .
      \end{gather}
      
      \item A tree $T_k$ exists such that
      \begin{gather}
         \bar{x}_k[f(\ell)] - \bar{x}_k[i(\ell)] + n_k(\ell) = 0 \,,
      \end{gather}
      for every $\ell \in T_k$.
   \end{enumerate}
\end{lemma}

\begin{proof}
   Consider an assignment $\bar{x}_k$ satisfying \textit{1} and \textit{2},
   such that the number of elements in the set
   \begin{gather}
      Z(\bar{x}_k) = \{ \ell \in \mathcal{L} \text{ s.t. }\bar{x}_k[f(\ell)] - \bar{x}_k[i(\ell)] + n_k(\ell)=0 \} \,,
   \end{gather}
   is the maximum possible. It is easy to show that such $\bar{x}_k$ exists, and
   we want to prove that it satisfies \textit{3}. To do so, it is enough to
   prove that $Z(\bar{x}_k)$ contains a tree, i.e.~that for any vertex $v \neq
   a$ a path in $Z(\bar{x}_k)$ exists from $a$ to $v$.
   
   Let us work by contradiction, and assume that a vertex $\bar{v} \neq a$
   exists such that no path in $Z(\bar{x}_k)$ exists connecting $a$ to
   $\bar{v}$. We want to show that it is possible to construct a new
   $\bar{x}'_k$ satisfying \textit{1} and \textit{2}, such that the number of
   elements in $Z(\bar{x}'_k)$ is strictly larger than the number of elements in
   $Z(\bar{x}_k)$, which is in contradiction with the maximality of
   $Z(\bar{x}_k)$.
   
   Consider the set $\mathcal{V}_1$ containing $a$ and all vertices connected to
   $a$ by paths in $Z(\bar{x}_k)$, and let $\mathcal{V}_2$ be its complement in
   $\mathcal{V}$. Obviously $\bar{v} \in \mathcal{V}_2$. Let $\mathcal{L}_1$
   (resp. $\mathcal{L}_2$) be the set of lines with both endpoints in
   $\mathcal{V}_1$ (resp. $\mathcal{V}_2$), and let $\tilde{\mathcal{L}}$ be the
   set of remaining lines. Since the considered graph is connected,
   $\tilde{\mathcal{L}}$ is not empty. With no loss of generality, we assume
   that, for every $\ell \in \tilde{\mathcal{L}}$, then $i(\ell) \in
   \mathcal{V}_1$ and $f(\ell) \in \mathcal{V}_2$.

   We define the function
   \begin{gather}
      h(x_k)
      =
      \sum_{\ell \in \mathcal{L}}
      | x_k[f(\ell)] - x_k[i(\ell)] + n_k(\ell) |
      \ ,
   \end{gather}
   in terms of which
   \begin{gather}
      \hat{\epsilon}_k(n_k) = \min_{x_k \text{ with } x_k(a)=0 } h(x_k) = h(\bar{x}_k)
      \ .
   \end{gather}

   Given a number $\alpha$ (which does not depend on the vertex), consider the
   following function
   \begin{gather}
      g(\alpha) = h(\bar{x}_k - \alpha \chi_{\mathcal{V}_2})
      = \\ \nonumber \qquad =
      \sum_{j=1,2} \sum_{\ell \in \mathcal{L}_j}
      | \bar{x}_k[f(\ell)] - \bar{x}_k[i(\ell)] + n_k(\ell) |
      + \sum_{\ell \in \tilde{\mathcal{L}}}
      | \!-\alpha + \bar{x}_k[f(\ell)] - \bar{x}_k[i(\ell)] + n_k(\ell) |
      \ ,
   \end{gather}
   where the characteristic function $\chi_{\mathcal{V}_2}(v)$ is equal to 1 if
   $v \in \mathcal{V}_2$ and 0 otherwise. $g$ is a continuous piecewise linear
   function, bounded from below. A value $\bar{\alpha}$ exists such that
   $\bar{\alpha}$ is an absolute minimum for $g$, and $g$ is not differentiable
   in $\bar{\alpha}$. The latter condition implies that a line $\tilde{\ell} \in
   \tilde{\mathcal{L}}$ exists such that
   \begin{gather}
      \bar{\alpha} = \bar{x}_k[f(\tilde{\ell})] - \bar{x}_k[i(\tilde{\ell})] + n_k(\tilde{\ell}) \ .
   \end{gather}
   We define
   \begin{gather}
      \bar{x}'_k = \bar{x}_k - \bar{\alpha} \chi_{\mathcal{V}_2}
      \ .
   \end{gather}
   Obviously $\bar{x}'$ satisfies \textit{1}. Then we observe that
   \begin{gather}
 h(\bar{x}'_k) \geq  \min_{x_k \text{ with } x_k(a)=0 } h(x_k)=  \hat{\epsilon}_k(n_k) = h(\bar{x}_k)  = g(0) \geq \min_\alpha g(\alpha) = g(\bar{\alpha})= h(\bar{x}'_k) \,,
\end{gather}
implying
\begin{gather}
 \hat{\epsilon}_k(n_k)  = h(\bar{x}'_k) \,,
\end{gather}
   i.e.~$\bar{x}'_k$ satisfies \textit{2}, and also
   \begin{gather}
      \bar{x}'_k[f(\ell)] - \bar{x}'_k[i(\ell)] + \bar{x}'_k(\ell)
      =
      \begin{cases}
         \bar{x}_k[f(\ell)] - \bar{x}_k[i(\ell)] + \bar{x}_k(\ell)
         \qquad & \text{if } \ell \in \mathcal{L}_1 \sqcup \mathcal{L}_2 \\
         0 & \text{if } \ell = \tilde{\ell} \in \tilde{\mathcal{L}}
      \end{cases}
      \ ,
   \end{gather}
   which implies $Z(\bar{x}_k) \sqcup \{ \tilde{\ell} \} \subseteq
   Z(\bar{x}'_k)$. This is contradiction with the maximality of $Z(\bar{x}_k)$.
\end{proof}

\bigskip

\begin{lemma} \label{lemma:minumum-0}
   An assignment $\mathcal{V} \ni v \mapsto \bar{x}_0(v) \in \mathbb{R}$ of numbers to each vertex exists, with the following properties:
   \begin{enumerate}[topsep=0pt,itemsep=1ex]
      \item $\bar{x}_0$ satisfies
      \begin{gather}
         \bar{x}_0(a)=0
         \text{ and }
         0 \le \bar{x}_0(b) \le \tfrac{1}{2}
         \ .
      \end{gather}
      
      \item $\bar{x}_0$ realizes the minimum in eq.~\eqref{eq:hatepsilon-0}, i.e.
      \begin{gather}
         \hat{\epsilon}_0(n_0) = 
         \sum_{\ell \in \mathcal{L}} \left| \bar{x}_0[f(\ell)] - \bar{x}_0[i(\ell)] + n_0(\ell) \right|
         \ .
      \end{gather}
      
      \item A 2-tree $\mathring{T}$ which disconnects $a$ and $b$ exists such that
      \begin{gather}
         \bar{x}_0[f(\ell)] - \bar{x}_0[i(\ell)] + n_0(\ell) = 0
      \end{gather}
      for every $\ell \in \mathring{T}$.
   \end{enumerate}
\end{lemma}

\begin{proof}
   Consider an assignment $\bar{x}_0$ satisfying \textit{1} and \textit{2},
   such that the number of elements in the set
   \begin{gather}
      Z(\bar{x}_0) = \{ \ell \in \mathcal{L} \text{ s.t. }\bar{x}_0[f(\ell)] - \bar{x}_0[i(\ell)] + n_0(\ell)=0 \} \,,
   \end{gather}
   is the maximum possible. It is easy to show that such an $\bar{x}_0$ exists, and
   we want to prove that it satisfies \textit{3}. To do so, it is enough to prove
   that for any vertex $v \neq a,b$ either a path from $a$ to $v$ or a path from
   $b$ to $v$ exists in $Z(\bar{x}_0)$.
   
   Let us work by contradiction, and assume a vertex $\bar{v}$
   such that no path in $Z(\bar{x}_0)$ exists connecting $a$ to $\bar{v}$ and no
   path in $Z(\bar{x}_0)$ exists connecting $b$ to $\bar{v}$. We want to show
   that it is possible to construct a new $\bar{x}'_0$ satisfying \textit{1} and
   \textit{2}, such that the number of elements in $Z(\bar{x}'_0)$ is strictly
   larger than the number of elements in $Z(\bar{x}_0)$, 
   contradicting  the maximality of $Z(\bar{x}_0)$.
   
   Consider the set $\mathcal{V}_2$ containing $\bar{v}$ and all vertices
   connected to $\bar{v}$ by paths in $Z(\bar{x}_0)$, and let $\mathcal{V}_1$ be
   its complement in $\mathcal{V}$. Obviously $a,b \in \mathcal{V}_1$. From this
   point, the construction of $\bar{x}'_0$ follows the same steps as the
   construction of $\bar{x}'_k$ in the proof of lemma \ref{lemma:minumum-k}, and
   we will not repeat it here. As in the other lemma, $\bar{x}_0$ satisfies
   \textit{1} and \textit{2} and $|Z(\bar{x}_0)| < |Z(\bar{x}'_0)|$, %
  giving the desired contradiction.
\end{proof}

\bigskip

\begin{theorem} \label{theo:epshat}
   For every $k=1,2,3$, a tree $T_k$ exists such that
   \begin{gather}
      \hat{\epsilon}_k(n_k) = \sum_{C \in \mathbf{C}(T_k)} |n_k(C)| \ .
      \label{eq:theo:epshat:k}
   \end{gather}
   A 2-tree $\mathring{T}_0$ which disconnects $a$ and $b$ exists such that
   \begin{gather}
      \hat{\epsilon}_0(n_0) = \sum_{C \in \mathbf{C}(\mathring{T}_0)} |n_0(C)|
      + \sum_{P \in \mathbf{P}_{ab}(\mathring{T}_0)} |n_0(P)|
      +  \min \bigg \{0 , \, \frac{ p_0^+(n|\mathring{T}_0) + p_0^0(n|\mathring{T}_0) - p_0^-(n|\mathring{T}_0) }{2} \bigg\}
      \ ,
      \label{eq:theo:epshat:0}
   \end{gather}
   where $p_0^+(n|\mathring{T}_0)$, $p_0^0(n|\mathring{T}_0)$,
   $p_0^-(n|\mathring{T}_0)$ are the number of paths $P$ in
   $\mathbf{P}_{ab}(\mathring{T}_0)$ such that $n_0(P) \ge 1$, $n_0(P) = 0$,
   $n_0(P) \le -1$ respectively.

   The following loose but useful bounds hold
   \begin{gather}
      \hat{\epsilon}_k(n_k) \ge c_k(n|T_k) \ ,
      \label{eq:theo:epshat:bound:k}
   \end{gather}
   where $c_k(n|T_k)$ is the number of loops $C$ in $\mathbf{C}(T_k)$ with $n_k(C) \neq 0$, and
   \begin{gather}
      \hat{\epsilon}_0(n_0) \ge c_0(n|\mathring{T}_0)
      + p_0^+(n|\mathring{T}_0) + \frac{1}{2} p_0^-(n|\mathring{T}_0)
      +   \min \bigg \{ \frac{p_0^-(n|\mathring{T}_0)}{2} , \frac{ p_0^+(n|\mathring{T}_0) + p_0^0(n|\mathring{T}_0) }{2} \bigg \}
      \ ,
      \label{eq:theo:epshat:bound:0}
   \end{gather}
   where $c_0(n|\mathring{T}_0)$ is the number of loops $C$ in $\mathbf{C}(\mathring{T}_0)$ with $n_0(C) \neq 0$.
\end{theorem}

\bigskip

\begin{proof}
   Let $\bar{x}_k$ and $T_k$ be as in lemma~\ref{lemma:minumum-k}. Given $\ell
   \in T_k^*$, let $C$ be the only loop in $T_k \sqcup \{\ell\}$. Then
   \begin{gather}
      [ C : \ell ] \big \{ \bar{x}_k[f(\ell)] - \bar{x}_k[i(\ell)] + n_k(\ell) \big  \}
      = \nonumber  \\ \qquad =
      \sum_{\ell' \in C} [ C : \ell' ] \, \big \{ \bar{x}_k[f(\ell')] - \bar{x}_k[i(\ell')] + n_k(\ell') \big \}
      =
      \sum_{\ell' \in C} [ C : \ell' ] n_k(\ell')
      =
      n_k(C)
      \ ,
      \label{eq:theo:epshat:k:C}
   \end{gather}
   where the first equality is based on the fact that all terms in the sum
   vanish except for $\ell'=\ell$, the second equality is based on the
   observation that all the $\bar{x}(v)$'s cancel in pairs. This implies
   \begin{gather}
      \hat{\epsilon}_k(n_k)
      = 
      \sum_{\ell \in T_k^*} \left| \bar{x}_k[f(\ell)] - \bar{x}_k[i(\ell)] + n_k(\ell) \right|
      =
      \sum_{C \in \mathbf{C}(T_k)} |n_k(C)|
      \ .
   \end{gather}
   
   Define the function
   \begin{gather}
      h_0(x_0)
      =
      \sum_{\ell \in \mathcal{L}}
      | x_0[f(\ell)] - x_0[i(\ell)] + n_0(\ell) |
      \ ,
   \end{gather}
   and let $\bar{x}_0$ and $\mathring{T}_0$ be as in
   lemma~\ref{lemma:minumum-0}, then
   \begin{gather}
      \hat{\epsilon}_0(n_0) = \min_{\substack{ x_0 \text{ with } x_0(a)=0 \\ \text{and } x_0(b) \in [0,1/2] } } h_0(x_0) = h_0(\bar{x}_0)
      \ .
      \label{eq:theo:epshat:min}
   \end{gather}
   Let $\mathcal{W}$ be the set of $b$ and all vertices connected to $b$ by
   paths in $\mathring{T}_0$. Given a number $\alpha$ (which does not depend on
   the vertex), consider the following family of coordinate assignments
   \begin{gather}
      x_0^\alpha = \bar{x}_0 + (\alpha-\bar{x}_0(b)) \chi_{\mathcal{W}}
      \ .
   \end{gather}
   Note that
   \begin{gather}
      x_0^\alpha(a) = 0
      \ , \qquad 
      x_0^\alpha(b) = \alpha
      \ .
   \end{gather}
   When $\alpha$ is varied in $[0,1/2]$, $x_0^\alpha$ spans a one-parameter
   family of coordinate assignments $x_0$ satisfying $x_0^\alpha(a)=0$ and $0
   \le x_0^\alpha(b) \le 1/2$. Using eq.~\eqref{eq:theo:epshat:min}, it follows
   that
   \begin{gather}
      \hat{\epsilon}_0(n_0)
      =
      h_0(\bar{x}_0)
      =
      [ h_0(x_0^\alpha) ]_{\alpha = \bar{x}_0(b)}
      =
      \min_{\alpha \in [0,1/2]} h_0(x_0^\alpha)
      \ .
      \label{eq:theo:epshat:minalpha}
   \end{gather}
   We want to find now a more explicit representation for $h_0(x_0^\alpha)$. Given
   $\ell \in \mathring{T}_0^* \setminus S(\mathring{T}_0)$, let $C$ be the only
   loop in $\mathring{T}_0 \sqcup \{\ell\}$. Note that two possibilities are
   given: either both endpoints of $\ell$ are in $\mathcal{W}$, or both
   endpoints of $\ell$ are in $\mathcal{W}^*$. In both cases:
   \begin{gather}
      [ C : \ell ]  \big \{ x_0^\alpha[f(\ell)] - x_0^\alpha[i(\ell)] + n_0(\ell) \big \}
      =
      [ C : \ell ] \big \{ \bar{x}_0[f(\ell)] - \bar{x}_0[i(\ell)] + n_0(\ell) \big \}
      =
      n_0(C)
      \ ,
      \label{eq:theo:epshat:0:C}
   \end{gather}
   where the calculation poceeds as in eq.~\eqref{eq:theo:epshat:k:C}. Given
   $\ell \in S(\mathring{T}_0)$, let $P$ the only path from $a$ to $b$ in
   $\mathring{T}_0 \sqcup \{\ell\}$. One endpoint of $\ell$ is in
   $\mathcal{W}$ and the other is in $\mathcal{W}^*$, therefore
   \begin{multline}
      [ P : \ell ]  \big \{ x_0^\alpha[f(\ell)] - x_0^\alpha[i(\ell)] + n_0(\ell) \big \}
      =
      \alpha - \bar{x}_0(b)
      + [ P : \ell ] \big \{ \bar{x}_0[f(\ell)] - \bar{x}_0[i(\ell)] + n_0(\ell) \big \}
      = \\ \qquad \qquad =
      \alpha - \bar{x}_0(b)
      + \sum_{\ell' \in P} [ P : \ell' ] \, \big \{ \bar{x}_0[f(\ell')] - \bar{x}_0[i(\ell')] + n_0(\ell') \big \}
      = \\ =%
      \alpha + \sum_{\ell' \in P} [ P : \ell' ] n_0(\ell')
      =
      \alpha + n_0(P)
      \ ,
      \label{eq:theo:epshat:0:P}
   \end{multline}
   where the first equality uses the definition of $x_0^\alpha$, the second
   equality is based on the fact that all terms in the sum vanish except for
   $\ell'=\ell$, the third equality is based on the observation that all the
   $\bar{x}_0(v)$'s cancel in pairs, except for $\bar{x}_0(a)=0$. Putting
   everything together
   \begin{multline}
      \hat{\epsilon}_0(n_0)
      =
      \min_{\alpha \in [0,1/2]} h_0(x_0^\alpha)
      = 
      \min_{\alpha \in [0,1/2]} \sum_{\ell \in \mathring{T}_0^*} \left| x_0^\alpha[f(\ell)] - x_0^\alpha[i(\ell)] + n_0(\ell) \right|
    = \\ = %
      \sum_{C \in \mathbf{C}(\mathring{T}_0)} |n_0(C)|
      + \min_{\alpha \in [0,1/2]} \sum_{P \in \mathbf{P}_{ab}(\mathring{T}_0)} |\alpha + n_0(P)|
      \ .
      \label{eq:theo:epshat:0:alpha}
   \end{multline}
   We note that, for $0 \le \alpha \le \tfrac{1}{2}$,
   \begin{gather}
      |\alpha + n_0(P)|
      =
      \begin{cases}
         \alpha + n_0(P) = \alpha + |n_0(P)| \qquad & \text{if } n_0(P) \ge 0 \\
         -\alpha - n_0(P) = -\alpha + |n_0(P)| \qquad & \text{if } n_0(P) \le -1
      \end{cases}
      \ ,
   \end{gather}
   therefore
   \begin{multline}
      \min_{\alpha \in [0,1/2]} \sum_{P \in \mathbf{P}_{ab}(\mathring{T}_0)} |\alpha + n_0(P)|
      =
      \sum_{P \in \mathbf{P}_{ab}(\mathring{T}_0)} |n_0(P)|
      + \min_{\alpha \in [0,1/2]} \alpha [ p_0^+(n|\mathring{T}_0) + p_0^0(n|\mathring{T}_0) - p_0^-(n|\mathring{T}_0) ]
     = \\ =
      \sum_{P \in \mathbf{P}_{ab}(\mathring{T}_0)} |n_0(P)|
      + \frac{1}{2} \min \left\{0 , [ p_0^+(n|\mathring{T}_0) + p_0^0(n|\mathring{T}_0) - p_0^-(n|\mathring{T}_0) ] \right\}
      \ .
   \end{multline}
   Eq.~\eqref{eq:theo:epshat:0} follows from this and eq.~\eqref{eq:theo:epshat:0:alpha}.
   
   The bounds follow trivially from the observations
   \begin{gather}
      \sum_{C \in \mathbf{C}(T_k)} |n_k(C)| \ge c_k(n|T_k)
      \ , \\
      \sum_{C \in \mathbf{C}(\mathring{T}_0)} |n_0(C)| \ge c_0(n|\mathring{T}_0)
      \ , \\
      \sum_{P \in \mathbf{P}_{ab}(\mathring{T}_0)} |n_0(P)| \ge p_0^+(n|\mathring{T}_0) + p_0^-(n|\mathring{T}_0)
      \ .
   \end{gather}
\end{proof}
 
\subsection{Axial gauge: definition and applications}
\label{app:subsec:axial}

The axial gauge provides a convenient way to select a representative gauge field
per gauge orbit. In this subsection we provide the definition and the proof of
existence and uniqueness of the axial gauge (theorem~\ref{theo:axial}). Using
the axial gauge as a tool, we provide a characterization of gauge equivalence
(theorems~\ref{theo:gauge-equivalence-k} and~\ref{theo:gauge-equivalence-0}) in
terms of a minimal set of gauge-invariant quantities, i.e.~certain Wilson loops
and Wilson lines. Finally, corollary~\ref{corollary:epshat:axial} is a
reformulation of theorem~\ref{theo:epshat} in axial gauge, which will turn out
to be particularly useful.

\bigskip

\textit{Axial gauge.} The following definitions are given:
\begin{enumerate}[topsep=0pt,itemsep=1ex]
   \item Given a tree $T$, the gauge field $n_k$ with $k \in \{1,2,3\}$ is
   said to be in axial gauge with respect to $T$ if and only if $n_k(\ell)=0$ for any $\ell
   \in T$.
   \item Given a 2-tree $\mathring{T}$ which disconnects $a$ and $b$, the gauge
   field $n_0$ is said to be in axial gauge with respect to $\mathring{T}$ iff
   $n_0(\ell)=0$ for any $\ell \in \mathring{T}$.
   \item Consider a fourplet $\mathcal{T} = (\mathring{T}_0,T_1,T_2,T_3)$, where
   $\mathring{T}_0$ is a 2-tree which disconnects $a$ and $b$ and $T_k$ with
   $k=1,2,3$ are trees. The gauge field $n$ is said to be in axial gauge with
   respect to $\mathcal{T}$ if and only if $n_0$ is in axial gauge with respect to
   $\mathring{T}_0$ and $n_k$ is in axial gauge with respect to $T_k$ for
   $k=1,2,3$.
\end{enumerate}

\bigskip

\begin{theorem} \label{theo:axial}
   Consider a fourplet $\mathcal{T} = (\mathring{T}_0,T_1,T_2,T_3)$, where
   $\mathring{T}_0$ is a 2-tree which disconnects $a$ and $b$ and $T_k$ with
   $k=1,2,3$ are trees. Given a gauge field $n$, there is a unique gauge field
   $n^{\mathcal{T}}$ which is gauge equivalent to $n$ and is in axial gauge with
   respect to $\mathcal{T}$.
\end{theorem}

\begin{proof}
   Consider an admissible gauge transformation $\lambda$, and let $n^\lambda$ be
   the result of transforming $n$ with $\lambda$.
   Fix $k \in \{1,2,3\}$. Given a vertex $v$, a unique path $P_a(v)$ from $a$ to
   $v$ exists in $T_k$. One easily checks that $n^\lambda_k$ is in axial gauge
   with respect to $T_k$ if and only if
   \begin{gather}
      \lambda_k(v) = - \sum_{\ell \in P_a(v)} [P_a(v) : \ell] \, n(\ell)
      \ .
   \end{gather}
   
   Let $\mathcal{V}_a$ (resp. $\mathcal{V}_b$) be the set of vertices connected
   to $a$ (resp. $b$) by paths in $\mathring{T}_0$. Then $\mathcal{V} =
   \mathcal{V}_a \sqcup \mathcal{V}_b$. If $v \in \mathcal{V}_a$ a unique path
   $P_a(v)$ from $a$ to $v$ exists in $\mathring{T}_0$, and if $v \in
   \mathcal{V}_b$ a unique path $P_b(v)$ from $b$ to $v$ exists in
   $\mathring{T}_0$.  One easily checks that $n^\lambda_0$ is in axial gauge
   with respect to $\mathring{T}_0$ if and only if
   \begin{gather}
      \lambda_0(v) =
      \begin{dcases}
         - \sum_{\ell \in P_a(v)} [P_a(v) : \ell] \, n(\ell)
         \qquad & \text{if } v \in \mathcal{V}_a
         \\
         - \sum_{\ell \in P_b(v)} [P_b(v) : \ell] \, n(\ell)
         \qquad & \text{if } v \in \mathcal{V}_b
      \end{dcases}
      \ .
   \end{gather}
   
   This is enough to prove existence and uniqueness of $n^{\mathcal{T}}$.
\end{proof}

\bigskip

\begin{theorem} \label{theo:gauge-equivalence-k}
   Consider a tree $T$ and two gauge fields $n^1_k$ and $n^2_k$. The following statements are equivalent:
   \begin{enumerate}[topsep=0pt,itemsep=1ex]
      \item $n^1_k$ and $n^2_k$ are gauge equivalent;
      \item $n^1_k(C) = n^2_k(C)$ for every $C \in \mathbf{C}(T)$.
   \end{enumerate}
\end{theorem}

\begin{proof}
   Note that the implication $1 \Rightarrow 2$ follows trivially from gauge
   invariance of Wilson loops. We only need to prove the implication $2
   \Rightarrow 1$.
   
   For $j=1,2$, let $\tilde{n}^j_k$ be the only representative in
   $[{n}^j_k]$ which is in axial gauge with respect to $T$. Existence and
   uniqueness of $\tilde{n}^j_k$ are guaranteed by theorem \ref{theo:axial}. In
   order to prove that $n^1_k$ and $n^2_k$ are gauge equivalent, it is enough to
   prove that $\tilde{n}^1_k=\tilde{n}^2_k$.
   
   By definition of axial gauge $\tilde{n}^1_k(\ell)=0=\tilde{n}^2_k(\ell)$ if
   $\ell \in T$.
   
   Consider $\ell \in T^*$, then a unique loop $C$ exists in $T \sqcup
   \{\ell\}$, and
   \begin{gather}
      n^j_k(C) = \tilde{n}^j_k(C) = [C : \ell] \tilde{n}^j_k(\ell) \ ,
   \end{gather}
   where we have used the gauge invariance of Wilson loops. By definition $C \in
   \mathbf{C}(T)$, and by hypothesis
   \begin{gather}
      [C : \ell] \tilde{n}^1_k(\ell) = n^1_k(C) =
      n^2_k(C) = [C : \ell] \tilde{n}^2_k(\ell)
      \ ,
   \end{gather}
   i.e.~$\tilde{n}^1_k(\ell) = \tilde{n}^2_k(\ell)$. This proves that
   $\tilde{n}^1_k=\tilde{n}^2_k$.
\end{proof}

\bigskip

\begin{theorem} \label{theo:gauge-equivalence-0}
   Consider a 2-tree $\mathring{T}$ which disconnects $a$ and $b$, a tree $T$ and two gauge fields $n^1_0$ and $n^2_0$. The following statements are equivalent:
   \begin{enumerate}[topsep=0pt,itemsep=1ex]
      \item $n^1_0$ and $n^2_0$ are gauge equivalent;
      \item $n^1_0(C) = n^2_0(C)$ for every $C \in \mathbf{C}(\mathring{T})$, and
      $n^1_0(P) = n^2_0(P)$ for every $P \in \mathbf{P}_{ab}(\mathring{T})$;
      \item $n^1_0(C) = n^2_0(C)$ for every $C \in \mathbf{C}(T)$, and $n^1_0(P)
      = n^2_0(P)$ where $P$ is the only path from $a$ to $b$ in $T$.
   \end{enumerate}
\end{theorem}

\begin{proof}
   Note that the implications $1 \Rightarrow 2$ and $1 \Rightarrow 3$ follow
   trivially from gauge invariance of Wilson loops and Wilson lines from $a$ to
   $b$. We only need to prove the implications $2 \Rightarrow 1$ and $3
   \Rightarrow 1$.
   
   \textit{Implication $2 \Rightarrow 1$.} For $j=1,2$, let $\tilde{n}^j_0$ be
   the only representative in $[{n}^j_0]$ which is in axial gauge with
   respect to $\mathring{T}$. Existence and uniqueness of $\tilde{n}^j_0$ are
   guaranteed by theorem \ref{theo:axial}. In order to prove that $n^1_0$ and
   $n^2_0$ are gauge equivalent, it is enough to prove that
   $\tilde{n}^1_0=\tilde{n}^2_0$.
   
   By definition of axial gauge $\tilde{n}^1_k(\ell)=0=\tilde{n}^2_k(\ell)$ if
   $\ell \in \mathring{T}$.

   Consider $\ell \in S(\mathring{T})$, then a unique path $P$ from $a$ to $b$
   exists in $\mathring{T} \sqcup \{\ell\}$, and
   \begin{gather}
      n^j_0(P) = \tilde{n}^j_0(P) = [P : \ell] \tilde{n}^j_0(\ell) \ ,
   \end{gather}
   where we have used the gauge invariance of Wilson lines from $a$ to $b$. By
   definition $P \in \mathbf{P}_{ab}(\mathring{T})$, and by hypothesis
   \begin{gather}
      [P : \ell] \tilde{n}^1_0(\ell) = n^1_0(P) =
      n^2_0(P) = [P : \ell] \tilde{n}^2_0(\ell)
      \ ,
   \end{gather}
   i.e.~$\tilde{n}^1_0(\ell) = \tilde{n}^2_0(\ell)$.
   
   Consider $\ell \in \mathring{T}^* \setminus S(\mathring{T})$, then a unique
   loop $C$ exists in $\mathring{T} \sqcup \{\ell\}$, and
   \begin{gather}
      n^j_0(C) = \tilde{n}^j_0(C) = [C : \ell] \tilde{n}^j_0(\ell) \ ,
   \end{gather}
   where we have used the gauge invariance of Wilson loops. By definition $C \in
   \mathbf{C}(\mathring{T})$, and by hypothesis
   \begin{gather}
      [C : \ell] \tilde{n}^1_0(\ell) = n^1_0(C) =
      n^2_0(C) = [C : \ell] \tilde{n}^2_0(\ell)
      \ ,
   \end{gather}
   i.e.~$\tilde{n}^1_0(\ell) = \tilde{n}^2_0(\ell)$. This proves that
   $\tilde{n}^1_0=\tilde{n}^2_0$.

   \textit{Implication $3 \Rightarrow 1$.} Let $P$ the only path from $a$ to $b$
   in $T$, and let $\bar{\ell}$ some line in $P$. Then $\mathring{T} = T
   \setminus \{ \bar{\ell} \}$ is a 2-tree which disconnects $a$ and $b$.
   
   For $j=1,2$, let $\tilde{n}^j_0$ be the only representative in
   $[{n}^j_0]$ which is in axial gauge with respect to $\mathring{T}$.
   Existence and uniqueness of $\tilde{n}^j_0$ are guaranteed by theorem
   \ref{theo:axial}. In order to prove that $n^1_0$ and $n^2_0$ are gauge
   equivalent, it is enough to prove that $\tilde{n}^1_0=\tilde{n}^2_0$.
   
   By definition of axial gauge $\tilde{n}^1_0(\ell)=0=\tilde{n}^2_0(\ell)$ if
   $\ell \in \mathring{T}$.
   
   Gauge invariance of the Wilson line along $P$ then implies,
   \begin{gather}
      n^j_0(P) = \tilde{n}^j_0(P) = [P : \bar{\ell}] \tilde{n}^j_0(\bar{\ell}) \ .
   \end{gather}
   By hypothesis
   \begin{gather}
      [P : \bar{\ell}] \tilde{n}^1_0(\bar{\ell}) = n^1_0(P) =
      n^2_0(P) = [P : \bar{\ell}] \tilde{n}^2_0(\bar{\ell})
      \ ,
   \end{gather}
   i.e.~$\tilde{n}^1_0(\bar{\ell}) = \tilde{n}^2_0(\bar{\ell})$.
   
   Consider $\ell \in T^*$, then a unique loop $C$ exists in $T \sqcup \{\ell\}
   = \mathring{T} \sqcup \{\bar{\ell},\ell\}$, and
   \begin{gather}
      n^j_0(C) = \tilde{n}^j_0(C) = [C : \ell] \tilde{n}^j_0(\ell) + [C : \bar{\ell}] \tilde{n}^j_0(\bar{\ell}) \ ,
   \end{gather}
   where we have used the gauge invariance of Wilson loops. Note that $[C :
   \bar{\ell}]$ is zero if $\bar{\ell}$ is not in $C$. By definition $C \in
   \mathbf{C}(T)$, and by hypothesis
   \begin{gather}
      [C : \ell] \tilde{n}^1_0(\ell) + [C : \bar{\ell}] \tilde{n}^1_0(\bar{\ell}) = n^1_0(C) =
      n^2_0(C) = [C : \ell] \tilde{n}^2_0(\ell) + [C : \bar{\ell}] \tilde{n}^2_0(\bar{\ell})
      \ .
   \end{gather}
   We have already proved that $\tilde{n}^1_0(\bar{\ell}) =
   \tilde{n}^2_0(\bar{\ell})$, therefore it follows that $\tilde{n}^1_0(\ell) =
   \tilde{n}^2_0(\ell)$. This proves that $\tilde{n}^1_0=\tilde{n}^2_0$.
\end{proof}

\bigskip

\begin{corollary} \label{corollary:epshat:axial}
   Let $\mathring{T}_0$ be a 2-tree and let $T_k$ for $k=1,2,3$ be trees as in
   theorem~\ref{theo:epshat}. Since the functions $\hat{\epsilon}_\mu$ are
   gauge-invariant, with no loss of generality, we can assume $n$ to be in axial
   gauge with respect to the fourplet $\mathcal{T} =
   (\mathring{T}_0,T_1,T_2,T_3)$ (thanks to theorem~\ref{theo:axial}). The lines
   in the cut-set $S(\mathring{T}_0)$ are assumed to be oriented from $a$ to
   $b$, i.e.~if $\ell \in S(\mathring{T}_0)$ then $i(\ell)$ is connected to $a$
   in $\mathring{T}_0$ and $f(\ell)$ is connected to $b$ in $\mathring{T}_0$.
   
   It follows that
   \begin{gather}
      \hat{\epsilon}_k(n_k) = \sum_{\ell \in T_k^*} |n_k(\ell)|
      \label{eq:corollary:epshat:axial:k}
      \ , \\
      \hat{\epsilon}_0(n_0) = \sum_{\ell \in \mathring{T}_0^*} |n_0(\ell)|
      +   \min \bigg \{0 , \, \frac{ p_0^+(n|\mathring{T}_0) + p_0^0(n|\mathring{T}_0) - p_0^-(n|\mathring{T}_0)  }{2} \bigg \}
      \ ,
      \label{eq:corollary:epshat:axial:0}
   \end{gather}
   where $p_0^+(n|\mathring{T}_0)$, $p_0^0(n|\mathring{T}_0)$,
   $p_0^-(n|\mathring{T}_0)$ are equal to the number of lines $\ell$ in the
   cut-set $S(\mathring{T}_0)$ such that $n_0(\ell) \ge 1$, $n_0(\ell) = 0$,
   $n_0(\ell) \le -1$ respectively.

   The following loose but useful bounds hold
   \begin{gather}
      \hat{\epsilon}_k(n_k) \ge c_k(n|T_k) \ ,
      \label{eq:corollary:epshat:bound:k}
   \end{gather}
   where $c_k(n|T_k)$ is equal to the number of lines $\ell$ in $T_k^*$ with
   $n_k(\ell) \neq 0$, and
   \begin{gather}
      \hat{\epsilon}_0(n_0) \ge c_0(n|\mathring{T}_0)
      + p_0^+(n|\mathring{T}_0) + \frac{1}{2} p_0^-(n|\mathring{T}_0)
      +  \min \bigg\{ \frac{p_0^-(n|\mathring{T}_0)}{2} ,\frac{ p_0^+(n|\mathring{T}_0) + p_0^0(n|\mathring{T}_0) }{2} \bigg\}
      \ ,
      \label{eq:corollary:epshat:bound:0}
   \end{gather}
   where $c_0(n|\mathring{T}_0)$ is the number of lines $\ell$ in
   $\mathring{T}_0^* \setminus S(\mathring{T}_0)$ with $n_0(\ell) \neq 0$.
\end{corollary}

\begin{proof}
   Note that the set $T_k^*$ is in a bijective correspondence with
   $\mathbf{C}(T_k)$. In fact, by definition, if $C \in \mathbf{C}(T_k)$ then a
   line $\ell \in T_k^*$ exists such that $C$ is the unique loop in $T_k \sqcup
   \{\ell\}$. Invertibility follows from the observation that the line is
   uniquely determined by the loop via $\{ \ell \} = C \setminus T_k$. Observe
   that
   \begin{gather}
      |n_k(C)|
      =
      \left| \sum_{\ell' \in C} [C:\ell'] \, n_k(\ell') \right|
      =
      \left| [C:\ell] \, n_k(\ell) + \sum_{\ell' \in C \cap T_k} [C:\ell'] \, n_k(\ell') \right|
      =
      |n_k(\ell)|
      \ ,
   \end{gather}
   where we have used the fact that $n_k$ is in axial gauge with respect to
   $T_k$. Therefore
   \begin{gather}
      \sum_{C \in \mathbf{C}(T_k)} |n_k(C)| = \sum_{\ell \in T_k^*} |n_k(\ell)| \ .
   \end{gather}
   
   The set $\mathring{T}_0^* \setminus S(\mathring{T}_0)$ is in a bijective
   correspondence with $\mathbf{C}(\mathring{T}_0)$. The construction is
   identical to the previous case, therefore
   \begin{gather}
      \sum_{C \in \mathbf{C}(\mathring{T}_0)} |n_0(C)| = \sum_{\ell \in \mathring{T}_0^* \setminus S(\mathring{T}_0)} |n_0(\ell)| \ .
   \end{gather}
   The set $S(\mathring{T}_0)$ is in a bijective correspondence with
   $\mathbf{P}_{ab}(\mathring{T}_0)$. In fact, by definition, if $P \in
   \mathbf{P}_{ab}(\mathring{T}_0)$ then a line $\ell \in S(\mathring{T}_0)$
   exists such that $P$ is the unique path from $a$ to $b$ in $\mathring{T}_0
   \sqcup \{\ell\}$. Invertibility follows from the observation that the line is
   uniquely determined by the loop via $\{ \ell \} = P \setminus
   \mathring{T}_0$. Thanks to the particular choice of orientation of the lines
   in $S(\mathring{T}_0)$, $[P : \ell]=1$. Observe that
   \begin{gather}
      n_0(P)
      =
      \sum_{\ell' \in P} [P:\ell'] \, n_0(\ell')
      =
      [P:\ell] \, n_0(\ell) + \sum_{\ell' \in P \cap \mathring{T}_0} [P:\ell'] \, n_0(\ell')
      =
      n_0(\ell)
      \ ,
   \end{gather}
   where we have used the fact that $n_0$ is in axial gauge with respect to
   $\mathring{T}_0$. Therefore
   \begin{gather}
      \sum_{P \in \mathbf{P}(\mathring{T}_0)} |n_0(P)| = \sum_{\ell \in S(\mathring{T}_0)} |n_0(\ell)| \ .
   \end{gather}
   The corollary is a simple application of these relations to
   theorem~\ref{theo:epshat}.
\end{proof}
 
\subsection{Possible values of \texorpdfstring{$\hat{\epsilon}_\mu(n_\mu)$}{epsilon\_mu(n\_mu)}}
\label{app:susec:onehalf}

The deceptively simple theorem discussed in this subsection concludes the
characterization of the function $\hat{\epsilon}_\mu(n_\mu)$, and constitutes
the  backbone of the analysis of the asymptotic behaviour of Feynman
integrals in the large-$L$ limit presented in sections~\ref{subsec:diag:asymptotic} and \ref{subsec:diag:separation}
and following.

\bigskip

\begin{theorem} \label{theo:onehalf}
   The following statements hold:
   \begin{enumerate}[topsep=0pt,itemsep=1ex]
      \item For any gauge field $n$, $\hat{\epsilon}_k(n_k) \in \mathbb{N}$ for
      $k=1,2,3$, and $\hat{\epsilon}_0(n_0) \in \frac{\mathbb{N}}{2} \setminus
      \{ \frac{1}{2} \}$.
      \item For $\mu \in \{0,1,2,3\}$, $\hat{\epsilon}_\mu(n_\mu) = 0$ if and
      only if $n_\mu$ is a pure gauge field.
   \end{enumerate}
\end{theorem}

\begin{proof}
   
   We prove each point separately.

   \begin{enumerate}[topsep=0pt,itemsep=1ex]
      
      \item With no loss of generality, we assume that $n$ is in axial gauge as
      in corollary~\ref{corollary:epshat:axial}. The observation that
      \begin{gather}
         \hat{\epsilon}_k(n_k) \in \mathbb{N}
         \quad \text{and} \quad
         \hat{\epsilon}_0(n_0) \in \tfrac{\mathbb{N}}{2} \,,
      \end{gather}
      follows trivially from eqs.~\eqref{eq:corollary:epshat:axial:k}
      and~\eqref{eq:corollary:epshat:axial:0}. We only need to prove that the
      value $\hat{\epsilon}_0(n_0) = 1/2$ is excluded. Let us work by
      contradiction and assume that $\hat{\epsilon}_0(n_0)=1/2$. Since
      $\hat{\epsilon}_0(n_0)=1/2$ is not an integer, thanks to
      eq.~\eqref{eq:corollary:epshat:axial:0}, we must have
      \begin{gather}
         p_0^+(n|\mathring{T}_0) + p_0^0(n|\mathring{T}_0) - p_0^-(n|\mathring{T}_0) < 0
         \ .
         \label{eq:theo:onehalf:ineq}
      \end{gather}
      The bound~\eqref{eq:corollary:epshat:bound:0} implies that
      \begin{gather}
         \frac{1}{2}
         =
         \hat{\epsilon}_0(n_0)
         \ge
         c_0(n|\mathring{T}_0)
         + \frac{3}{2} p_0^+(n|\mathring{T}_0)
         + \frac{1}{2} p_0^-(n|\mathring{T}_0)
         + \frac{1}{2} p_0^0(n|\mathring{T}_0)
         \ ,
      \end{gather}
      which implies, together with the inequality~\eqref{eq:theo:onehalf:ineq},
      \begin{gather}
         p_0^+(n|\mathring{T}_0) = p_0^0(n|\mathring{T}_0) = 0
         \ , \qquad
         p_0^-(n|\mathring{T}_0) = 1
         \ .
      \end{gather}
      In particular this means that $S(\mathring{T}_0)$ contains only one
      element. Let $\ell$ be the only line in $S(\mathring{T}_0)$. Every path
      from $a$ to $b$ must contain $\ell$, which means that $\ell$ disconnects
      $a$ and $b$. This is in contradiction with the fact that $\mathcal{G}$ is
      1PI between $a$ and $b$ (proposition~\ref{prop:1PI}).
      
      \item We need to prove two implications. First note that
      eqs.~\eqref{eq:theo:epshat:k} and \eqref{eq:theo:epshat:0} are gauge
      invariant. If $n_\mu$ is pure, up to a gauge transformation $n_\mu=0$. By
      using eqs.~\eqref{eq:theo:epshat:k} and \eqref{eq:theo:epshat:0} in this
      gauge, one trivially finds $\hat{\epsilon}_\mu(n_\mu)=0$.
      
      Let us prove the other implication. We assume that $n$ is in axial gauge
      as in corollary~\ref{corollary:epshat:axial}, and we want to prove that
      $\hat{\epsilon}_\mu(n_\mu) = 0$ implies $n_\mu=0$.
      
      If $\hat{\epsilon}_k(n_k) = 0$,
      bound~\eqref{eq:corollary:epshat:bound:k} becomes simply
      \begin{gather}
         c_k(n|T_k) \le 0 \ .
      \end{gather}
      Since $c_k(n|T_k)$ is the number of lines $\ell$ in $T_k^*$ with
      $n_k(\ell) \neq 0$, this means that $c_k(n|T_k)=0$, and $n_k(\ell)=0$ for
      all lines in $T_k^*$. On the other hand, since $n_k$ is in axial gauge
      with respect to $T_k$, $n_k(\ell)=0$ for all lines in $T_k$. Therefore
      $n_k=0$.

      If $\hat{\epsilon}_0(n_0) = 0$,
      using bound~\eqref{eq:corollary:epshat:bound:0} one sees that
      \begin{gather}
         c_0(n|\mathring{T}_0)
         + p_0^+(n|\mathring{T}_0) + \frac{1}{2} p_0^-(n|\mathring{T}_0)
         +   \min \bigg\{ \frac{p_0^-(n|\mathring{T}_0)}{2} , \frac{ p_0^+(n|\mathring{T}_0) + p_0^0(n|\mathring{T}_0) }{2} \bigg\} \le 0
         \ .
      \end{gather}
      Since the left-hand side is a sum of non-negative terms, this implies
      \begin{gather}
         c_0(n|\mathring{T}_0) = p_0^+(n|\mathring{T}_0) = p_0^-(n|\mathring{T}_0) = 0
         \ .
      \end{gather}
      Again, by looking at the definitions of these numbers, one easily sees
      that this implies that $n_0(\ell)=0$ for all lines in $\mathring{T}_0^*$.
      On the other hand, since $n_0$ is in axial gauge with respect to
      $\mathring{T}_0$, $n_0(\ell)=0$ for all lines in $\mathring{T}_0$.
      Therefore $n_0=0$.

   \end{enumerate}

\end{proof}

\subsection{Characterization of gauge fields with \texorpdfstring{$\hat{\epsilon}_\text{s}(\vec{n}) < \sqrt{2+\sqrt{3}}$}{epsilon\_s(n)<1.93}}
\label{app:subsec:sqrt}

In this appendix we use the following definitions, which are compatible with the
ones provided in sections~\ref{subsec:diag:leading-L} and
\ref{subsec:diag:leading-T}.

\textit{Localized gauge fields. Localizable and simple gauge fields and orbits.}
Given a subset of lines $A \subset \mathcal L$, the gauge field $n_k$ (resp.
$\vec{n}$) is said to be localized on $A$ if and only if it is zero for each line
in $\mathcal{L} \setminus A $, and non-zero on each line in $A$. The gauge field
$n_k$ (resp. $\vec{n}$) and the gauge orbit $[n_k]$ (resp. $[\vec{n}]$) are said
to be localizable on $A$ if and only $n_k$ (resp. $\vec{n}$) is gauge equivalent
to a gauge field that is localized on $A$. The gauge field $n_k$ (resp.
$\vec{n}$) and the gauge orbit $[n_k]$ (resp. $[\vec{n}]$) are said to be simple
if and only if they are localizable on a single line.

\bigskip

\begin{theorem} \label{theo:one-k}
   If $\hat{\epsilon}_k(n_k) = 1$ then, up to a gauge transformation, $n_k$ is
   localized on a line $\ell$ and $|n_k(\ell)| = 1$.
\end{theorem}

\begin{proof}
   
   If $\hat{\epsilon}_k(n_k)=1$,
   eq.~\eqref{eq:corollary:epshat:axial:k} implies that a line $\bar{\ell}$
   exists such that $|n_k(\bar{\ell})|=1$, and $n_k(\ell)=0$ for $\ell \in
   T_k^* \setminus \{\bar{\ell}\}$. The thesis follows from the observation
   that $n_k$ is in axial gauge with respect to $T_k$, therefore
   $n_k(\ell)=0$ for every $\ell \neq \bar{\ell}$.
   
\end{proof}

\bigskip

\begin{theorem} \label{theo:s-equivalence}
   The following notions of $s$-equivalence between two lines $\ell_1$ and
   $\ell_2$ are logically equivalent:
   \begin{enumerate}[topsep=0pt,itemsep=1ex]
      \item $\ell_1$ and $\ell_2$ are said to be $s$-equivalent if and only if either
      $\ell_1 = \ell_2$ or $\{\ell_1,\ell_2\}$ is a cut-set.
      \item $\ell_1$ and $\ell_2$ are said to be $s$-equivalent if and only if every loop
      containing $\ell_1$ contains also $\ell_2$.
   \end{enumerate}
   The notion of $s$-equivalence is indeed an equivalence relation, and
   $s$-equivalence classes are given by
   \begin{gather}
      [ \ell ]_s = \bigcap_{\substack{ C \textrm{ is a loop} \\ \textrm{with } \ell \in C }} C \ .
   \end{gather}
\end{theorem}

\begin{proof}
   Reflexivity and symmetry of $s$-equivalence are manifest in notion 1, while
   transitivity is manifest in notion 2. The characterization of the
   $s$-equivalence classes follows immediately from notion 2. We are left with
   the task to prove the logical equivalence of the two notions.
   
   \textit{Implication} 1 $\Rightarrow$ 2. The claim is trivial for
   $\ell_1=\ell_2$. Assume $\ell_1 \neq \ell_2$. We need to prove that, if $S =
   \{\ell_1,\ell_2\}$ is a cut-set made of two lines, every loop which contains
   $\ell_1$ contains also $\ell_2$. In fact, assume by contradiction that $C$ is
   a loop containing $\ell_1$ but not $\ell_2$, then $C \setminus \{ \ell_1 \}$
   is a path in $\mathcal{G}-S$ connecting the two endpoints of $\ell_1$.
   However, from the definition of cut-set, it easily follows that the two
   endpoints of every line in $S$ belong to different connected components of
   $\mathcal{G}-S$, which proves that the loop $C$ does not exist.

   \textit{Implication} 2 $\Rightarrow$ 1. We need to prove that, given
   $\ell_1 \neq \ell_2$, if every loop containing $\ell_1$ contains also
   $\ell_2$, then $S = \{\ell_1,\ell_2\}$ is a cut-set. In fact, assume by
   contradiction that $S$ is not a cut-set. Since $\mathcal{G}$ is 1PI it
   follows that $\mathcal{G}-S$ is connected. Therefore a path $P$ exists in
   $\mathcal{G}-S$ which connects the two endpoints of $\ell_1$. Then $P \sqcup
   \{\ell_1\}$ is a loop containing $\ell_1$ but not $\ell_2$, in contradiction
   with the hypothesis.
\end{proof}

\bigskip

\begin{theorem} \label{theo:s-localization}
   The orientation of the lines of $\mathcal{G}$ can be chosen in such a way
   that, if $\ell$ and $\ell'$ are $s$-equivalent and $C$ is a loop that
   contains both, then $[C:\ell]=[C:\ell']$, i.e.~either both $\ell$ and $\ell'$
   have the same orientation as $C$, or they both have the opposite orientation
   of $C$. This choice of orientation is assumed here.
   
   Let $n_k$ and $n'_k$ be gauge fields localized on $\ell$ and $\ell'$
   respectively. $n_k$ and $n'_k$ are gauge equivalent if and only if $\ell$ and
   $\ell'$ are $s$-equivalent and $n_k(\ell)=n'_k(\ell')$.
   
   Let $\vec{n}$ and $\vec{n}'$ be gauge fields localized on $\ell$ and $\ell'$
   respectively. $\vec{n}$ and $\vec{n}'$ are gauge equivalent if and only if
   $\ell$ and $\ell'$ are $s$-equivalent and $\vec{n}(\ell)=\vec{n}'(\ell')$.
\end{theorem}

\begin{proof}
   Choose a line $\ell$, then one can define the map $\omega_\ell : [\ell]_s \to
   \{-1,1\}$ in the following way. Let $C$ be a loop that contains $\ell$ (this
   loop exists since $\mathcal{G}$ is 1PI) and a particular orientation for $C$,
   then we define for $\ell' \in [\ell]_s$
   \begin{gather}
      \omega_\ell(\ell') = \frac{[C : \ell]}{[C : \ell']} \ .
   \end{gather}
   We want to see that this definition depends neither on the orientation of
   $C$, nor on the loop $C$ (as long as $\ell \in C$). This is obvious for
   $\ell'=\ell$, since $\omega_\ell(\ell)=1$. Therefore we can assume that
   $\ell' \in [\ell]_s$ and $\ell' \neq \ell$.
   
   First we note that flipping the orientation of $C$ is equivalent to
   replacing $[C : \ell''] \to -[C : \ell'']$ for every line $\ell''$.
   $\omega_\ell(\ell')$ is invariant under this replacement, i.e.
   $\omega_\ell(\ell')$ does not depend on the orientation of $C$. Now assume
   that $C'$ is another loop that contains $\ell$. Thanks to
   theorem~\ref{theo:s-equivalence}, $C'$ contains also $\ell'$. Also, thanks to
   the same theorem, $S=\{\ell,\ell'\}$ is a cut-set. Choose an orientation for
   $S$. Then theorem 2-14 in \cite{nakanishi1971graph} implies that
   \begin{gather}
      [ C : \ell ] [ S : \ell ] + [ C : \ell' ] [ S : \ell' ] = 0 \ , \\
      [ C' : \ell ] [ S : \ell ] + [ C' : \ell' ] [ S : \ell' ] = 0 \ ,
   \end{gather}
   which imply
   \begin{gather}
      \frac{[ C : \ell ]}{[ C : \ell' ]} = - \frac{[ S : \ell' ]}{[ S : \ell ]} = \frac{[ C' : \ell ]}{[ C' : \ell' ]} \ .
   \end{gather}
   This is exactly the statement that $\omega_\ell(\ell')$ does not depend on
   the particular choice of $C$.
   
   One can define a new orientation for the lines $\ell' \in [\ell]_s$ in the
   following way: one keeps the original orientation of $\ell'$ if
   $\omega_\ell(\ell')=+1$, and one flips the orientation of $\ell'$ if
   $\omega_\ell(\ell')=-1$. With this new orientation is it straightforward to
   prove that, for every loop $C$ that contains both $\ell$ and $\ell'$,
   $\frac{[C : \ell]}{[C : \ell']}=1$. This concludes the proof of the first
   part of the theorem.
   
   In the following we assume that lines are oriented as explained above. We
   observe that the third part of the theorem is a simple application of the
   second part. We assume that $n_k$ and $n'_k$ are gauge fields localized on
   $\ell$ and $\ell'$ respectively, and we prove the two implications of the
   second part separately.
   
   $\bullet \ $ \textit{$n_k$ and $n'_k$ are gauge equivalent $\Rightarrow$
   $\ell$ and $\ell'$ are $s$-equivalent and $n_k(\ell)=n'_k(\ell')$.}
      
   To prove that $\ell$ and $\ell'$ are $s$-equivalent, we need to prove that
   any loop containing $\ell$ contains also $\ell'$. Let $C$ be a loop which
   contains $\ell$. Note that
   \begin{gather}
      [C:\ell'] n'_k(\ell') = n'_k(C) = n_k(C) = [C:\ell] n_k(\ell) \neq 0 \ ,
      \label{theo:s-localization:eq0}
   \end{gather}
   where we have used the fact that $n'_k$ is localized on $\ell'$ in the first
   equality, gauge invariance of the Wilson loops in the second equality, the
   fact that $n_k$ is localized on $\ell$ in the third equality, and finally the
   assumption that $\ell \in C$ in the final inequality. It follows that
   $[C:\ell'] \neq 0$, i.e.~$\ell' \in C$, and thus that $\ell$ and $\ell'$
   are $s$-equivalent. Because of the particular choice of orientation for the
   lines of the equivalence class, $[C : \ell]=[C : \ell']$. Together with the
   eq.~\ref{theo:s-localization:eq0}, this implies immediately
   $n_k(\ell)=n'_k(\ell')$.
   
   $\bullet \ $ \textit{$\ell$ and $\ell'$ are $s$-equivalent and
   $n_k(\ell)=n'_k(\ell')$ $\Rightarrow$ $n_k$ and $n'_k$ are gauge equivalent.}

   Consider a loop $C$. Because of $s$-equivalence, two
   possibilities are given: either $C$ contains both $\ell$ and $\ell'$, or it
   contains neither of the two.
   If $C$ contains both $\ell$ and $\ell'$,
   \begin{gather}
      n'_k(C) = [C:\ell'] n'_k(\ell') = [C:\ell] n_k(\ell) = n_k(C) \ ,
   \end{gather}
   where we have used the fact that $n'_k$ is localized on $\ell'$ in the first
   equality, the fact that $[{C}:\ell']=[{C}:\ell]$ (choice of
   orientation) and $n'_k(\ell')=n_k(\ell)$ (hypothesis) in the second equality,
   and finally the fact that $n_k$ is localized on $\ell$ in the third equality.
   On the other hand, if $C$ contains neither $\ell$ nor $\ell'$, by hypothesis
   of localization,
   \begin{gather}
      n'_k(C) = 0 = n_k(C) \ .
   \end{gather}
   Therefore, for all loops $C$, $n'_k(C)=n_k(C)$. Gauge equivalence of $n_k$
   and $n'_k$ follows from theorem~\ref{theo:gauge-equivalence-k}.

\end{proof}

\begin{theorem} \label{theo:pure-s}
   $\hat{\epsilon}_\text{s}(\vec{n}) = 0$ if and only if $\vec{n}$ is a pure
   gauge field. If $\vec{n}$ is not a pure gauge field, then
   $\hat{\epsilon}_\text{s}(\vec{n}) \ge 1$.
\end{theorem}

\begin{proof}
   Note that, for any $k \in \{1,2,3\}$ one has $\left\| \delta \vec{x}(\ell)
   + \vec{n}(\ell) \right\|_2 \ge | \delta x_k(\ell) + n_k(\ell) |$, which
   implies the loose bound
   \begin{gather}
      \hat{\epsilon}_\text{s}(\vec{n}) \ge \hat{\epsilon}_k(n_k)
      \ , \label{eq:hatepsilon-s-k-ineq}
   \end{gather}
   where $\hat{\epsilon}_k(n_k)$ is defined in eq.~\eqref{eq:hatepsilon-k}. The
    statement is proven by the following observations:
   \begin{itemize}[topsep=0pt,itemsep=0ex]
      \item If $\hat{\epsilon}_\text{s}(\vec{n}) = 0$, by
      inquality~\eqref{eq:hatepsilon-s-k-ineq}, also $\hat{\epsilon}_k(n_k) =
      0$. Theorem~\ref{theo:onehalf} implies that $n_k$ is a pure gauge field
      for all $k \in \{1,2,3\}$, which is equivalent to say that $\vec{n}$ is a
      pure gauge field.
      \item If $\vec{n}$ is a pure gauge field, since
      $\hat{\epsilon}_\text{s}(\vec{n})$ is gauge invariant,
      $\hat{\epsilon}_\text{s}(\vec{n}) = \hat{\epsilon}_\text{s}(\vec{0}) = 0$.
      The second equality follows trivially from the fact that the minimum in
      eq.~\eqref{eq:hatepsilon-s} for $\vec{n}=\vec{0}$ is realized for
      $\vec{x}(v)=\vec{0}$.
      \item If $\vec{n}$ is not a pure gauge field, then a value of $k \in
      \{1,2,3\}$ exists such that $n_k$ is not pure. Thanks to
      inequality~\eqref{eq:hatepsilon-s-k-ineq} and theorem~\ref{theo:onehalf},
      $\hat{\epsilon}_\text{s}(\vec{n}) \ge \hat{\epsilon}_k(n_k) \ge 1$.
   \end{itemize}
\end{proof}

\bigskip

\begin{theorem} \label{theo:sqrt}
   If $1 \le \hat{\epsilon}_\text{s}(\vec{n}) < \sqrt{2+\sqrt{3}}$ then
   $\vec{n}$ is simple.
\end{theorem}

\begin{proof}
   Thanks to inequality~\eqref{eq:hatepsilon-s-k-ineq}, by hypothesis we have
   \begin{gather}
      \hat{\epsilon}_k(n_k) \le \hat{\epsilon}_\text{s}(\vec{n}) < \sqrt{2+\sqrt{3}} < 2  \,.
   \end{gather}
   As a consequence of theorem~\ref{app:susec:onehalf}, two possibilities are
   given: either $\hat{\epsilon}_k(n_k)=0$, or $\hat{\epsilon}_k(n_k)=1$.
   Thanks to theorems~\ref{app:susec:onehalf} and~\ref{theo:one-k}, $\vec{n}$ is
   gauge equivalent to a field $\vec{n}'$ with the following property:
   \begin{itemize}[topsep=0pt,itemsep=0ex]
      \item if $\hat{\epsilon}_k(n_k)=0$, $n'_k=0$ (i.e.~$n_k$ is pure);
      \item if $\hat{\epsilon}_k(n_k)=1$, $n'_k$ is localized on a line $\ell_k$
      and $|n'_k(\ell_k)|=1$ (in particular, $n_k$ is simple).
   \end{itemize}
   By theorem~\eqref{theo:pure-s}, since $\hat{\epsilon}_\text{s}(\vec{n}) \ge
   1$ then $\vec{n}$ is not pure, which means that at least one of its
   components is not pure, hence simple. Up to an irrelevant relabelling of the
   coordinates we can assume that $n_1$ is simple. Note that, if $n_2$ and
   $n_3$ are both pure then $\vec{n}$ is simple.
   
   Let us assume by contradiction that $\vec{n}$ is not simple. Then either $n_2$ and $n_3$, or both, has the
   following properties (using $n_2$ for concreteness): $n_2$ is simple, and $n_2$ is not gauge equivalent to
   any field localized on $\ell_1$. Then theorem~\ref{theo:s-localization}
   implies that $\ell_1$ and $\ell_2$ are not $s$-equivalent, which in turn
   implies $\ell_1 \neq \ell_2$ and $\{ \ell_1 , \ell_2 \}$ is not a cut-set
   (theorem \ref{theo:s-equivalence}). In particular $\mathcal{G} - \{ \ell_1 ,
   \ell_2 \}$ is connected. Let $T$ be a tree of $\mathcal{G} - \{ \ell_1 ,
   \ell_2 \}$, then $T$ is also a tree of $\mathcal{G}$
   (\cite{nakanishi1971graph}, theorem 2-6). Let $C_j$ be the only loop in $T
   \sqcup \{ \ell_j \}$ with $j=1,2$. Two possibilities are given: either $C_1
   \cap C_2$ is empty, or it is a path in $T$. We will use this fact in a
   moment.
   
   The third component of the field $\vec{n}$ plays no role in the what follows.
   Given a generic 3-component vector $\vec{a}$, it is convenient to introduce
   the projected 2-component vector
   \begin{gather}
      \ul{a} = (a_1,a_2) \ .
   \end{gather}
   Note that
   \begin{gather}
      \sum_{\ell \in \mathcal{L}} \left\| \vec{x}[f(\ell)] - \vec{x}[i(\ell)] + \vec{n}(\ell) \right\|_2
      \ge
      \sum_{\ell \in \mathcal{L}} \left\| \ul{x}[f(\ell)] - \ul{x}[i(\ell)] + \ul{n}(\ell) \right\|_2
      \ ,
   \end{gather}
   which implies
   \begin{gather}
      \hat{\epsilon}_\text{s}(\vec{n})
      \ge
      \min_{\ul{x} \text{ with } \ul{x}(a)=\ul{0}}
      \sum_{\ell \in \mathcal{L}} \left\| \ul{x}[f(\ell)] - \ul{x}[i(\ell)] + \ul{n}(\ell) \right\|_2
      \ .
   \end{gather}
   Gauge invariance of the Wilson loops yields:
   \begin{gather}
      \ul{n}(C_1) = \ul{n}'(C_1) = \left( n'_1(\ell_1) , 0 \right)
      \label{eq:fermat:nC1}
      \ , \\
      \ul{n}(C_2) = \ul{n}'(C_2) = \left( 0, n'_2(\ell_2) \right)
      \label{eq:fermat:nC2}
      \ ,
   \end{gather}
   where the orientation of $C_j$ has been chosen in such a way that $[C_j :
   \ell_j]=1$. We recall that $|n'_1(\ell_1)| = |n'_2(\ell_2)| = 1$, which implies
   \begin{gather}
      \| \ul{n}(C_1) \|_2 = \| \ul{n}(C_2) \|_2 = 1 \ .
      \label{eq:fermat:nC12}
   \end{gather}
   
   Let us analyze now the possibility that $C_1 \cap C_2 = \varnothing$. One
   obtains
   \begin{gather}
      \sum_{\ell \in \mathcal{L}} \left\| \ul{x}[f(\ell)] - \ul{x}[i(\ell)] + \ul{n}(\ell) \right\|_2
      \ge \nonumber \\ \qquad \ge
      \sum_{j=1}^2 \sum_{\ell \in C_j} \big \| [C_j : \ell] \{ \ul{x}[f(\ell)] - \ul{x}[i(\ell)] + \ul{n}(\ell) \} \big \|_2
      \ge \nonumber \\ \qquad \ge
      \sum_{j=1}^2 
      \left\| \sum_{\ell \in C_j} [C_j : \ell] \{ \ul{x}[f(\ell)] - \ul{x}[i(\ell)] + \ul{n}(\ell) \} \right\|_2
      = \nonumber \\ \qquad =
      \sum_{j=1}^2
      \left\| \sum_{\ell \in C_j} [C_j : \ell] \ul{n}(\ell) \right\|_2
      = \nonumber \\ \qquad =
      \|\ul{n}(C_1)\|_2 + \|\ul{n}(C_2)\|_2
      =
      2 \ .
      \label{eq:fermat:in1}
   \end{gather}
   In the first inequality we have used $C_1 \cap C_2 = \varnothing$, and $[C_j
   : \ell] = \pm 1$ if $\ell \in C_j$. The second inequality is nothing but the
   triangular inequality. Then we have used the fact that all $\ul{x}$'s cancel
   in pairs in the sum over $\ell \in C_j$, and finally
   eq.~\eqref{eq:fermat:nC12}. By minimizing inequality~\eqref{eq:fermat:in1}
   with respect to $\ul{x}$ one obtains $\hat{\epsilon}_\text{s}(\vec{n}) \ge
   2$, which is contradiction with the hypothesis.
   
   If $C_1 \cap C_2 \neq \varnothing$, then one proves that the following sets
   \begin{gather}
      P_0 = C_1 \cap C_2 \ , \qquad
      P_1 = C_1 \setminus P_0 \ , \qquad
      P_2 = C_2 \setminus P_0 \ ,
   \end{gather}
   are disjoint paths with $\ell_1 \in P_1$ and $\ell_2 \in P_2$. It is clear
   that the three paths connect the same pair of vertices $v$ and $w$. We will
   choose the orientation of the three paths in such a way that they all go from
   $v$ to $w$. In a similar way to inequality~\eqref{eq:fermat:in1}, we
   calculate
   \begin{gather}
      \sum_{\ell \in \mathcal{L}} \left\| \ul{x}[f(\ell)] - \ul{x}[i(\ell)] + \ul{n}(\ell) \right\|_2
      \ge \nonumber \\ \qquad \ge
      \sum_{j=0}^2 \sum_{\ell \in P_j} \left\| [P_j : \ell] \{ \ul{x}[f(\ell)] - \ul{x}[i(\ell)] + \ul{n}(\ell) \} \right\|_2
      \ge \nonumber \\ \qquad \ge
      \sum_{j=0}^2 
      \left\| \sum_{\ell \in P_j} [P_j : \ell] \{ \ul{x}[f(\ell)] - \ul{x}[i(\ell)] + \ul{n}(\ell) \} \right\|_2
      = \nonumber \\ \qquad =
      \sum_{j=0}^2
      \left\| \ul{x}(w) - \ul{x}(v) + \sum_{\ell \in P_j} [P_j : \ell] \ul{n}(\ell) \right\|_2
      = \nonumber \\ \qquad =
      \sum_{j=0}^2
      \left\| \ul{x}(w) - \ul{x}(v) + \ul{n}(P_j) \right\|_2
      \ .
      \label{eq:fermat:in2}
   \end{gather}
   In this case, all $\ul{x}$'s cancel in pairs in the sum over $\ell \in P_j$,
   except the ones corresponding to the endpoints of the three paths. By
   minimizing inequality~\eqref{eq:fermat:in1} with respect to $\ul{x}$, and
   defining $\ul{z} = \ul{x}(v) - \ul{x}(w)$, one obtains
   \begin{gather}
      \hat{\epsilon}_\text{s}(\vec{n}) \ge
      \min_{\ul{z}} \sum_{j=0}^2
      \left\| \ul{z} - \ul{n}(P_j) \right\|_2
      \ .
      \label{eq:fermat:min}
   \end{gather}
        This is the classical Fermat-Toricelli problem (see e.g. \cite{Boltyanski}).
   If the three vectors $\ul{n}(P_j)$ are interpreted as the vertices of a
   triangle $ABC$, the minimum $\vec{z}$ is the Fermat point of the given
   triangle. The solution of this minimization problem is known in terms of
   geometrical properties of the triangle.

   In particular, we identify
   \begin{gather}
      \overrightarrow{OA} = \ul{n}(P_0)
      \ , \qquad
      \overrightarrow{OB} = \ul{n}(P_1)
      \ , \qquad
      \overrightarrow{OC} = \ul{n}(P_2)
      \ ,
   \end{gather}
   which imply
   \begin{gather}
      \overrightarrow{AB} = \ul{n}(P_1) - \ul{n}(P_0) = \ul{n}(C_1)
      \ , \\
      \overrightarrow{AC} = \ul{n}(P_2) - \ul{n}(P_0) = \ul{n}(C_2)
      \ .
   \end{gather}
   Using eq.~\eqref{eq:fermat:nC12} it turns out that the triangle $ABC$ is right-angled in $A$, and it is isosceles with $|AB|=|AC|=1$. The location of the Fermat point $P$ is illustrated in figure~\ref{fig:fermat}. Elementary trigonometry yields
   \begin{gather}
      |PB| = |PC| = |BC| \frac{\sin 30^\circ}{\sin 120^\circ} = \frac{\sqrt{6}}{3}
      \ , \\
      |PA| = |AB| \frac{\sin 15^\circ}{\sin 120^\circ} = 
      \frac{3\sqrt{2}-\sqrt{6}}{6}
      \ .
   \end{gather}
   The minimization problem~\eqref{eq:fermat:min} is solved by
   \begin{gather}
      \hat{\epsilon}_\text{s}(\vec{n}) \ge
      \min_{\ul{z}} \sum_{j=0}^2
      \left\| \ul{z} - \ul{n}(P_j) \right\|_2
      =
      |PA|+|PB|+|PC|
      =
      \frac{\sqrt{6}+\sqrt{2}}{2}
      =
      \sqrt{2 + \sqrt{3}}
      \ .
   \end{gather}

   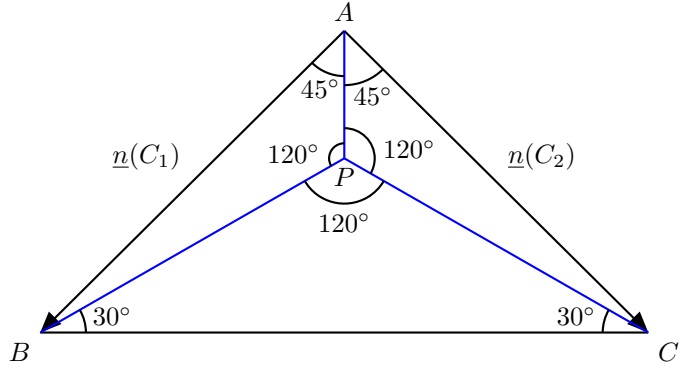
\begin{figure}[tb]
      \centering
       
      \begin{tikzpicture}[scale=4,thick]
         \coordinate [label=above:$A$]  (A) at (0,1);
         \coordinate [label=below left:$B$] (B) at (-1,0);
         \coordinate [label=below right:$C$] (C) at (1,0);
         \coordinate [label=below:$P$] (P) at (0,0.57735);

         \draw (B) -- (C);
         \draw[-triangle 45] (A) -> node[above left,inner sep=4pt] {$\ul{n}(C_1)$} (B);
         \draw[-triangle 45] (A) -> node[above right,inner sep=4pt] {$\ul{n}(C_2)$} (C);
         
         \draw[blue] (A) -- (P);
         \draw[blue] (B) -- (P);
         \draw[blue] (C) -- (P);
         
         \draw ($(A)+({.15*cos(270)},{.15*sin(270)})$) arc (270:225:.15) node[pos=0.7,below] {$45^\circ$};
         \draw ($(A)+({.18*cos(270)},{.18*sin(270)})$) arc (270:315:.18) node[pos=0.7,below] {$45^\circ$};

         \draw ($(B)+(.15,0)$) arc (0:30:.15) node[pos=0.7,right] {$30^\circ$};
         \draw ($(C)+(-.15,0)$) arc (180:150:.15) node[pos=0.7,left] {$30^\circ$};

         \draw ($(P)+({0,.05})$) arc (90:210:.05) node[pos=0.6,left] {$120^\circ$};
         \draw ($(P)+({0,.1})$) arc (90:-30:.1) node[pos=0.6,right] {$120^\circ$};
         \draw ($(P)+({.15*cos(210)},{.15*sin(210)})$) arc (210:330:.15) node[pos=0.5,below] {$120^\circ$};

      \end{tikzpicture}
      
      \caption{
      The minimization problem in eq.~\eqref{eq:fermat:min} is equivalent to
      finding the Fermat point $P$ of the triangle $ABC$, which is right-angled in
      $A$, and has side lengths $|AB|=|AC|=1$ and $|BC|=\sqrt{2}$. The Fermat point
      is characterized by the property that the three angles in $P$ are all equal
      to $120^\circ$.
      } \label{fig:fermat}
   \end{figure}

   This is in contradiction with the hypothesis.
\end{proof}

\subsection{Characterization of gauge fields with \texorpdfstring{$\hat{\epsilon}_0(n_0)=1$}{epsilon\_0(n\_0)=1}}
\label{app:subsec:one-0}

\begin{theorem} \label{theo:one-0}
   If $\hat{\epsilon}_0(n_0) = 1$, then at least one of the following two
   possibilities is realized:
   \begin{enumerate}[topsep=0pt,itemsep=1ex]
      \item Up to a gauge transformation, $n_0$ is localized on a line $\ell$
      and $|n_0(\ell)| = 1$.
      \item Up to a gauge transformation, $n_0$ is localized on a cut-set
      $S=\{\ell_1,\ell_2\}$ which disconnects $a$ and $b$. Assuming that, with
      no loss of generality, $i(\ell_{1,2})$ is connected to $a$ in $\mathcal G - S$,
      Then $n_0(\ell_{1,2}) = -1$.
   \end{enumerate}
\end{theorem}

\begin{proof}
      
   We assume $\hat{\epsilon}_0(n_0)=1$, and we look at all possible solutions of the bound~\eqref{eq:corollary:epshat:bound:0}.
   \begin{enumerate}[topsep=0pt,itemsep=1ex]
      \item $c_0(n|\mathring{T}_0) + p_0^+(n|\mathring{T}_0) =
      p_0^-(n|\mathring{T}_0) = 0$, $p_0^0(n|\mathring{T}_0) \ge 0$. This
      implies that $n_0(\ell)=0$ for every $\ell \in \mathring{T}_0^*$, i.e.~$n$
      is a pure gauge fields and $\hat{\epsilon}_0(n_0)=0$.
      
      \item $c_0(n|\mathring{T}_0) + p_0^+(n|\mathring{T}_0) +
      p_0^-(n|\mathring{T}_0) = 1$, $p_0^0(n|\mathring{T}_0) \ge
      p_0^-(n|\mathring{T}_0)$, which implies that one line $\ell$ exists such
      that $|n_0(\ell)|=1$, and $n_0(\ell')=0$ for $\ell' \neq \ell$, i.e.
      \begin{gather}
         n_0(\star) = n_0(\ell) \delta_{\star,\ell} \ .
      \end{gather}
      Then eq.~\eqref{eq:corollary:epshat:axial:0} implies
      \begin{gather}
         1 = \hat{\epsilon}_0(n_0) = \sum_{\ell' \in \mathring{T}_0^*} |n_0(\ell')|
         =
         |n_0(\ell)|
         \ .
      \end{gather}
      
      \item $p_0^-(n|\mathring{T}_0) = 2$, $c_0(n|\mathring{T}_0) =
      p_0^+(n|\mathring{T}_0) = p_0^0(n|\mathring{T}_0) = 0$, which implies that
      the cutset $S(\mathring{T}_0)$ contains exacly two lines $\ell_1$ and
      $\ell_2$, and
      \begin{gather}
         n_0(\star) = \sum_{j=1,2} n_0(\ell_j) \delta_{\star,\ell_j} \ ,
      \end{gather}
      with $n_0(\ell_j) < 0$. Then eq.~\eqref{eq:corollary:epshat:axial:0}
      implies
      \begin{gather}
         1
         =
         \hat{\epsilon}_0(n_0) = \sum_{\ell' \in \mathring{T}_0^*} |n_0(\ell')|
         - 1
         =
         \sum_{j=1,2} |n_0(\ell_j)| - 1
         \ ,
      \end{gather}
      which is solved only by $|n_0(\ell_j)| = 1$.

   \end{enumerate}
      
\end{proof}

\bigskip

\begin{theorem} \label{theo:t-equivalence}
   The following notions of $t$-equivalence between two lines $\ell_1$ and
   $\ell_2$ are logically equivalent:
   \begin{enumerate}[topsep=0pt,itemsep=1ex]
      \item $\ell_1$ and $\ell_2$ are said to be $t$-equivalent if and only if either
      $\ell_1 = \ell_2$ or $\{\ell_1,\ell_2\}$ is a cut-set which does not
      disconnect $a$ and $b$.
      \item $\ell_1$ and $\ell_2$ are said to be $t$-equivalent if and only if every loop
      containing $\ell_1$ contains also $\ell_2$, and every path from $a$ to $b$
      containing $\ell_1$ contains also $\ell_2$.
      \item $\ell_1$ and $\ell_2$ are said to be $t$-equivalent if and only if every loop
      containing $\ell_1$ contains also $\ell_2$, and a path from $a$ to $b$
      exists which contains both $\ell_1$ and $\ell_2$.
   \end{enumerate}
   The notion of $t$-equivalence is indeed an equivalence relation,
   $t$-equivalence implies $s$-equivalence, and $t$-equivalence classes are
   given by
   \begin{gather}
      [ \ell ]_t =
      \begin{dcases}
         [ \ell ]_s \cap P_1 \qquad & \text{if } \ell \in P_1 \,,
         \\
         [ \ell ]_s \cap P_2 \qquad & \text{if } \ell \in P_2    \,,     
         \\
         [ \ell ]_s \qquad & \text{otherwise} \,,
      \end{dcases}
   \end{gather}
   where $P_1$ and $P_2$ are two disjoint paths from $a$ to $b$.
\end{theorem}

\begin{proof}
   Reflexivity and symmetry of $t$-equivalence are manifest in notion 1, while
   transitivity is manifest in notion 2. The characterization of the
   $s$-equivalence classes follows immediately from notion 2.
   It is also evident that $t$-equivalence implies $s$-equivalence. In
   particular it follows that $[\ell]_t \subseteq [\ell]_s$ for any line $\ell$.
   Since $\mathcal{G}$ is 1PI, two disjoint paths $P_1$ and $P_2$ from $a$ to
   $b$ exist. Let $\ell$ be a line. Since $C = P_1 \sqcup P_2$ is a loop, two
   possibilities are given:
   \begin{enumerate}[topsep=0pt,itemsep=1ex]
      \item $[\ell]_s \cap C = \varnothing$. For any $\ell' \in [\ell]_s$ with
      $\ell' \neq \ell$, $\mathcal{G} - \{\ell,\ell'\}$ contains $P_1$ and $P_2$
      as paths. Therefore $\{\ell,\ell'\}$ is a cut-set which does not
      disconnect $a$ and $b$. Notion 1 implies that $\ell$ and $\ell'$ are
      $t$-equivalent. Therefore $[\ell]_s \subseteq [\ell]_t$, which implies
      $[\ell]_s = [\ell]_t$.
      \item $[\ell]_s \subset C$. Then $\ell$ must be an element of either $P_1$
      or $P_2$. Let us say $\ell \in P_1$. Notion 3 implies that $[\ell]_s \cap
      P_1 \subset [\ell]_t$, while notion 3 implies that $[\ell]_t \subseteq
      [\ell]_s \cap P_1$. It follows that $[\ell]_t = [\ell]_s \cap P_1$.
      Symmetrically, if $\ell \in P_2$ then $[\ell]_t = [\ell]_s \cap P_2$.
   \end{enumerate}
   We are left with the task to prove the logical equivalence of the three
   notions.
   
   \textit{Implication} 1 $\Rightarrow$  2. The claim is trivial for
   $\ell_1=\ell_2$. Assume $\ell_1 \neq \ell_2$. Since $S=\{\ell_1,\ell_2\}$ is
   a cut-set, then $\ell_1$ and $\ell_2$ are $s$-equivalent. By
   theorem~\ref{theo:s-equivalence}, every loop containing $\ell_1$ contains
   also $\ell_2$. We only need to prove that any path from $a$ to $b$ containing
   $\ell_1$ contains also $\ell_2$. Let $P$ a path from $a$ to $b$ containing
   $\ell_1$ and assume that it does not contain $\ell_2$. With no loss of
   generality, we can assume that $[P:\ell_1]=1$. Then $P \setminus  S = P
   \setminus \{ \ell_1 \}$ is the disjoint union of two paths: a path from $a$
   to $i(\ell_1)$, and a path from $f(\ell_1)$ to $b$. Therefore, $a$ belongs to
   the same connected components as $i(\ell_1)$ in $\mathcal{G}-S$, and $b$
   belongs to the same connected components as $f(\ell_1)$ in $\mathcal{G}-S$.
   Since $S$ is a cut-set, the endpoins of $\ell_1$ belong to different
   connected components of $\mathcal{G}-S$. Therefore also $a$ and $b$ belong to
   different connected components of $\mathcal{G}-S$, contradicting the assumption that
   $S$ does not disconnect $a$ and $b$.
   
   \textit{Implication} 2 $\Rightarrow$ 1. Trivial.
   
 \textit{Implication} 3 $\Rightarrow$ 1. If $\ell_1=\ell_2$ the implication
   is trivial. Let us assume $\ell_1 \neq \ell_2$. Since every loop containing
   $\ell_1$ contains also $\ell_2$, then $\ell_1$ and $\ell_2$ are
   $s$-equivalent. By theorem~\ref{theo:s-equivalence}, $S=\{\ell_1,\ell_2\}$ is
   a cut-set. We need to prove that $S$ does not disconnect $a$ and $b$. By
   assumption, a path $P$ from $a$ to $b$ exists which contains $\ell_1$ and
   $\ell_2$. Note that $P  \setminus  S$ is the disjoint union of three paths.
   With no loss of generality, we can assume that $a$ is connected to
   $i(\ell_1)$ in $P  \setminus  S$, $f(\ell_1)$ is connected to $i(\ell_2)$ in
   $P \setminus  S$, and $f(\ell_2)$ is connected to $b$ in $P  \setminus  S$.
   Since $S$ is a cut-set, the endpoins of $\ell_1$ belong to different
   connected components of $\mathcal{G}-S$, so do the endpoints of $\ell_2$. It
   follows that $a$ and $b$ belong to the same connected component of
   $\mathcal{G}-S$.
\end{proof}

\bigskip

\begin{theorem} \label{theo:t-localization}
   The orientation of the lines of $\mathcal{G}$ can be chosen in such a way
   that, if $\ell$ and $\ell'$ are $t$-equivalent and $C$ is a loop that
   contains both, then $[C:\ell]=[C:\ell']$, i.e.~either both $\ell$ and $\ell'$
   have the same orientation as $C$, or they both have the opposite orientation
   of $C$. This choice of orientation is assumed here.
   
   Let $n_0$ and $n'_0$ be gauge fields localized on $\ell$ and $\ell'$
   respectively. $n_0$ and $n'_0$ are gauge equivalent if and only if $\ell$ and
   $\ell'$ are $t$-equivalent and $n_0(\ell)=n'_0(\ell')$.
\end{theorem}

\begin{proof}
   The proof of the first part of the theorem is identical to the proof of the
   corresponding part in theorem~\ref{theo:s-localization}.
   
   In the following we assume that lines are oriented as explained in the first
   part of the theorem. Let $n_0$ and $n'_0$ be gauge fields localized on $\ell$
   and $\ell'$ respectively. Now we prove the two implications of the second
   part separately.
   
   $\bullet \ $ \textit{$n_0$ and $n'_0$ are gauge equivalent $\Rightarrow$
   $\ell$ and $\ell'$ are $t$-equivalent and $n_0(\ell)=n'_0(\ell')$.}

   We first show that any path $P$ from $a$ to $b$ containing $\ell$ contains
   also $\ell'$. Note that
   \begin{gather}
      [P:\ell'] n'_0(\ell') = n'_0(P) = n_0(P) = [P:\ell] n_0(\ell) \neq 0 \ ,
   \end{gather}
   where we have used the fact that $n'_0$ is localized on $\ell'$ in the first
   equality, gauge invariance of the Wilson line in the second equality, the
   fact that $n_0$ is localized on $\ell$ in the third equality, and finally the
   assumption that $\ell \in P$ in the inequality. It follows that $[P:\ell']
   \neq 0$, i.e.~$\ell' \in P$.
   
  Next we need to show that any loop $C$ containing $\ell$ contains also
   $\ell'$. This is done analogously observing that
   \begin{gather}
      [C:\ell'] n'_0(\ell') = n'_0(C) = n_0(C) = [C:\ell] n_0(\ell) \neq 0 \,,
      \label{theo:t-localization:eq0}
   \end{gather}
   implies $[C:\ell'] \neq 0$, i.e.~$\ell' \in C$. 
   
   Because of the particular
   choice of orientation for the lines of the equivalence class, we also have $[C :
   \ell]=[C : \ell']$. Together with eq.~\eqref{theo:t-localization:eq0}, this
   implies immediately $n_0(\ell)=n'_0(\ell')$.
   
   $\bullet \ $ \textit{$\ell$ and $\ell'$ are $t$-equivalent and
   $n_0(\ell)=n'_0(\ell')$ $\Rightarrow$ $n_0$ and $n'_0$ are gauge equivalent.}

   Consider a loop $C$. Because of $t$-equivalence, two
   possibilities are given: either $C$ contains both $\ell$ and $\ell'$, or it
   contains neither of the two.
   If $C$ contains both $\ell$ and $\ell'$,
   \begin{gather}
      n'_0(C) = [C:\ell'] n'_0(\ell') = [C:\ell] n_0(\ell) = n_0(C) \ ,
   \end{gather}
   where we have used the fact that $n'_0$ is localized on $\ell'$ in the first
   equality, the fact that $[ {C}:\ell']=[ {C}:\ell]$ (choice of
   orientation) and $n'_0(\ell')=n_0(\ell)$ (hypothesis) in the second equality,
   and finally the fact that $n_0$ is localized on $\ell$ in the third equality.
   On the other hand, if $C$ contains neither $\ell$ nor $\ell'$, by hypothesis
   of localization,
   \begin{gather}
      n'_0(C) = 0 = n_0(C) \ .
   \end{gather}
   Therefore, for all loops $C$, $n'_0(C)=n_0(C)$.
   Since $\mathcal{G}$ is 1PI, two disjoint paths $P_1$ and $P_2$ from $a$ to
   $b$ exist. Because of $t$-equivalence, at least one of these two path (let us
   say $P_1$) contains neither $\ell$ nor $\ell'$. By hypothesis of
   localization,
   \begin{gather}
      n'_0(P_1) = 0 = n_0(P_1) \ .
   \end{gather}
   Gauge equivalence of $n_0$ and $n'_0$ follows from
   theorem~\ref{theo:gauge-equivalence-0} (with the observation that it is
   always possible to find a tree that contains $P_1$).
\end{proof}

\bigskip

\begin{theorem}\label{theo:simple-0b}
   Let $n_0^1$ be a gauge field localized on $\ell_1$ and let $n_0^2$ be a gauge
   field localized on a cut-set $S = \{\ell_2,\ell'_2\}$ which disconnects $a$
   and $b$. Then $n_0^1$ and $n_0^2$ are not gauge equivalent.
\end{theorem}

\begin{proof}
   Since $\mathcal{G}$ is 1PI, two disjoint paths $P$ and $P'$ from $a$ to $b$
   exist. Since $S$ disconnects $a$ and $b$, it intesects both paths. With no
   loss of generality we can assume $\ell_2 \in P$ and $\ell'_2 \in P'$. Also,
   with no loss of generality, we can assume that $\ell_1 \not\in P$. Then
   \begin{gather}
      n^1_0(P) = [P : \ell_1] n^1_0(\ell_1) = 0
      \ , \\
      n^2_0(P) = [P : \ell_2] n^2_0(\ell_2) \neq 0
      \ .
   \end{gather}
   Since the Wilson line is gauge invariant, $n^1_0$ and $n^2_0$ are not gauge
   equivalent.
\end{proof}

\bigskip

\begin{theorem}\label{theo:simple-0c}
   Let $n_0^1$ and $n_0^2$ be gauge fields localized on
   $S_1=\{\ell_{11},\ell_{12}\}$ and $S_2=\{\ell_{21},\ell_{22}\}$ respectively,
   where $S_1$ and $S_2$ are cut-sets which disconnect $a$ and $b$. With no loss
   of generality we assume that $i(\ell_{jk})$ is connected with $a$ in
   $\mathcal{G}-S_j$, for $j,k=1,2$. Assume that
   $n_0^j(\ell_{j1})=n_0^j(\ell_{j2})=\nu_j$ for $j=1,2$. $n_0^1$ and $n_0^2$
   are gauge equivalent if and only if $\nu_1=\nu_2$.
\end{theorem}

\begin{proof}
   Let $C$ be any loop. Since $\ell_{11}$ and $\ell_{12}$ are $s$-equivalent,
   (theorem \ref{theo:s-equivalence}) two possibilities are given: either $C
   \cap S_1 = \varnothing$ or $S_1 \subseteq C$. If $C \cap S_1 = \varnothing$,
   since $n^1_0$ is localized on $S_1$, trivially $n^1_0(C)=0$. If $S_1 \subseteq C$,
   \begin{gather}
      n^1_0(C) = [C:\ell_{11}] n_0^1(\ell_{11}) + [C:\ell_{12}] n_0^1(\ell_{12}) = [C:\ell_{11}] \nu_1 + [C:\ell_{12}] \nu_1 \ .
   \end{gather}
   Since $S_1$ is a cut-set, and given the orientation of $\ell_{1j}$, it is
   clear that $C$ must have the same direction as one of the two lines and the
   opposite direction of the other (this is also a consequence of theorem 2-14
   of \cite{nakanishi1971graph}), i.e.~$[C:\ell_{11}] = - [C:\ell_{12}]$. It
   follows that $n^1_0(C) = 0$.
   The same argument can be repeated symmetrically for $n^2_0$ yielding
   \begin{gather}
      n^1_0(C) = 0 = n^2_0(C) \ ,
   \end{gather}
   for any loop $C$.
   Since $\mathcal{G}$ is 1PI, two disjoint paths $P_1$ and $P_2$ from $a$ to
   $b$ exist. Since $S_j$ disconnects $a$ and $b$, it intesects both paths. With
   no loss of generality we can assume $\ell_{j1} \in P_1$ and $\ell_{j2} \in
   P_2$. Also, given the orientation of $\ell_{1j}$, it is clear that
   $[P_k:\ell_{jk}]=1$. It follows that
   \begin{gather}
      n^j_0(P_k) = [P_k:\ell_{jk}] \nu_j = \nu_j \ ,
   \end{gather}
   for $j,k=1,2$.
   By theorem \ref{theo:gauge-equivalence-0} (with the observation that it is
   always possible to find a tree that contains $P_1$), $n^1_0$ and $n^2_0$ are
   gauge equivalent if and only if $n^1_0(P_1) = n^2_0(P_1)$, i.e.~$\nu_1 =
   \nu_2$.
\end{proof}

\section{Analitycity of 1PI vertices}
\label{app:analyticity}

\begin{theorem} \label{theo:ana1}
   To all order in perturbation theory, the 1PI proper vertices
   \begin{gather}
      \Gamma^{\pi\gamma\pi}_{q \mu \bar{q}}(p-\tfrac{k}{2},k,-p-\tfrac{k}{2})
      \ , \qquad
      \Gamma^{\pi\gamma\gamma\pi}_{q \mu_1 \mu_2 \bar{q}}(p,k,-k,-p)
      \ ,
   \end{gather}
   which are initially defined for $(p,k) \in \mathbb{R}^4 \times \mathbb{R}^4$,
   analytically extend to the domain
   \begin{gather}
      \{ (p,k) \in \mathbb{C}^4 \times \mathbb{C}^4 \ | \ (\Im p \pm \Im k)^2 < 4m^2 \}
      \ .
   \end{gather}
   In particular, if $k$ and $\vec{p}$ are real, the 1PI proper vertices are
   analytic in the strip $|\! \Im p_0| < 2m$.
\end{theorem}

\begin{proof}
   Modulo a trivial mapping of notation, this theorem is nothing but theorem 2.3
   in~\cite{\MartinStable}.
\end{proof}

\bigskip

\begin{theorem} \label{theo:ana2}
   Define $\bar{p} = ( i E(\vec{p}) , \vec{p} )$ and $E(\vec{p}) = \sqrt{m^2 +
   \vec{p}^2}$. Let $\vec{p}$ be a real vector satisfying $\vec{p}^2 <
   3m^2$, and let $\vec{k}$ be a generic real vector. To all orders in
   perturbation theory, the 1PI proper vertex
   \begin{gather}
      \Gamma^{\pi\gamma\pi}_{q \mu \bar{q}}(\bar{p},k,-\bar{p}-k)
      \ ,
   \end{gather}
   is an analytic function of $k_0$ in the strip
   \begin{gather}
      -m < \Im k_0 < 2m - E(\vec{p}) \ .
   \end{gather}
\end{theorem}

\begin{proof}
   Let $\mathcal{G}$ be the abstract graph associated to a generic Feynman diagram
   contributing to the proper vertex function. Let $v$ and $w$ be the vertices
   associated to the insertions of the external pion fields, and let $a$ be the vertex
   associated to the insertion of the electromagnetic current. Since
   $\mathcal{G}$ is 1PI, two disjoint paths $P_1$ and $P_2$ from $v$ to $w$
   exist. One can also easily construct a path $P_\gamma$ from $a$ to a vertex
   $z$ with the properties that $P_\gamma \cap (P_1 \cup P_2) = \varnothing$ and
   $z$ is an endpoint of one of the lines in $P_1 \cup P_2$. Note that
   $P_\gamma$ may be empty if $a$ itself is an endpoint of one of the lines in
   $P_1 \cup P_2$. In this case, we set $z=a$. With no loss of generality we can
   assume that $z$ is an endpoint of one of the lines in $P_1$, then $P_1$ is
   split in the disjoint union of two paths: $P_{vz}$ from $v$ to $z$, and
   $P_{zw}$ from $z$ to $w$. Note that one of these two paths may be empty, if
   $z$ coincides with either $v$ or $w$.
   
   We can use the paths constructed above to define a flow $Q(\ell)$ of the
   external momenta through the graph $\mathcal{G}$ as follows
   \begin{alignat}{3}
      Q(\ell) = \, \, & \alpha \bar{p}  &\hspace{2cm}& \text{if } \ell \in P_{vz} \ , \\
      Q(\ell) = \, \, & \alpha \bar{p} + k          && \text{if } \ell \in P_{zw} \ , \\
      Q(\ell) = \, \, & (1-\alpha) \bar{p}          && \text{if } \ell \in P_2 \ , \\
      Q(\ell) = \, \, & k                           && \text{if } \ell \in P_\gamma \ , \\
      Q(\ell) = \, \, & 0                           && \text{otherwise} \ ,
   \end{alignat}
   where $\alpha \in [0,1]$ is an adjustable parameter.
   
   The inverse propagator associated to the line $\ell$ is
   \begin{gather}
      [ Q(\ell)+q(\ell) ]^2 + m^2 = 2 i [ \Re Q(\ell)+q(\ell) ] \Im Q(\ell) + [ \Re Q(\ell)+q(\ell) ]^2 - [ \Im Q(\ell) ]^2 + m^2 \ ,
   \end{gather}
   where $q(\ell)$ is a real loop momentum. The condition $[ \Im Q(\ell) ]^2 <
   m^2$ for every line $\ell$ is sufficient to guarantee regularity of the
   Feynman integrand, hence analyticity of the Feynman integral. With the momentum routing that we have chosen, this condition
   is equivalent to a pair of  inequalities
   \begin{gather}
      1 - \frac{m}{E(\vec{p})} < \alpha < \frac{m}{E(\vec{p})} \ , \qquad \qquad %
      -m < \Im k_0 < m - \alpha E(\vec{p}) \ .
   \end{gather}
The first inequality admits a solution because, in the kinematic regime of interest,
   \begin{gather}
      \frac{1}{2} < \frac{m}{E(\vec{p})} \le 1 \ .
   \end{gather}
To give an explicit construction we choose a small $\epsilon>0$, and set
   \begin{gather}
      \alpha = 1 - \frac{m}{E(\vec{p})} + \frac{\epsilon}{E(\vec{p})} \ .
   \end{gather}
   The Feynman integral is thus analytic for
   \begin{gather}
      -m < \Im k_0 < 2m - E(\vec{p}) - \epsilon \ .
   \end{gather}
\end{proof}

\section{Pole and regular parts of the Compton scattering amplitude}
\label{sec:CompDecom}
In order to estimate the finite-$L$ corrections, we need to provide an expression
for the Compton scattering amplitude in the space-like region. A useful step in this direction is to decompose the amplitude into a pole and regular piece as we have done in eq.~\eqref{eq:estL:Compton}. 
To derive this relation we begin by substituting the
identiy
\begin{gather}
   \mathcal{J}_\rho(x) = e^{i P_\mu x^\mu} \mathcal{J}_\rho(0) e^{-i P_\mu x^\mu} \ ,
\end{gather}
into
eq.~\eqref{eq:stat:Compton}. Setting the Minkowski four-vectors to $p=(E(p_3),p_3 \vec{e}_3)$ and $k=(0,k_3 \vec{e}_3)$  (with $E(p) = \sqrt{m^2+p^2}$), 
and  integrating over $x$, one obtains the
representation
\begin{multline}
   T(-k_3^2,-p_3k_3)
   =
   \lim_{p'_3 \to p_3} \sum_{q=0,\pm 1}
   \langle p'_3 \vec{e}_3, q | \mathcal{J}_\rho(0) \frac{(2\pi)^3 \delta(P_1) \delta(P_2) \delta(P_3-p_3-k_3)}{H-\sqrt{m^2+p_3^2}-i\epsilon} \mathcal{J}^\rho(0) | p_3 \vec{e}_3, q \rangle
   \\+ (k_3 \to -k_3)
   \ .
   \label{eq:stat:Compton-Hamiltonian}
\end{multline}

The one-pion and multi-pion contributions to the Compton scattering amplitude
can then be separated by inserting the identity between the two currents in the form
\begin{gather}
   1 = | \Omega \rangle \langle \Omega | + \sum_{q=0,\pm 1} \int \frac{d^3 p}{(2\pi)^3 2 E(\vec{p}) } | \vec{p},q \rangle \langle \vec{p},q | + \theta(M-2m) \ ,
\end{gather}
where $M = (P_\mu P^\mu)^{1/2}$ is the mass operator, and we have used the fact
that $M$ has two discrete eigenvalues corresponding to the vacuum and the
one-pion states, and then a gap up to the two-pion threshold. 

The one-pion
matrix elements of the electromagnetic current are written in terms of the
electromagnetic form factor $F(Q^2)$ of the pion, defined in eq.~\eqref{eq:pionFF}.
Some lengthy but straightforward algebra yields the expression
\begin{gather}
   T^\text{1P}(k_3,p_3)
   =
   \frac{\mathcal{F}(k_3,p_3)}{ k_3^2 + 2k_3p_3 - i\epsilon }
   + \frac{\mathcal{F}(-k_3,p_3)}{ k_3^2 - 2k_3p_3 - i\epsilon }
   \ ,
\end{gather}
with the definition
\begin{multline}
   \mathcal{F}(k_3,p_3)
   =
   \left[ 2m^2 + 2 E(p_3) E(p_3+k_3) - 2 p_3 (p_3+k_3) \right]
   \left[ E(p_3) + E(p_3+k_3) \right] 
 \times \\ \times
  [E(p_3+k_3)]^{-1}  \ F \!\left[ 2m^2 - 2 E(p_3) E(p_3+k_3) + 2 p_3 (p_3+k_3) \right]^2
   \ .
\end{multline}
 One readily proves that the argument of the form factor satisfies
\begin{gather}
   2m^2 - 2 E(p_3) E(p_3+k_3) + 2 p_3 (p_3+k_3) \le - k_3^2 \le 0 \ .
\end{gather}
Since the form factor $F(Q^2)$ is analytic for $Q^2 < (2m)^2$ \cite{\BD,Barton:1965}, the function $\mathcal{F}(k_3,p_3)$ is infinitely
differentiable in both its variables. Using the fact that
$\mathcal{F}(0,p_3)=8m^2$ does not depend on $p_3$, one can decompose
\begin{gather}
   \mathcal{F}(k_3,p_3) = \mathcal{F}(k_3,-\tfrac{k_3}{2}) + 2 k_3  \big ( p_3 + \tfrac{k_3}{2} \big) \mathcal{G}(k_3,p_3)
   \ ,
\end{gather}
where $\mathcal{G}$ is infinitely differentiable in both its variables, and
\begin{gather}
   \mathcal{F}(k_3,-\tfrac{k_3}{2})
   =
   2(4m^2 + k_3^2) \, F(-k_3^2)^2
   \ .
\end{gather}
Therefore, the one-pion contribution to the Compton scattering amplitude has the
following representation
\begin{gather}
   T^\text{1P}(k_3,p_3)
   =
   \frac{2(4m^2 + k_3^2) \, F(-k_3^2)^2}{ k_3^2 + 2k_3p_3 - i\epsilon }
   + \frac{2(4m^2 + k_3^2) \, F(-k_3^2)^2}{ k_3^2 - 2k_3p_3 - i\epsilon }
   + \mathcal{G}(k_3,p_3) + \mathcal{G}(-k_3,p_3)
   \ .
\end{gather}

The multi-pion contribution $T^\text{MP}(-k_3^2,-p_3k_3)$ can be seen to be
analytic in $p_3$ and $k_3$ as long as $p_3^2 < 3m^2$. Also note that the
integral in eq.~\eqref{eq:stat:DeltaGs} can be restricted to $p_3^2 < 3m^2$ up
to an error of order $e^{-2mL}$ that we are already neglecting. Therefore only
the analytic region contributes to the leading exponentials. It is convenient to
separate the pole and regular parts, yielding the decomposition
\begin{gather}
   T(-k_3^2,-p_3k_3) = T^\text{pole}(-k_3^2,-p_3k_3) + T^\text{reg}(-k_3^2,-p_3k_3)
   \label{eq:stat:Compton-dec}
   \ , \\
   T^\text{pole}(-k_3^2,-p_3k_3)
   =
   \frac{2(4m^2 + k_3^2) \, F(-k_3^2)^2}{ k_3^2 + 2k_3p_3 - i\epsilon }
   + \frac{2(4m^2 + k_3^2) \, F(-k_3^2)^2}{ k_3^2 - 2k_3p_3 - i\epsilon }
   \label{eq:stat:Compton-pole}
   \ , \\
   T^\text{reg}(-k_3^2,-p_3k_3)
   =
   \mathcal{G}(k_3,p_3) + \mathcal{G}(-k_3,p_3) + T^\text{MP}(-k_3^2,-p_3k_3)
   \label{eq:stat:Compton-reg}
   \ .
\end{gather}
This completes the demonstration of eq.~\eqref{eq:estL:Compton} and gives alternative definitions to the quantities appearing in that equation.

Finally we look at the integral over $\vec{p}$ in eq.~\eqref{eq:stat:DeltaGs}, defining
\begin{gather}
   \mathcal{T}(k_3^2 | L )
   =
   \int \frac{d p_3}{2\pi}
   e^{-L \sqrt{m^2+p_3^2}}
   \Re T(-k_3^2,-p_3k_3)
   \ ,
\end{gather}
and the corresponding quantities $\mathcal{T}_\text{pole}$ and
$\mathcal{T}_\text{reg}$, obtained by substituting the
decomposition~\eqref{eq:stat:Compton-dec}. The pole part is conveniently written
as
\begin{gather}
   \mathcal{T}_\text{pole}(k_3^2|L)
   =
   2(4m^2 + k_3^2) \, F(-k_3^2)^2 \, \zeta(k_3^2|L)
   \ , \\
   \zeta(k_3^2|L)
   =
   \int \frac{d p_3}{2\pi}
   \frac{
   e^{-L \sqrt{m^2+(p_3-\frac{k_3}{2})^2}}
   - e^{- L \sqrt{m^2+(p_3+\frac{k_3}{2})^2}}
   }{ 2k_3p_3 }
   \ .
   \label{eq:zetadef}
\end{gather}
An alternative form of $ \zeta(k_3^2|L)$ is given in eq.~\eqref{eq:estL:zeta}.

\bibliographystyle{JHEP}      
\bibliography{refs.bib}

\end{document}